\documentclass[]{IEEEtran}
\usepackage[round]{natbib}
\usepackage{multicol}
\usepackage[bookmarks=false]{hyperref}
\hypersetup{
    a4paper=true
    colorlinks=true,
    citecolor=black,
    filecolor=black,
    linkcolor=black,
    urlcolor=black,
    pdfauthor={author},
    pdfsubject={subject},
    pdftitle={title},
    bookmarksopen=false,
    plainpages=false,
    hidelinks
}
\usepackage{graphicx}
\usepackage{latexsym}
\usepackage{amsmath,amssymb,subfigure}
\usepackage{amsthm}
\usepackage{epstopdf}
\usepackage[ruled, vlined]{algorithm2e}
\usepackage{Cris}
\usepackage{color}
\usepackage{varwidth}
\usepackage{bbm}
\usepackage{empheq}
\usepackage{cases}

\usepackage{flushend}

\usepackage{tikz}
\usetikzlibrary{calc,decorations.pathreplacing}

\newcommand{\tikzmark}[1]{%
  \tikz[overlay,remember picture,baseline] \node [anchor=base] (#1) {};}

\newcommand{\DrawVerticalBrace}[3][]{%
    \begin{tikzpicture}[overlay,remember picture]
        \draw[decorate,decoration={brace,amplitude=1ex}, #1] 
            ($(#3)+(0.0em,-0.5ex)$) to
            ($(#2)+(0.0em,+1.7ex)$);%
    \end{tikzpicture}%
}

\newcommand{\mc}[1]{\mathcal{#1}}

\newcommand{\bea}{\begin{eqnarray}}
\newcommand{\eea}{\end{eqnarray}}
\newcommand{\beas}{\begin{eqnarray*}}
\newcommand{\eeas}{\end{eqnarray*}}
\newcommand{\leftm}{\left[\begin{array}}
\newcommand{\rightm}{\end{array}\right]}
\newcommand{\reals}{\mbox{$\mathbb R$}}
\newcommand{\ones}{\mathbbm{1}}

\def\diag{{\mbox{\bf diag}}}
\def\rk{{\mbox{\bf rk}}}

\newtheorem{thm}{Theorem}[section]
\newtheorem{cor}[thm]{Corollary}
\newtheorem{prop}[thm]{Proposition}
\newtheorem{lem}[thm]{Lemma}
\newtheorem{definition}[thm]{Definition}
\newtheorem{rem}[thm]{Remark}

\newtheorem{assumption}[thm]{Assumption}

\def\Vlambda{{\includegraphics[width=0.5\columnwidth]{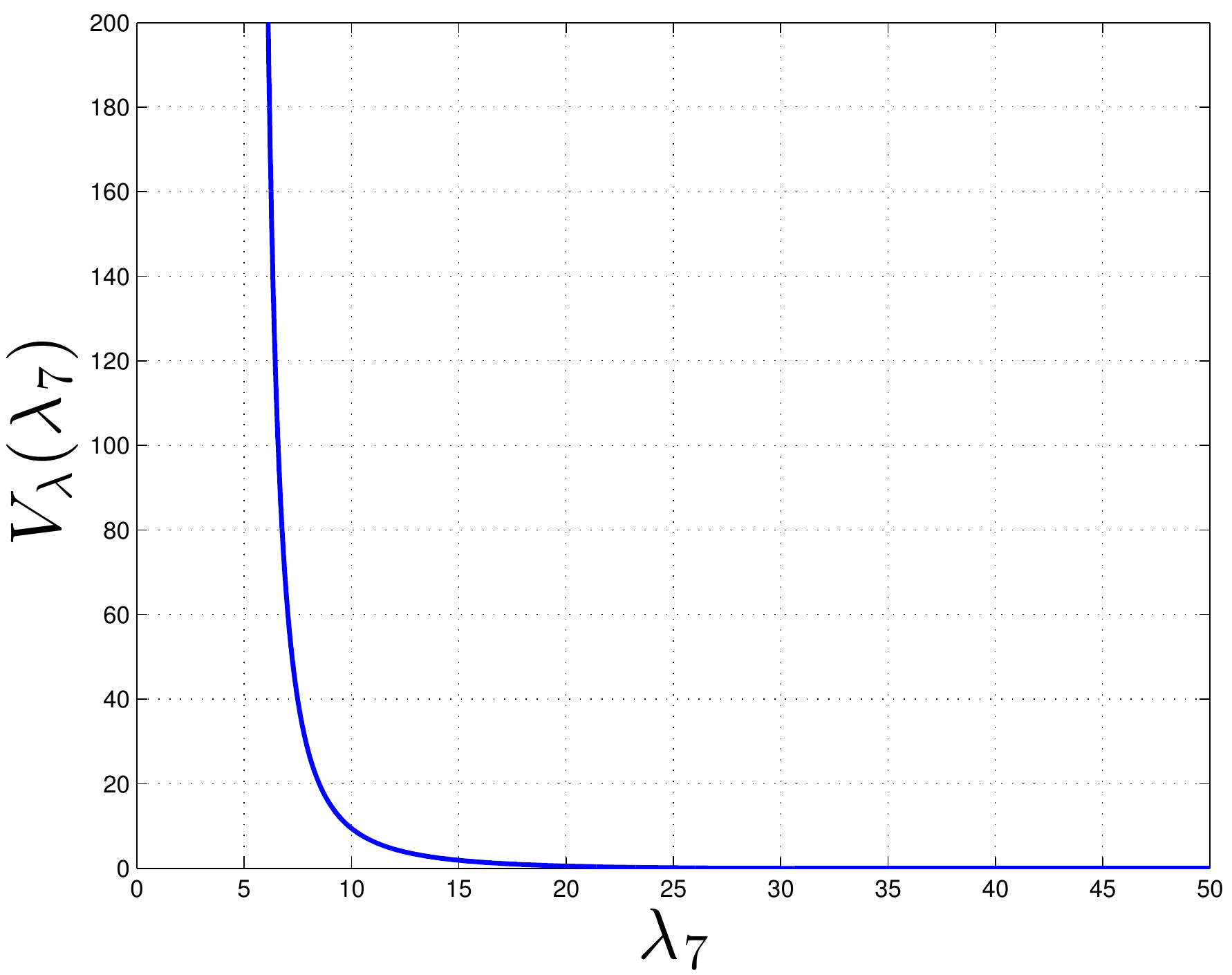}}}
\def\gammaa{{\includegraphics[scale=0.17]{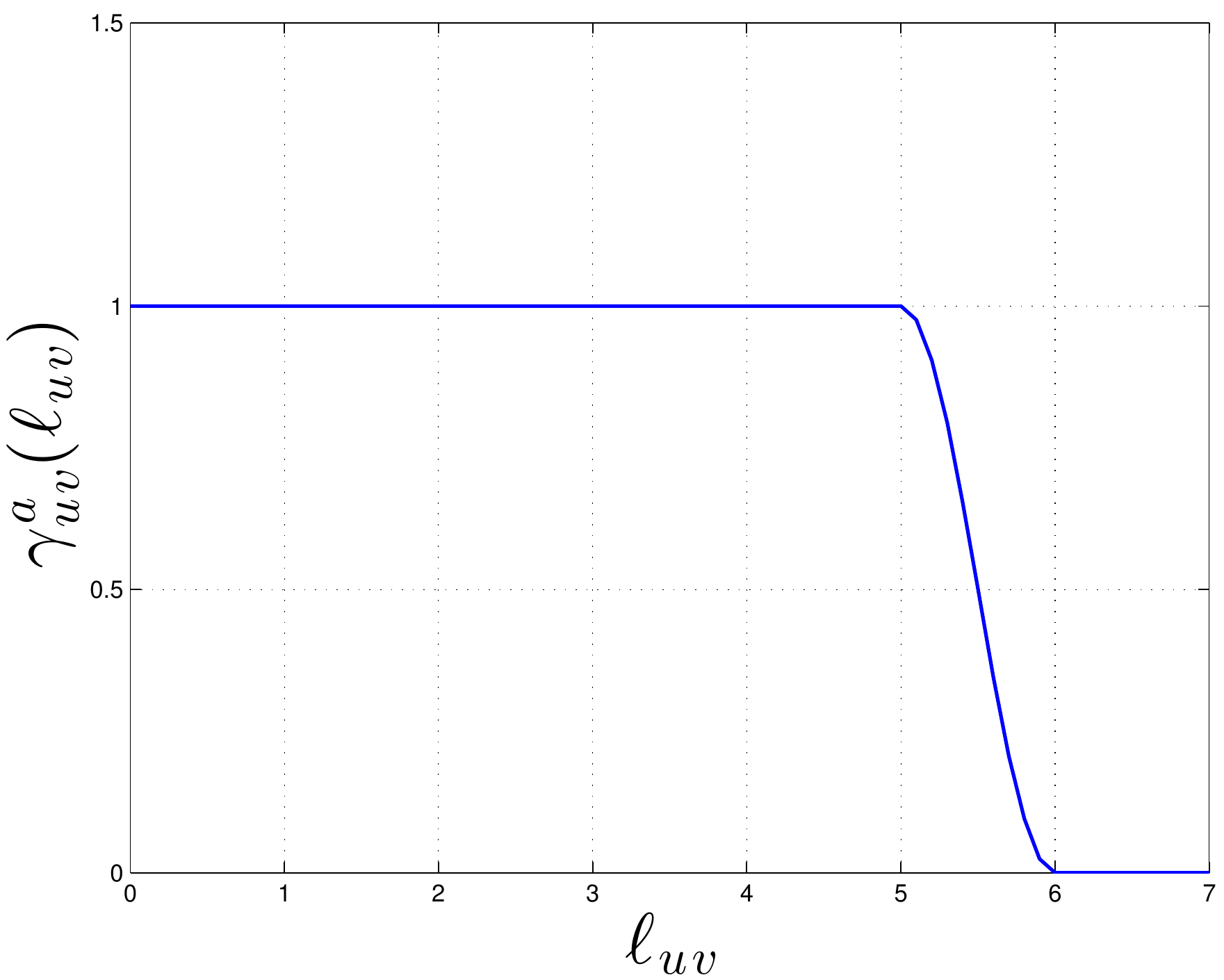}}}
\def\gammab{{\includegraphics[scale=0.17]{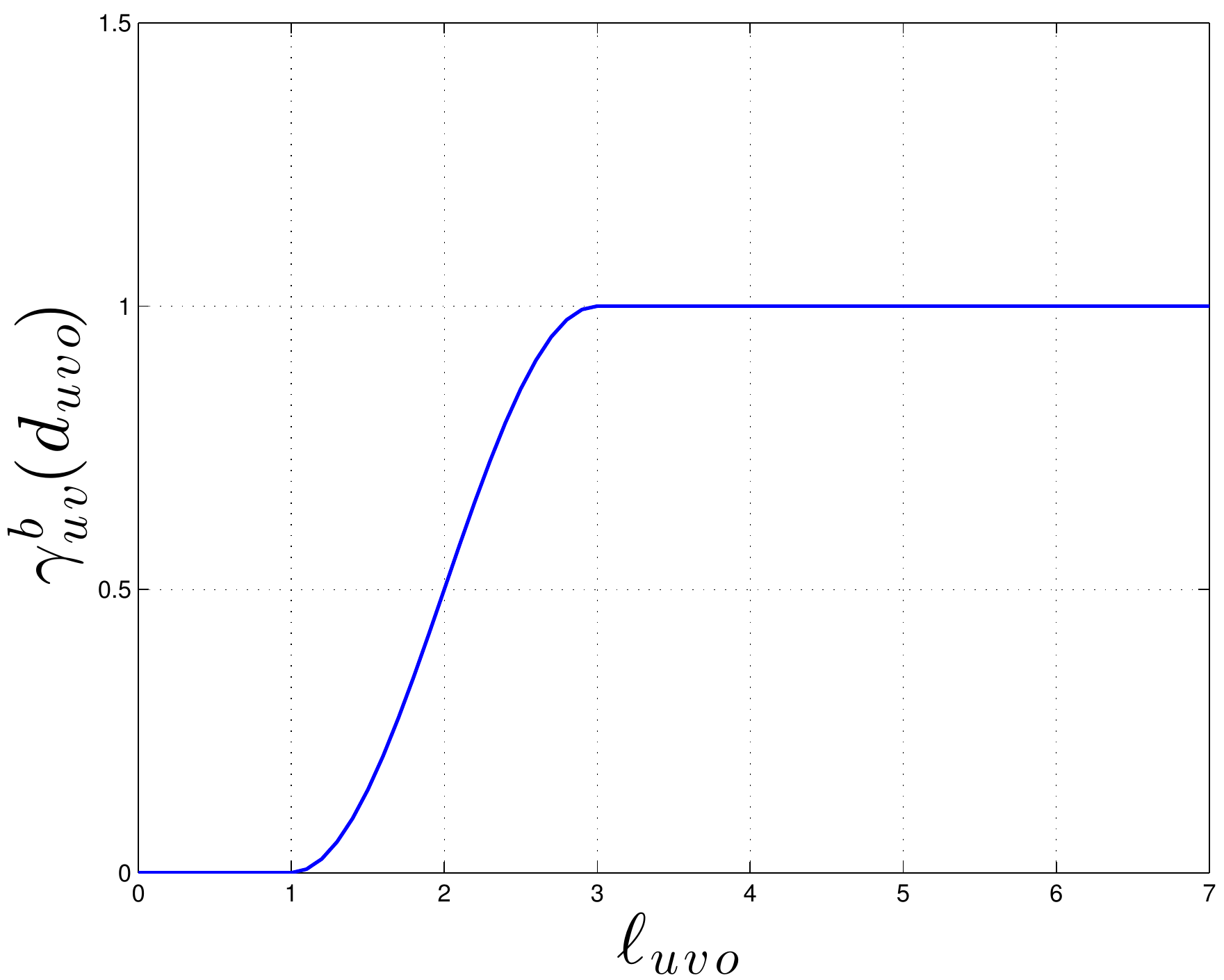}}}
\def\betaf{{\includegraphics[scale=0.17]{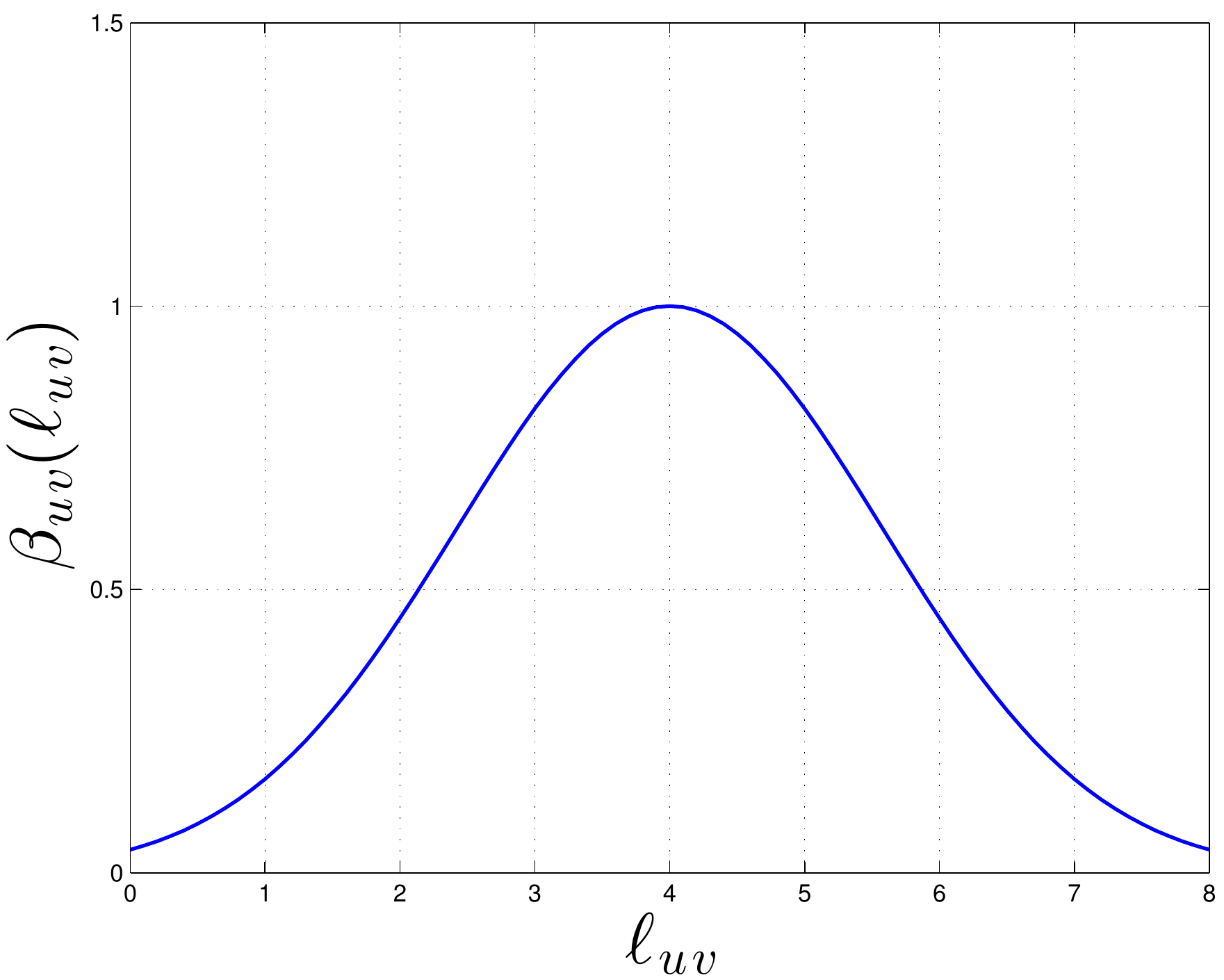}}}

\graphicspath{{final_figures/}}

\begin{document}

\title{Decentralized Rigidity Maintenance Control with Range Measurements for Multi-Robot Systems}

\author{Daniel Zelazo, Antonio Franchi, Heinrich H. B\"ulthoff, and Paolo Robuffo Giordano
\thanks{D.~Zelazo is with the Faculty of Aerospace Engineering, Technion - Israel Institute of Technology, Haifa 32000, Israel {\tt \scriptsize dzelazo@technion.ac.il}}%
\thanks{A~Franchi is with the Centre National de la Recherche Scientifique (CNRS),
Laboratoire d'Analyse et d'Architecture des Syst\`emes (LAAS),
7 Avenue du Colonel Roche, 31077 Toulouse CEDEX 4, France. {\tt \href{mailto:antonio.franchi@laas.fr}{antonio.franchi@laas.fr}}}
\thanks{H.~H.~B\"ulthoff  is with the Max Planck Institute for Biological Cybernetics, Spemannstra\ss{}e 38, 72076 T\"ubingen, Germany {\tt \scriptsize \href{mailto:hhb@tuebingen.mpg.de}{hhb}@tuebingen.mpg.de}. H.~H.~B\"ulthoff is additionally with the Department of Brain and Cognitive Engineering, Korea University,  Seoul, 136-713 Korea.}
\thanks{P.~Robuffo~Giordano is with the CNRS at Irisa and Inria Rennes Bretagne Atlantique, Campus de Beaulieu, 35042 Rennes Cedex, France {\tt \scriptsize  prg@irisa.fr}.}}

\markboth{}{}

\maketitle

\begin{abstract}
This work proposes a fully decentralized strategy for maintaining the formation rigidity of a multi-robot system using only range measurements, while still allowing the graph topology to change freely over time. 
In this direction, a first contribution of this work is an extension of rigidity theory to \emph{weighted frameworks} and the \emph{rigidity eigenvalue}, which when positive ensures the infinitesimal rigidity of the framework.  We then propose a distributed algorithm for estimating a common relative position reference frame amongst a team of robots with only range measurements {in addition to one agent endowed with the capability of measuring the bearing to two other agents}.   This first estimation step is embedded into a subsequent distributed algorithm for estimating the rigidity eigenvalue associated with the weighted framework.  The estimate of the rigidity eigenvalue is finally used to generate a local control action for each agent that both maintains the rigidity property and enforces additional constraints such as collision avoidance and sensing/communication range limits and occlusions.
As an additional feature of our approach, the communication and sensing links among the robots are also left free to change over time while preserving rigidity of the whole framework. The proposed scheme is then experimentally validated with a robotic testbed consisting of $6$ quadrotor UAVs operating in a cluttered environment.
\end{abstract}

\begin{IEEEkeywords}
graph rigidity, decentralized control, multi-robot, distributed algorithms, distributed estimation.
\end{IEEEkeywords}

\IEEEpeerreviewmaketitle

\section{Introduction}

The coordinated and decentralized control of multi-robot systems is an enabling technology for a variety of applications.  Multi-robot systems benefit from an increased robustness against system failures due to their ability to adapt to dynamic and uncertain environments.  There are also numerous economic benefits by considering the price of small and cost-effective autonomous systems as opposed to their more expensive monolithic counterparts.  Currently, there is a great interest in implementing these systems from deep space interferometry missions and distributed sensing and data collection, to civilian search and rescue operations, among others ~\citep{2002-AkySanCay, 2008-AndFidYu_Van, Bristow2000, 2011-LinMelKum,2010-MesEge, 2009-MicFinKum, 2006-Mur}.

The challenges associated with the design and implementation of multi-agent systems range from hardware and software considerations to the development of a solid theoretical foundation for their operation.  In particular, the \emph{sensing} and \emph{communication} capabilities of each agent will dictate the distributed protocols used to achieve team objectives.  For example, if each agent in a multi-robot system is equipped with a GPS-like sensor, then tasks such as formation keeping or localization can be trivially accomplished by communication between robots of their state information in a common world-frame.  However, in applications operating in harsher environments, i.e., indoors, underwater, or in deep-space, GPS is not a viable sensing option~\citep{2013-ScaAchDoiFraKosMarAchChlChaKneGurHenLeeLynMeiPolRenSieStuTanTroWei}.  Indeed, in these situations, agents must rely on sensing without knowledge of a common inertial reference frame~\citep{2012q-FraMasGraRylBueRob}.  In these scenarios, \emph{relative sensing} can provide accurate measurements of, for example, range or bearing, but without any common reference frame. 

A further challenge related to the sensing capabilities of multi-robot systems is the availability of these measurements.  Sensing constraints such as line-of-sight requirements, range, and power limitations introduce an important system-level requirement, and also lead to an inherently time-varying description of the sensing network.  Successful decentralized coordination protocols, therefore, must also be able to manage these constraints.

These issues lead to important architectural requirements for the sensing and communication topology in order to achieve the desired higher level tasks (i.e., formation keeping or localization).   The \emph{connectivity} of the sensing and communication topology is one such property that has received considerable attention in the multi-robot communities \citep{2011d-RobFraSecBue, RoFrSeBu:13, Ji2007}.  However, connectivity alone is not sufficient to perform certain tasks when only relative sensing is used.   For these systems, the concept of \emph{rigidity} provides the correct framework for defining an appropriate sensing and communication topology architecture.   Rigidity is a combinatorial theory for characterizing the ``stiffness" or ``flexibility" of structures formed by rigid bodies connected by flexible linkages or hinges.

The study of rigidity has a rich history with contributions from mathematics and engineering disciplines \citep{Connelly2009, Jacobs1997, 2009-KriBroFra, 1970-Lam, 2009-ShaFidAnd, 1985-TayWhi, Eren2004}.  Recently, rigidity theory has taken an outstanding role in the motion control of mobile robots.  The rigidity framework allows for applications, such as formation control, to employ control algorithms relying on only \emph{relative distance measurements}, as opposed to \emph{relative position measurements} from a global or relative internal frame~\citep{2008-AndFidYu_Van, 2008-AndYu_FidHen, Baillieul2007, 2009-KriBroFra, 2002-OlfMur,smith2007automatic}.   For example, in \citep{2009-KriBroFra} it was shown that formation stabilization using only distance measurements can be achieved only if rigidity of the formation is maintained.  
Moreover, rigidity represents also a necessary condition for estimating relative positions using only relative distance measurements \citep{2006-AspEreGolMorWhiYanAndBel,2010-CalCarWei}.

In a broader context, rigidity turns out to be an important architectural property of many multi-agent systems when a common inertial reference frame is unavailable.  Applications that rely on sensor fusion for localization, exploration, mapping and cooperative tracking of a target, all can benefit from notions in rigidity theory~\citep{2009-ShaFidAnd, 2006-AspEreGolMorWhiYanAndBel, Calafiore2010, Williams2014, Wu2010}.  
The concept of rigidity, therefore, provides the theoretical foundation for approaching \emph{decentralized} solutions to the aforementioned problems using \emph{distance measurement sensors}, and thus establishing an appropriate framework for relating system level architectural requirements to the sensing and communication capabilities of the system.

\subsection{Main Contributions} \label{contributions}

In general, rigidity as a property of a given formation {(i.e., of the robot spatial arrangement)} has been studied from either a purely combinatorial perspective~\citep{1970-Lam}, or by providing an algebraic characterization via the state-dependent \emph{rigidity matrix}~\citep{1985-TayWhi}.  In our previous work \citep{2012-ZelFraRob}, we introduced a 
related matrix termed the \emph{symmetric rigidity matrix}.  A main result of \citep{2012-ZelFraRob} was to provide a necessary and sufficient conditions for rigidity in the plane in terms of the positivity of a particular eigenvalue of the symmetric rigidity matrix; this eigenvalue we term the \emph{rigidity eigenvalue}.
This result is in the same spirit as the celebrated Fiedler eigenvalue\footnote{The second smallest eigenvalue of the graph Laplacian matrix.} and its relation to the connectivity of a graph \citep{2001-GodRoy}.   A first contribution of this work is the extension of the results on the rigidity eigenvalue provided in \citep{2012-ZelFraRob} to 3-dimensional frameworks, as well as the 
introduction of the concept of \emph{weighted rigidity} and the corresponding \emph{weighted rigidity matrix}. This notion allows for the concept of rigidity to include state-dependent weight functions on the edges of the graph, weights which can then be exploited to take into account inter-agent sensing and communication constraints and/or requirements.

{A gradient-based \emph{rigidity maintenance action} aimed at `maximizing' the rigidity eigenvalue was also proposed in~\citep{2012-ZelFraRob}.
However, while this gradient control law was decentralized in structure, there was still a dependence on the availability of several global quantities, namely, of the robot relative positions in some \emph{common reference frame}, of the value of the rigidity eigenvalue, and of the \emph{rigidity eigenvector} associated with the rigidity eigenvalue.}
{A main contribution of this work is then the development of the machinery needed to distributedly estimate all these global quantities by resorting to only \emph{relative distance measurements} among neighbors, so as to ultimately allow for a \emph{fully distributed and range-based} implementation of the rigidity maintenance controller. To this end, we first show that if the formation is infinitesimally rigid, it is possible to \emph{distributedly} estimate the relative positions of neighboring robots in a common reference frame from only range-based measurements. {Our approach relies explicitly on the form of the symmetric rigidity matrix developed here, in contrast to other approaches focusing on distributed implementations of centralized estimation schemes, such as a Gauss-Newton approach used in \cite{Calafiore2010}.}This first step is then instrumental for the subsequent development of the distributed estimation of the rigidity eigenvalue and eigenvector needed by the rigidity gradient controller. This is obtained by exploiting an appropriate modification of the \emph{power iteration method} for eigenvalue estimation following from the works~\citep{2011d-RobFraSecBue, 2010-YanFreGorLynSriSuk} for the distributed estimation of the connectivity eigenvalue of the graph Laplacian and now applied to rigidity. Finally, we show how to exploit the weights on the graph edges to embed constraints and requirements such as inter-robot and obstacle avoidance, limited communication and sensing ranges, and line-of-sight occlusions, into a unified gradient-based \emph{rigidity maintenance control law}.}

Our approach, therefore, can be considered as a contribution to the general problem of distributed strategies for maintaining certain architectural features of a multi-robot system (i.e. connectivity or rigidity) {with minimal sensing requirements (only relative distance measurements).}
Additionally, we also provide a thorough experimental validation of the entire framework by employing a group of $6$ quadrotor UAVs as robotic platforms to demonstrate the feasibility of our approach in real-world conditions.

The organization of this paper is as follows. Section~\ref{prelim} provides a brief overview of some notation and fundamental theoretical properties of graphs. In Section~\ref{sec:rigidity}, the theory of rigidity is introduced, and our extension of the rigidity eigenvalue to 3-dimensional weighted frameworks is given. We then proceed to present a general strategy for a distributed rigidity maintenance controller in Section~\ref{sec:rigidity_control}.  This section will provide details on certain operational constraints of the multi-robot team and how these constraints can be embedded in the control law.  This section also highlights the need to develop distributed algorithms for estimating a common reference frame for the team, outlined in Section~\ref{sec:relpos_est}, and estimation of the rigidity eigenvalue and eigenvector, detailed in Section~\ref{sec:rigidity_est}.  The results of the previous sections are then summarized in Section~\ref{sec:final_control} where the full distributed rigidity maintenance controller is given.  The applicability of these results are then experimentally demonstrated on a robotic testbed consisting of $6$ quadrotor UAVs operating in a obstacle populated environment.  Details of the experimental setup and results are given in Section~\ref{sec:experiment}.  Finally, some concluding remarks are offered in Section~\ref{sec:conclusion}.

\subsection{Preliminaries and Notations} \label{prelim}
The notation employed is standard.  Matrices are denoted by capital letters (e.g.,~$A$), and vectors by lower case letters (e.g.,~$x$).  {The $ij$-th entry of a matrix $A$ is denoted $[A]_{ij}$.}  The rank of a matrix $A$ is denoted $\rk[A]$.  Diagonal matrices will be written as $D = \diag{\{d_1,\ldots,d_n\}}$; this notation will also be employed for block-diagonal matrices.  A matrix and/or a vector that consists of all zero entries will be denoted by ${\bf 0}$; whereas, `$0$' will simply denote the scalar zero. Similarly, the vector $\ones_{{n}}$ denotes the ${n \times 1}$ vector of all ones.  The $n \times n$ identity matrix is denoted as $I_n$.  The set of real numbers will be denoted as $\reals$, and $\| \, \cdot \,  \|$ denotes the standard Euclidean $2$-norm for vectors.  The Kronecker product of two matrices $A$ and $B$ is written as $A \otimes B$ \citep{Horn1991}.

Graphs and the matrices associated with them will be widely used in this work; see, e.g., \citep{2001-GodRoy}.  An undirected (simple) weighted graph $\mc{G}$ is specified by a vertex set $\mc{V}$, an edge set $\mc{E}$ whose elements characterize the incidence relation between distinct pairs of $\mc{V}$, and diagonal $|\mc{E}|\times |\mc{E}|$ weight-matrix ${W}$, with ${[}{W}{]}_{kk}\geq 0$ the weight on edge $e_k \in \mc{E}$. {In this work we consider only finite graphs and denote the cardinality of the node and edge sets as $|\mc{V}|=n$ and $|\mc{E}|=m$.}  
Two vertices $i$ and $j$ are called {\it adjacent} (or neighbors) when $\{i,j\}\in\mc{E}${.}
The \emph{neighborhood} of the vertex $i$ is the set $\mc{N}_i = \{j \in \mc{V} \, | \, \{i,j\} \in \mc{E}\}$.  An \emph{orientation} of an undirected graph $\mc{G}$ is the
assignment of directions to its edges, i.e., an edge $e_{k}$ is an ordered pair $(i,j)$ such that $i$ and $j$ are, respectively, the initial and the terminal nodes of $e_{k}$.

The incidence matrix $E(\mc{G}) \in \reals^{{n \times m}}$ is a $\{0,\pm 1\}$-matrix
with rows and columns indexed by the vertices and edges of $\mc{G}$ such that $[E(\mc{G})]_{ik}$ has the value `$+1$' if node $i$ is the initial node of edge $e_k$, `$-1$' if it is the terminal node, and `0' otherwise.  The degree of vertex $i$, $d_i$, is the cardinality of the set of vertices adjacent to it.  The degree matrix, $\Delta(\mc{G})$, and the adjacency matrix, $A(\mc{G})$, are defined in the usual way~\citep{2001-GodRoy}.  The (graph) Laplacian of $\mc{G}$, $L(\mc{G}) =  E(\mc{G})E(\mc{G})^T =  \Delta(\mc{G}) - A(\mc{G})$, is a positive-semidefinite matrix.  One of the most important results from algebraic graph theory in the context of 
collective motion control states that a graph is connected if and only if the second smallest eigenvalue of the Laplacian is positive~\citep{2001-GodRoy}.

Table \ref{tab:notation} provides a summary of the notations used throughout the document.

\begin{table}[b]
\caption{Notations}
\begin{center}
\begin{tabular}{|c||c|}\hline \hline
$\mc{G}=(\mc{V},\mc{E})$ & a graph defined by its vertex and edge sets \\ \hline
$\mc{N}_i(t)$ & \emph{time-varying} neighborhood of node $v_i \in \mc{V}$ \\ \hline
$p(i)$  & position vector in $\reals^3$ of the mapped node $v_i \in \mc{V}$; \\
$p_i^s$& $s \in \{x,y,z\}$ coordinate of position vector for node $i$ \\ \hline
$p(\mc{V})$ & stacked position {matrix of all nodes ($\reals^{{n} \times 3}$)} \\ \hline
$\xi(i)$ & velocity vector in $\reals^3$ of the node $v_i \in \mc{V}$ \\ \hline
$(\mc{G},p,\mc{W})$ & a weighted framework \\ \hline
$R(p,\mc{W})$ & rigidity matrix of a weighted framework \\ \hline
$\mc{R}$ & symmetric rigidity matrix of a weighted framework \\ \hline
$\lambda_7$, ${\bf v}_7$ (${\bf v}$) & rigidity eigenvalue and eigenvector \\ \hline
$ \ell_{ij} $ & distance between nodes $v_i,v_j \in \mc{V}$, i.e., $\|p(v_i)-p(v_j)\|$ \\ \hline
$\hat{\lambda}_7^i$ & agent $i$'s estimate of the rigidity eigenvalue \\ \hline
$\hat{{\bf v}}_i^s$ & $s$-coordinate of the agent $i$ estimation\\ &  of the rigidity eigenvector \\ \hline
$ \hat{p}_{i,c}$ & agent $i$ estimate of relative position vector $p_i-p_c$ \\ \hline
$\hat{p}$ & stacked vector of the relative  \\ 
& position vector estimate {$p_i-p_c$, $i=1\ldots {n}$}\\
\hline
{$\mbox{\bf{avg}}(x)$} & {the average of a vector $x \in \reals^n$, $\mbox{{\bf avg}}(x) = \frac{1}{n}\sum_{i=1}^nx_i$} \\ \hline
$ \overline{{\bf v}}_i^x$ & agent $i$ estimate of ${{\bf avg}}(\hat{{\bf v}}^x)$ \\ \hline
$ \overline{{\bf v}}_i^{2x}$ & agent $i$ estimate of ${{\bf avg}}(\hat{{\bf v}}^x\circ \hat{{\bf v}}^x)$ \\ \hline
$ z_i^{xy} $ & agent $i$ estimate of ${{\bf avg}}(\hat{ p}^{y,c}\circ \hat{{\bf v}}^x-\hat{p}^{x,c}\circ \hat{{\bf v}}^y)$ \\ \hline
$ z_i^{xz} $ & agent $i$ estimate of ${{\bf avg}}(\hat{ p}^{z,c}\circ \hat{{\bf v}}^x-\hat{p}^{x,c}\circ \hat{{\bf v}}^z)$ \\ \hline
$ z_i^{yz} $ & agent $i$ estimate of ${{\bf avg}}(\hat{p}^{y,c}\circ \hat{{\bf v}}^z-\hat{p}^{z,c}\circ \hat{{\bf v}}^y)$ \\ \hline
\end{tabular}
\end{center}
\label{tab:notation}
\end{table}%

\section{Rigidity and the Rigidity Eigenvalue}\label{sec:rigidity}

In this section we review the fundamental concepts of graph rigidity \citep{1993-GraSerSer, Jackson2007}.  A contribution of this work is an extension of our previous results on the concepts of the \emph{symmetric rigidity matrix} and \emph{rigidity eigenvalue} for 3-dimensional ambient spaces \citep{2012-ZelFraRob}, and the notion of \emph{weighted frameworks}.

\subsection{Graph Rigidity and the Rigidity Matrix}\label{sec:rigidity_matrix}

We consider graph rigidity from what is known as a $d$-\emph{dimensional bar-and-joint framework}.  A framework is the pair $(\mc{G}, p)$, where $\mc{G}=(\mc{V},\mc{E})$ is a graph, and $p\,:\,\mc{V} \rightarrow \reals^{d}$ maps each vertex to a point in $\reals^{d}$. In this work we consider frameworks in a three-dimensional ambient space, i.e., $d=3$.  Therefore, for node $u \in \mc{V}$, $p(u) = \leftm{ccc} p_u^x & p_u^y & p_u^z \rightm^T$ is the position vector in $\reals^3$ for the mapped node. We refer to the matrix $p(\mc{V}) = \leftm{ccc} p(v_1) & \cdots & p(v_{{n}}) \rightm^T \in \reals^{{n} \times 3}$ as the \emph{position matrix}. We now provide some basic definitions. 

\begin{definition}
Frameworks $(\mc{G},p_0)$ and $(\mc{G},p_1)$ are \emph{equivalent} if $\| p_0(u)-p_0(v)\| = \| p_1(u)-p_1(v)\|$ for all $\{u,v\} \in \mc{E}$, and are \emph{congruent} if $\| p_0(u)-p_0(v)\| = \| p_1(u)-p_1(v)\|$ for all $\{u,v\} \in \mc{V}$.
\end{definition}

\begin{definition}
 A framework $(\mc{G},p_0)$ is \emph{globally rigid} if every framework which is equivalent to $(\mc{G},p_0)$ is congruent to $(\mc{G},p_0)$.
\end{definition}
\begin{definition}
 A framework $(\mc{G},p_0)$ is \emph{rigid} if there exists an $\epsilon > 0$ such that every framework $(\mc{G},p_1)$ which is equivalent to $(\mc{G},p_0)$ and satisfies $\|p_0(v)-p_1(v)\| < \epsilon$ for all $v \in \mc{V}$, is congruent to $(\mc{G},p_0)$.
\end{definition}

\begin{definition}
A \emph{minimally rigid graph} is a rigid graph such that the removal of any edge results in a non-rigid graph.
\end{definition}

\begin{figure}[!t]
\begin{center}
	\subfigure[Two equivalent minimally rigid frameworks in $\reals^3$. The framework on the right side is obtained by the reflection of the position of $v_5$ with respect to the plane characterized by the positions of $v_1$, $v_2$, and $v_3$ (as illustrated in grey).]{\includegraphics[width=\columnwidth]{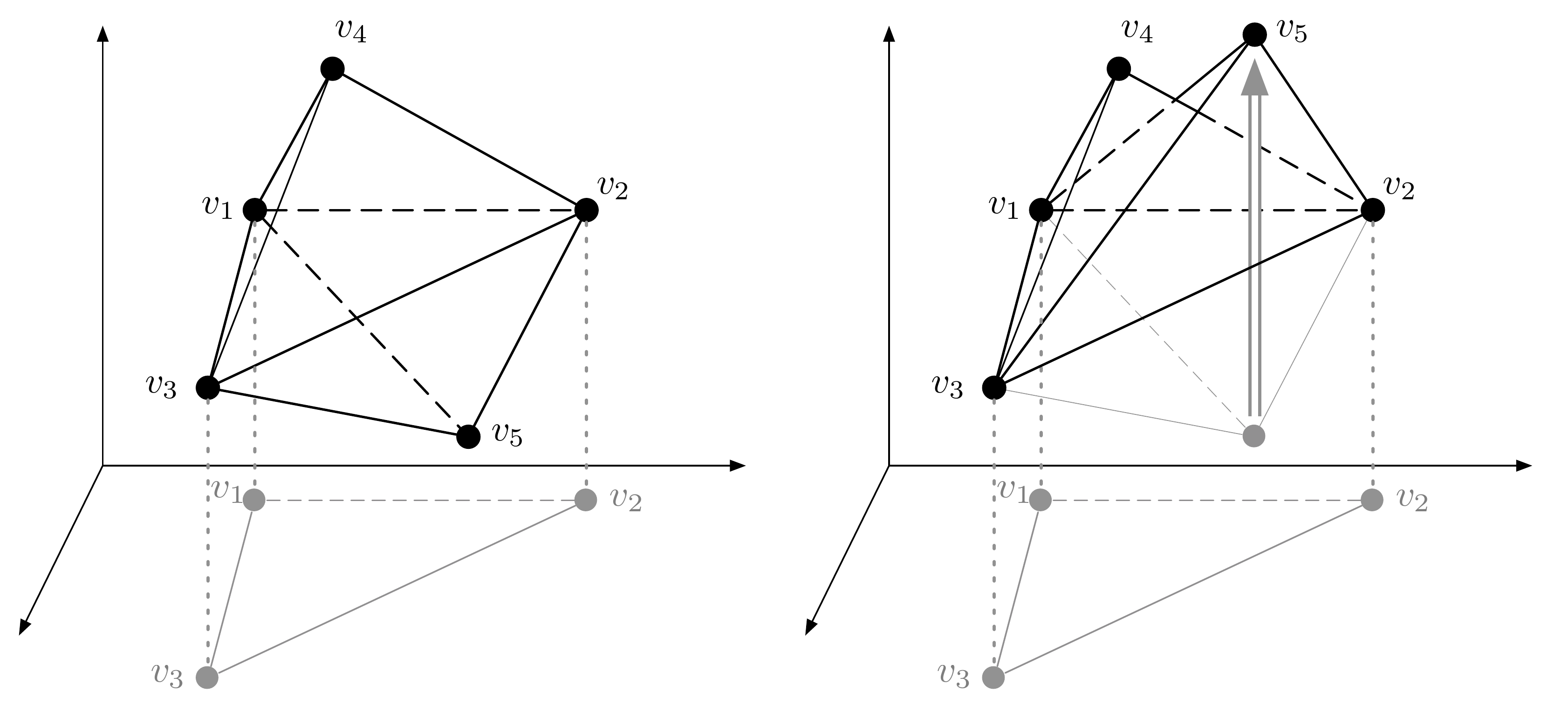}\label{fig:rigid1}}\hfill
	\subfigure[An infinitesimally and globally rigid framework in $\reals^3$.]{\includegraphics[width=0.46\columnwidth]{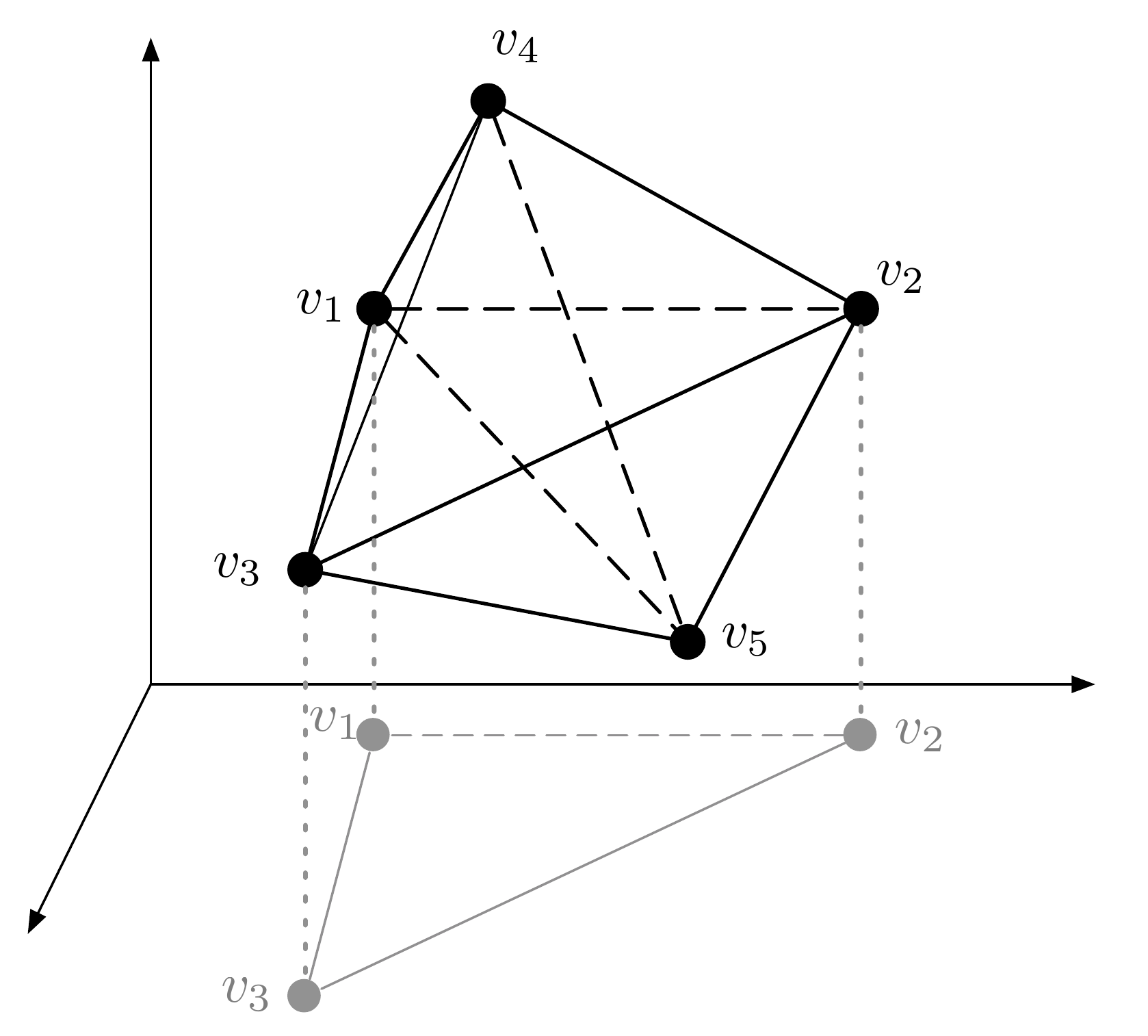}\label{fig:global_rigid}}\hfill
	\subfigure[A non-infinitesimally rigid framework (note that vertexes $v_1$ and $v_3$ are connected).]{
{\includegraphics[width=0.51\columnwidth]{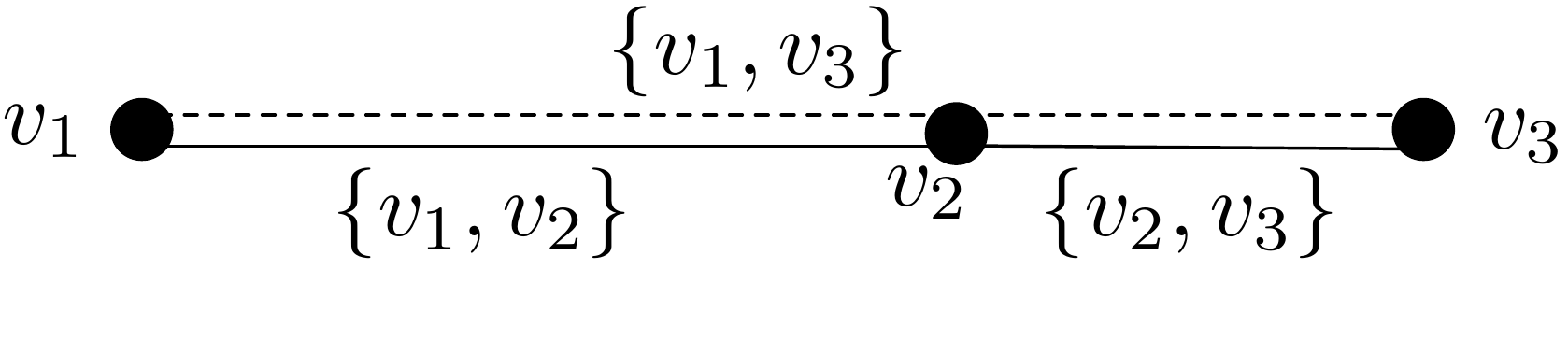}}\label{fig:noinf_rigid}}

  \caption{Examples of rigid and infinitesimally rigid frameworks in $\reals^3$. Notice that in Figs.~\subref{fig:rigid1} and~\subref{fig:global_rigid} the 3D points associated to each vertex do not lie on the same plane, while in Fig.~\subref{fig:noinf_rigid} the 3D points are aligned. } \label{fig:rigid_graphs}
\end{center}
\vspace{-20pt}
\end{figure}

Figure \ref{fig:rigid_graphs} shows three frameworks illustrating the above definitions.  The frameworks in Figure \ref{fig:rigid1} are both minimally rigid and are equivalent to each other, but are not congruent, and therefore not globally rigid.  By adding an additional edge, as in Figure \ref{fig:global_rigid} (the edge $\{v_4,v_5\}$), the framework becomes globally rigid.  The key feature of global rigidity, therefore, is that the distances between \emph{all} node pairs are maintained for different framework realizations, and not just those defined by the edge set.

 By parameterizing the position map by a positive scalar representing time, we can also consider \emph{trajectories} of a framework.  That is, the position map now becomes $p : \mc{V} \times \reals \rightarrow \reals^3$ and is assumed to be  continuously differentiable with respect to time.  We then explicitly write $(\mc{G},p,t)$ so as to represent a \emph{time-varying framework}.  In this direction, we can define a set of trajectories that are \emph{edge-length preserving}, in the sense that for each time $t \geq t_0$, the framework $(\mc{G},p,t)$ is equivalent to the framework $(\mc{G},p,{t_0})$.  More formally, an edge-length preserving framework must satisfy the constraint
\bea \label{framework_len}
\|p(v,t)-p(u,t)\| = \|p(v,t_0)-p(u,t_0)\| = \ell_{vu}, \, \forall t \geq t_0  
\eea
and for all $\{v,u\} \in \mc{E}$. 

One can similarly assign velocity vectors $\xi(u,t)\in\reals^3$ to each vertex $u \in \mc{V}$ for each point in the configuration space such that
\bea \label{inf_motion}
\hspace{-3pt}(\xi(u,t)-\xi(v,t))^T(p(u,t)-p(v,t)) &\hspace{-6pt}=&\hspace{-6pt} 0, \, \forall \, \{u,v\} \in \mc{E}.\;\;
\eea
Note that this relation can be obtained by time-differentiation of the length constraint described in (\ref{framework_len}).  These motions are referred to as \emph{infinitesimal motions} of the mapped vertices $p(u,t)$, and one has 
\begin{equation}\label{eq:vel_agent}
{\dot{p}(u,t) = \xi(u,t).}
\end{equation} 
For the remainder of this paper, we drop the explicit inclusion of time for frameworks and simply write $(\mc{G},p)$ and $p(u)$ and $\xi(u)$ for the time-varying positions and velocities.  The velocity vector {$\xi(u)$} will be treated as the \emph{agent velocity input} throughout the rest of the paper (see Section~\ref{sec:rigidity_control}).

Infinitesimal motions of a framework can be used to define a stronger notion of rigidity.

\begin{definition}\label{def:inf_rigidity}
A framework is called \emph{infinitesimally rigid} if every possible motion that satisfies (\ref{inf_motion}) is trivial (i.e., consists of only global rotations and translations of the whole set of points in the framework).
\end{definition}

An example of an infinitesimally rigid graph in $\reals^3$ is shown in Figure~\ref{fig:global_rigid}. Furthermore, note that infinitesimal rigidity implies rigidity, but the converse is not true~\citep{1985-TayWhi}, see Figure~\ref{fig:noinf_rigid} for a rigid graph in $\reals^3$ that is not infinitesimally rigid.   
 
The infinitesimal motions in (\ref{inf_motion}) define a system of ${m}$ linear equations in the vector of unknown velocities $\xi=[\xi^T(v_1) \; \ldots \; \xi^T(v_{{n}})]^T\in \reals^{3{n}}$.  This system can be equivalently written as the linear matrix equation
\[
R(p)\xi={\bf 0},
\] 
where $R(p) \in \reals^{{m} \times 3{n}}$ is called \emph{rigidity matrix} \citep{1985-TayWhi}. Each row of $R(p)$ corresponds to an edge $e=\{u,v\}$ and the quantity $(p(u)-p(v))$ represents the nonzero coefficients for that row. For example, the row corresponding to edge $e$ has the form
\beas
\leftm{ccccc} -{\bf 0}- & \underbrace{(p(u)-p(v))^T}_{\mbox{vertex } u} & -{\bf 0}- & \underbrace{(p(v)-p(u))^T}_{\mbox{vertex }v} &-{\bf 0}- \rightm .
\eeas
The definition of infinitesimal rigidity can then be restated in the following form:
\begin{lem}[\cite{1985-TayWhi}]\label{thm:rigidity}
A framework $(\mc{G},p)$ in $\mathbb{R}^3$ is infinitesimally rigid if and only if $\rk [R(p)] = 3{n}-6$.
\end{lem}
{Note that, as expected from Definition~\ref{def:inf_rigidity}, the six-dimensional kernel of $R(p)$ for an infinitesimally rigid graph only allows for six independent \emph{feasible} framework motions, that is, the above-mentioned collective roto-translations in $\reals^3$ space. Note also that, despite its name, the rigidity matrix is actually characterizing \emph{infinitesimal rigidity} rather than \emph{rigidity} of a framework.}

\subsection{Rigidity of Weighted Frameworks}\label{sec:sym_rigidity_matrix}

We now introduce an important generalization to the concept of rigidity and the rigidity matrix by introducing weights to the framework.  Indeed, as discussed in the introduction, our aim is to propose a control law able to not only maintain \emph{infinitesimal rigidity} of the formation as per Definition~\ref{def:inf_rigidity}, but to also concurrently manage additional constraints typical of multi-robot applications such as collision avoidance and limited sensing and communication.This latter objective will be accomplished via the introduction of suitable state-dependent weights, thus requiring an extension of the traditional results on rigidity to a weighted case.

\begin{definition}\label{def:weighted_framework}
A $d$-dimensional \emph{weighted} framework is the triple $(\mc{G},p,\mc{W})$, where $\mc{G}=(\mc{V},\mc{E})$ is a graph, $p : \mc{V} \rightarrow \reals^d$ is a function mapping each vertex to a point in $\reals^d$, and $\mc{W} : (\mc{G},p) \rightarrow \reals^{{m}}$ is a function {of the framework that assigns a scalar value to each edge in the graph.} 
\end{definition}

Using this definition, we can also define the corresponding \emph{weighted rigidity matrix}, $R(p,\mc{W})$ as
\bea\label{weighted_rigidity_matrix}
R(p,\mc{W}) &=& W(\mc{G},p) R(p),
\eea
where $W(\mc{G},p) \in \reals^{{m} \times {m}}$ is a diagonal matrix containing the elements of the vector $\mc{W}(\mc{G},p)$ on the diagonal.  Often we will simply refer to the weight matrix $W(\mc{G},p)$ as $W$ when the underlying graph and map $p$ is understood.

\begin{rem}\label{rem:weighted_rigidity}
Note that the rigidity matrix $R(p)$ can also be considered as a weighted rigidity matrix with $W(\mc{G},p) = I$.  Another useful observation is that the unweighted framework $(\mc{G},p)$ can also be cast as a weighted framework $(K_{{n}},p,\mc{W})$, where $K_{{n}}$ is the complete graph on ${n}$ nodes and $[W(\mc{G},p)]_{ii}$ is 1 whenever $e_i \in \mc{E}(K_{{n}})$ is also an edge in $\mc{G}$, and 0 otherwise.  
\end{rem}

Weighted rigidity can lead to a slightly different interpretation of infinitesimal rigidity, where the introduced weights might cause the rigidity matrix to lose rank.  That is, an unweighted framework might be infinitesimally rigid, whereas a weighted version might not.  This observation is trivially observed by considering a minimally infinitesimally rigid framework $(\mc{G},p)$ and introducing a weight with a 0 entry on any edge. We formalize this with the following definitions.

\begin{definition}\label{def:unweighted}
The \emph{unweigted counterpart} of a weighted framework $(\mc{G},p,\mc{W})$ is the framework $(\hat{\mc{G}},p)$ where the graph $\hat{\mc{G}}=(\mc{V},\hat{\mc{E}})$ is such that $\hat{\mc{E}}\subset \mc{E}$ and the edge $e_i \in \mc{E}$ is also an edge in $\hat{\mc{G}}$ if and only if the corresponding weight is non-zero (i.e. $[W(\mc{G},p)]_{ii} \neq 0$).
\end{definition}

\begin{definition}\label{def:weighted_framework_rigid}
A weighted framework is called infinitesimally rigid if its unweighted counterpart is infinitesimally rigid.
\end{definition}

We now present a corollary to Lemma \ref{thm:rigidity} for weighted frameworks.

\begin{cor}\label{cor:weighted_ridgidity}
A weighted framework $(\mc{G},p,\mc{W})$ in $\reals^3$ is infinitesimally rigid if and only if $\rk [R(p,\mc{W})] = 3{n}-6$.
\end{cor}

\begin{proof}
The statement follows from the fact that $\rk [R(p,\mc{W})] = \rk [\hat{R}(p)]$, 
where $\hat{R}(p)$ is the rigidity matrix for the unweighted counterpart of $(\mc{G},p,\mc{W})$.
\end{proof}

\subsection{The Rigidity Eigenvalue}\label{subsec:rigidity_eig}

In our previous work \citep{2012-ZelFraRob}, we introduced an alternative representation of the rigidity matrix that transparently separates the underlying graph from the positions of each vertex.  Here we recall the presentation and  extend it to the case of 3-dimensional frameworks.

\begin{definition}[\cite{2012-ZelFraRob}]\label{def:localgraph}
Consider a graph $\mc{G}=(\mc{V},\mc{E})$ and its associated incidence matrix with arbitrary orientation $E(\mc{G})$. The \emph{directed local graph} at node $v_j$ is the sub-graph $\mc{G}_j=(\mc{V},\mc{E}_j)$ induced by node $v_j$ such that
$$\mc{E}_j = \{(v_j,v_i) \, | \, e_k=\{v_i,v_j\} \in \mc{E}\}.$$
The \emph{local incidence matrix} at node $v_j$ is the matrix
$$E_l(\mc{G}_j) = E(\mc{G}) \diag{\{s_1, \ldots, s_{{m}}}\}\in \reals^{{n} \times {m}}$$
where $s_k = 1$ if $e_k \in \mc{E}_j$ and $s_k=0$ otherwise. %
\end{definition}

Note, therefore, that the local incidence matrix will contain columns of all zeros in correspondence to those edges not adjacent to $v_j$.  This also implicitly assumes a predetermined labeling of the edges.

\begin{prop}[{\cite{2012-ZelFraRob}}]\label{prop:alt_rigidity}
Let $p(\mc{V}) \in \reals^{{n} \times 3}$ be the position matrix for the framework $(\mc{G},p)$.  The rigidity matrix $R(p)$ can be defined as
\bea \label{new_rigiditymatrix}
\hspace{-5pt}R(p) &\hspace{-7pt}=&\hspace{-7pt} \leftm{ccc} E_l(\mc{G}_1)^T & \cdots & E_l(\mc{G}_{{n}})^T\rightm \left(I_{{n}} \otimes p(\mc{V})\right),\eea
where $E_l(\mc{G}_i)$ is the local incidence matrix for node $v_i$.
\end{prop}
\noindent{A more detailed discussion and example of these definitions are provided in Appendix B.}

Lemma~\ref{thm:rigidity} and Corollary \ref{cor:weighted_ridgidity} relate the property of {infinitesimal} rigidity for a given (weighted) framework to the rank of a corresponding matrix.  A contribution of this work is the translation of the rank condition to that of a condition on the spectrum of a corresponding matrix that we term the \emph{symmetric rigidity matrix}.  For the remainder of this work, we will only consider weighted frameworks, since from the discussion in Remark \ref{rem:weighted_rigidity}, any framework can be considered as a weighted framework with appropriately defined weights.

The symmetric rigidity matrix for a weighted framework $(\mc{G},p,\mc{W})$ is a symmetric and positive-semidefinite matrix defined as
\bea \label{symmetric_rigidity_matrix}
\mc{R} &:=& R(p,\mc{W})^TR(p, \mc{W}) \in \reals^{3{n} \times 3{n}}.
\eea
An immediate consequence of the construction of the symmetric rigidity matrix is that $\rk[\mc{R}] = \rk[R(p, \mc{W})]$ \citep{Horn1985}, leading to the following corollary.
\begin{cor}\label{cor:sym_rigid_matrix}
A weighted framework $(\mc{G},p, \mc{W})$ is infinitesimally rigid if and only if $\rk[\mc{R}] = 3{n}-6$.
\end{cor}

The rank condition of Corollary \ref{cor:sym_rigid_matrix} can be equivalently stated in terms of the eigenvalues of $\mc{R}$.  Denoting the eigenvalues of $\mc{R}$ as $\lambda_1 \leq \lambda_2 \leq \ldots \leq \lambda_{3{n}}$, note that infinitesimal rigidity is equivalent to $\lambda_i=0$ for $i=1,\ldots, 6$ and $\lambda_7 > 0$.  Consequently, we term $\lambda_7$ the \emph{Rigidity Eigenvalue}.  We will now show that, in fact, for any connected graph,\footnote{If the graph is not connected, there will be additional eigenvalues at the origin corresponding to the number of connected components of the graph, see \citep{2001-GodRoy}.} the first six eigenvalues are always $0$. 

The first result in this direction shows that the symmetric rigidity matrix is similar to a weighted Laplacian matrix.

\begin{prop}\label{prop:rigidity_weighted_lap}
The symmetric rigidity matrix is similar to the weighted Laplacian matrix via a permutation of the rows and columns as
\bea \label{rigidity_weightedlap}
\hspace{-10pt}P\mc{R}P^T &=&\left(I_3 \otimes E(\mc{G})W\right)Q(p(\mc{V})) \left(I_3 \otimes WE(\mc{G})^T\right),
\eea
with
\bea \label{rigidity_weighted_bis}
\hspace{-7pt} Q(p(\mc{V})) \hspace{-7pt}&=\hspace{-7pt}&  \leftm{ccc} Q_x^2 & Q_{x}Q_y & Q_xQ_{z}  \\ Q_yQ_{x} & Q_y^2 & Q_yQ_{z}\\Q_z Q_{x} & Q_zQ_{y} & Q_{z}^2 \rightm {\in \reals^{3m \times 3m}},
\eea
where $Q_x$, $Q_y$, and $Q_z$ are {$m \times m$} diagonal weighting matrices for each edge in $\mc{G}$ such that for the edge $e_k = (v_i,v_j)$,
\beas
[Q_s]_{kk} &=  (p^s_i-p^s_j), \;  s \in \{x, y, z\}
\eeas
and $p_{i}^x$ ($p_i^y$, $p_i^z$) represents the $x$-coordinate ($y$-coordinate, $z$-coordinate) of the position of agent $i$.
\end{prop}
\begin{proof}
The proof is by direct construction using Proposition \ref{prop:alt_rigidity} and (\ref{symmetric_rigidity_matrix}).  Consider the permutation matrix $P$ as
\bea\label{permutation}
 P = \leftm{c} I_{{n} }\otimes \leftm{ccc} 1 & 0 & 0 \rightm \\  I_{{n}} \otimes \leftm{ccc} 0 & 1 & 0   \rightm \\  I_{{n}} \otimes \leftm{ccc} 0 & 0 & 1   \rightm \rightm.
 \eea
and let $\hat{E} =  \leftm{ccc} E_l(\mc{G}_1)^T & \cdots & E_l(\mc{G}_{{n}})^T\rightm$.  It is straightforward to verify that 
$$  (I_{{n}} \otimes (p^x)^T)\hat{E}^TW = E(\mc{G})W\hspace{-10pt}\underbrace{\leftm{ccc} \ddots && \\ & (p^x_i-p^x_j) & \\ && \ddots \rightm}_{\mbox{diagonal matrix of size }{m} \times {m}},$$
where $p^x$ represents the first column of the position vector.
The structure of the matrix in (\ref{rigidity_weightedlap}) then follows directly.\footnote{A more detailed proof for the two-dimensional case is provided in \citep{2012-ZelFraRob}.}
\end{proof}

The representation of the symmetric rigidity matrix as a weighted Laplacian allows for a more transparent understanding of certain eigenvalues related to this matrix.  The next result shows that the first six eigenvalues of $\mc{R}$ must equal zero for any connected graph $\mc{G}$. 
\begin{thm}\label{thm:eigs_symrigid}
Assume that a weighted framework $(\mc{G},p,\mc{W})$ has weights such that the weight matrix $W(\mc{G},p)$ is invertible and the underlying graph $\mc{G}$ is connected.  Then the symmetric rigidity matrix has at least six eigenvalues at the origin;  that is, $\lambda_i = 0$ for $i \in \{1,\ldots,6\}$.  Furthermore, a possible set of linearly independent eigenvectors associated with each 0 eigenvalue is, 
\beas
 \left\{P^T \leftm{c} \ones_{{n}} \\ {\bf 0} \\ {\bf 0} \rightm,  P^T \leftm{c} {\bf 0}\\ \ones_{{n}}\\ {\bf 0} \rightm, P^T \leftm{c} {\bf 0}\\ {\bf 0}\\ \ones_{{n}}  \rightm,  \right.&& \\
 \left. P^T \leftm{c} (p^y)\\ -(p^x)\\ {\bf 0} \rightm, P^T \leftm{c} (p^z)\\ {\bf 0}\\ -(p^x) \rightm, P^T \leftm{c} {\bf 0}\\ (p^z)\\ -(p^y) \rightm \right\}, &&
\eeas
where $P$ is defined in (\ref{permutation}).

\end{thm}
\begin{proof}
Recall that for any connected graph, one has $E(\mc{G})^{T} \ones_{{n}} = 0$ \citep{2001-GodRoy}.  Therefore, $P\mc{R}P^T$ must have three eigenvalues at the origin, with eigenvectors $u_1= \leftm{ccc} \ones_{{n}}^T & {\bf 0}^T& {\bf 0}^T \rightm^T$, $u_2= \leftm{ccc} {\bf 0}^T&\ones_{{n}}^T& {\bf 0}^T   \rightm^T$, and $u_3= \leftm{ccc} {\bf 0}^T& {\bf 0}^T&\ones_{{n}}^T   \rightm^T$.   We now demonstrate that the remaining three eigenvectors proposed in the theorem are indeed in the null-space of the symmetric rigidity matrix.

Let $u_4 = \leftm{ccc} (p^y)^T & -(p^x)^T & {\bf 0}^T \rightm^T$. Observe that $(I_3 \otimes WE(\mc{G})^T)u_4 = \leftm{ccc} b_1^T & b_2^T & {\bf 0}^T \rightm^T$ is such that $b_1$ is $\pm [W]_{kk}(p^y_i-p^y_j)$ only for edges $e_k = \{v_i,v_j\} \in \mc{E}$, and 0 otherwise.  Similarly, $b_2$ is $\pm [W]_{kk}(p^x_j-p^x_i)$ only for edges $e_k = \{v_i,v_j\} \in \mc{E}$.  The invertibility assumption of the weight matrix also guarantees that $[W]_{kk} \neq 0$.  It can now be verified that from this construction one has
$$  \leftm{ccc} Q_x^2 & Q_{x}Q_y & Q_xQ_{z}  \\ Q_yQ_{x} & Q_y^2 & Q_yQ_{z}\\Q_z Q_{x} & Q_zQ_{y} & Q_{z}^2 \rightm 
 (I_3 \otimes WE(\mc{G})^T)u_4 = 0.$$
The remaining two eigenvectors follow the same argument as above.  It is also straightforward to verify that $u_4$, $u_5$, and $u_6$ are linearly independent of the first 3 eigenvectors.
\end{proof}

Theorem~\ref{thm:eigs_symrigid} provides a precise characterization of the eigenvectors associated with the null-space of the symmetric rigidity matrix for an infinitesimally rigid framework.

\begin{rem}\label{rem:rigmtx_eigenvector}
It is important to note that the chosen eigenvectors associated with the null-space of the symmetric rigidity matrix are expressed in terms of the \emph{absolute positions} of the nodes in the framework. We note that these eigenvectors can also be expressed in terms of the \emph{relative position} of each node to any arbitrary reference point  $p_c = \leftm{ccc} p^x_c & p_c^y & p_c^z\rightm^T \in \reals^3$.  For example, vector $u_4$ could be replaced by
$$u_4^{p_c} = P^T\leftm{c} p^y - p_c^y \ones_{{n}}\\ p_c^x\ones_{{n}}-p^x \\ {\bf 0} \rightm, $$ that is a linear combination of the null-space eigenvectors $u_1,u_2$ and $u_4$.  The use of eigenvectors defined on relative positions, in fact, will be necessary for the implementation of a distributed estimator for the rigidity eigenvector and eigenvalue based on only \emph{relative measurements} available from onboard sensing.
\end{rem}

Theorem~\ref{thm:eigs_symrigid} can be used to arrive at the main result relating infinitesimal rigidity to the rigidity eigenvalue.

\begin{thm}\label{thm:rigidity_eig}
A weighted framework $(\mc{G}, p, \mc{W})$ is infinitesimally rigid if and only if the rigidity eigenvalue is strictly positive, i.e., $\lambda_7 > 0$.
\end{thm}

\begin{proof}
The proof is a direct consequence of Corollary \ref{cor:sym_rigid_matrix} and Theorem \ref{thm:eigs_symrigid}.
\end{proof}

Another useful observation relates infinitesimal rigidity of a framework to connectedness of the underlying graph.
\begin{cor}\label{cor:rigid_connectedness}
Rigidity of the weighted framework $(\mc{G},p, \mc{W})$ implies connectedness of the graph $\mc{G}$.
\end{cor}

The connection between infinitesimal rigidity of a framework and the spectral properties of the symmetric rigidity matrix inherits many similarities between the well studied relationship between graph connectivity and the graph Laplacian matrix \citep{2010-MesEge}. 

In the next section, we exploit this similarity and propose a \emph{rigidity maintenance} control law that aims to ensure the rigidity eigenvalue is always positive. Such a control action will be shown to depend on the rigidity eigenvalue, on its eigenvector, and on relative positions among neighboring pairs expressed in a common frame. The issue of how every agent in the group can distributedly estimate these quantities will be addressed in Sections~\ref{sec:relpos_est} {and} \ref{sec:rigidity_est}.

\section{A Decentralized Control Strategy\\ for Rigidity Maintenance}\label{sec:rigidity_control}

The results of Section \ref{sec:rigidity} highlight the role of the rigidity eigenvalue $\lambda_7$ as a measure of the ``degree of infinitesimal rigidity" of a weighted framework $(\mc{G}, p, \mc{W})$.  It provides a linear algebraic condition to test the infinitesimal rigidity of a framework and, especially in the case of weighted frameworks, provides a means of quantifying ``how rigid" a weighted framework is.  Moreover, the symmetric rigidity matrix was shown to have a structure reminiscent of a weighted graph Laplacian matrix, and thus can be considered as a naturally \emph{distributed} operator.

\begin{figure}[!t]
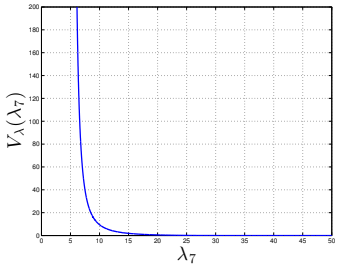

\begin{center}
\Vlambda
  \caption{A possible shape for the rigidity potential function $V_\lambda(\lambda_7)$ with $\lambda^{\mathrm{min}}_7=5$.}\label{fig:Vlambda7}
\end{center}
\end{figure}

The basic approach we consider for the maintenance of rigidity is to define a scalar potential function of the rigidity eigenvalue, $V_\lambda(\lambda_7)>0$, with the properties of growing unbounded as $\lambda_7\to \lambda^{\mathrm{min}}_7>0$ and vanishing (with vanishing derivative) as $\lambda_7\to \infty$ (see Fig.~\ref{fig:Vlambda7} for one possible shape or $V_\lambda$ with $\lambda^{\mathrm{min}}_7=5$). Here, $\lambda^{\mathrm{min}}_7$ represents some predetermined minimum allowable value for the rigidity eigenvalue determined by the needs of the application.  {In addition to maintaining rigidity, the potential function should also capture additional constraints in the system, such as collision avoidance or formation maintenance.}
Each agent should then follow the anti-gradient of this potential function, that is
\begin{equation}\label{eq:gradient_final}
\xi(u)=\dot{p}_u(t) =-\dfrac{\partial V_\lambda}{\partial p_u(t)} = -\dfrac{\partial V_\lambda}{\partial \lambda_7}\dfrac{\partial \lambda_7}{\partial p_u(t)},
\end{equation}
where $\xi(u)$ is the velocity input of agent $u$, as defined in~(\ref{eq:vel_agent}), and $p_u = \leftm{ccc} p_u^x & p_u^y&p_u^z\rightm^T$ is the position vector of the $u$-th agent.
This strategy will ensure that the formation maintains a ``minimum" level of rigidity (i.e., $\lambda_7^{\mathrm{min}}$) at all times.
Of course, this strategy is an inherently \emph{centralized} one, as the computation of the rigidity eigenvalue and of its gradient require full knowledge of the symmetric rigidity matrix. Nevertheless, we will proceed with this strategy and demonstrate that it can be implemented in a fully decentralized manner.

In the sequel, we examine in more detail the structure of the control scheme (\ref{eq:gradient_final}).  First, we show how the formalization of weighted frameworks allows to embed additional weights within the rigidity property that enforce explicit inter-agent sensing and communication constraints and group requirements such as collision avoidance and formation control. {For instance, the weighting machinery will be exploited so as to induce the agents to keep a desired inter-agent distance $\ell_0$ and to ensure a minimum safety distance $\ell_\mathrm{min}$ from neighboring agents and obstacles.} With these constraints, the controller will simultaneously maintain a minimum level of rigidity while also respecting the additional inter-agent constraints.  We then provide an explicit characterization of the gradient of the rigidity eigenvalue with respect to the agent positions, and highlight its distributed structure.  Finally, we present the general control architecture for implementing (\ref{eq:gradient_final}) in a fully decentralized way.

\subsection{Embedding Constraints in a Weighted Framework}\label{subsec:weights}

In real-world applications 
a team of mobile robots may not be able to maintain the same interaction graph throughout the duration of a mission because of various sensing and communication constraints preventing mutual information exchange and relative sensing.   
Furthermore, additional requirements such as collision avoidance with obstacles and among robots, as well as some degree of formation control, must be typically satisfied during the mission execution. Building on the design guidelines proposed in~\citep{RoFrSeBu:13} for dealing with \emph{connectivity} maintenance, 
we briefly discuss here a possible design of weights $\mc{W}$ 
aimed at taking into account the above-mentioned sensing and communication constraints and group requirements within the rigidity maintenance action.

To this end, we start with the following definition of \emph{neighboring agents}:
\begin{definition}\label{def:neigh}
Two agents $u$ and $v$ are considered \emph{neighbors} if and only if $(i)$ their relative distance $\ell_{uv}=\|p(u)-p(v)\|$ is smaller than $D\in\mathbb{R}^+$ (\emph{the sensing range}), 
{$(ii)$ the distance $\ell_{uvo}$ between the segment joining $u$ and $v$ and the closest obstacle point $o$ is larger than $\ell_{\mathrm{min}}$ (the minimum \emph{line-of-sight visibility}),}
and $(iii)$ neither $u$ nor $v$ are closer than $\ell_{\mathrm{min}}$ to any other agent or obstacle.
\end{definition}
Conditions $(i)$ and $(ii)$ are meant to take into account two typical sensing constraints in multi-robot applications: maximum communication and sensing ranges and line-of-sight occlusions. The purpose of condition $(iii)$, which will be better detailed later on, is to force disconnection from the group if an agent is colliding with any other agent or obstacle in the environment. {In the following we will denote with $\calS_u$ the set of neighbors of agent $u$ induced by Definition~\ref{def:neigh}.}

\begin{figure}[t]
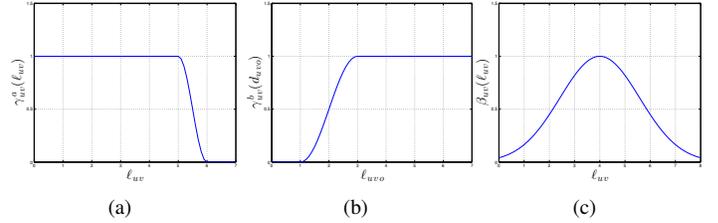

\centering
\subfigure[\label{fig:a_ij}]{\gammaa}\;
\subfigure[\label{fig:b_ij}]{\gammab}
\subfigure[\label{fig:c_ij}]{\betaf}
\caption{The shape of $\gamma^a_{uv}(\ell_{uv})$ for $D=6$ (a), $\gamma^b_{uv}(\ell_{uvo})$ for $\ell_\mathrm{min}=1$ (b), and $\beta_{uv}(\ell_{uv})$ for $\ell_0=4$ (c).}\label{fig:ab}
\end{figure}

This neighboring definition can be conveniently taken into account by designing the inter-agent weights $\mc{W}_{uv}$ as state-dependent functions smoothly vanishing as any of the above constraints and requirements are not met by the pair $(u,\,v)$ with the desired accuracy.  
Indeed, the use of state-dependent weights allows us to consider the ensemble of robots in the context of weighted frameworks, as introduced in Definition \ref{def:weighted_framework}.  In particular, we take the underlying graph to be the complete graph $K_{{n}}$ and the map $p$ corresponds to the physical position state of each agent in a common global frame. The weights are the maps $\mc{W}_{uv}$, and the weighted framework is the triple $(K_{{n}},p,\mc{W})$ {with, therefore, $\calN_u=\{v\in\calV|\;\mc{W}_{uv}\neq 0\}$.}

Following what was proposed in~\citep{RoFrSeBu:13}, {and recalling that $\ell_{uvo}$ represents the distance between the segment joining agents $u$ and $v$ and the closest obstacle point $o$, we then take}
\begin{equation}\label{eq:weights}
\mc{W}_{uv}=\alpha_{uv}\beta_{uv}\gamma^a_{uv}\gamma^b_{uv},
\end{equation}
with $\alpha_{uv}=\alpha_{uv}(\ell_{uk}|_{k\in\calS_u},\,\ell_{vk}|_{k\in\calS_v})$, $\beta_{uv}=\beta_{uv}(\ell_{uv})$, $\gamma^a_{uv}=\gamma^a_{uv}(\ell_{uv})$, $\gamma^b_{uv}=\gamma^b_{uv}(\ell_{uvo})$ and such that
\begin{itemize}
	\item 
	\begin{itemize}
	\item	$\lim_{\ell_{uk}\to \ell_{\mathrm{min}}} \alpha_{uv}=0$, $\forall k\in\calS_u$, 
	\item $\lim_{\ell_{vk}\to \ell_{\mathrm{min}}} \alpha_{uv}=0$, $\forall k\in\calS_v$, and
	\item  $\alpha_{uv}\equiv 0$ if $\ell_{uk}\leq \ell_{\mathrm{min}}$ or $\ell_{vk}\leq \ell_{\mathrm{min}}$, for any $k\in\calS_u$, $k\in\calS_v$;
	\end{itemize}
	\item $\lim_{|\ell_{uv}-\ell_0|\rightarrow \infty} \beta_{uv} = 0$ with $ \beta(\ell_{uv}) < \beta(\ell_0) \, \forall \, \ell_{uv} \neq \ell_{0}$;
	\item $ \lim_{\ell_{uv}\to D} \gamma^a_{uv} = 0$ with $\gamma^a_{uv}\equiv 0 \, \forall \, \ell_{uv}\geq D$;
	\item $ \lim_{\ell_{uvo}\to \ell_{\mathrm{min}}} \gamma^b_{uv} = 0$ with $\gamma^b_{uv}\equiv 0 \, \forall \,\ell_{uvo}\leq \ell_{\mathrm{min}}$.
\end{itemize}
{As explained}, $\ell_{\mathrm{min}}$ is a predetermined minimum safety distance for avoiding collisions and line-of-sight occlusions. Figures~\ref{fig:ab}\subref{fig:a_ij}--\subref{fig:c_ij} show an illustrative shape of weights $\gamma^a_{uv}$, $\gamma^b_{uv}$ and $\beta_{uv}$. The shape of the weights $\alpha_{uv}$ is conceptually equivalent to that of weights $\gamma^b_{uv}$ in Fig.~\ref{fig:b_ij}.

This weight design results in the following properties: for a given pair of agents $(u,\,v)$, the weight $\mc{W}_{uv}$ will vanish (because of the term $\gamma^a_{uv}\gamma^b_{uv}$) whenever the sensing and communication constraints of Definition~\ref{def:neigh} are violated (maximum range, obstacle occlusion), thus resulting in a decreased degree of connectivity of the graph $\calG$ (edge $\{u,\,v\}$ is lost).
The same will happen as the inter-distance $\ell_{uv}$ deviates too much from the desired $\ell_0$ because of the term $\beta_{uv}$. Finally, the term $\alpha_{uv}$ will force \emph{complete disconnection} of vertexes $u$ and $v$ from the other vertexes and therefore a complete loss of connectivity for the graph $\calG$ whenever a collision with another agent is approached.\footnote{As for collision with obstacles, an equivalent behavior is automatically obtained from weights $\gamma^b_{uv}$, see again~\citep{RoFrSeBu:13} for a full explanation. {Also note that, because of the definition of weights $W_{uv}$, one has $\calN_u\subseteq \calS_u$ but $\calS_u\not\subset \calN_u$.}}

We now recall from Corollary \ref{cor:rigid_connectedness} that infinitesimal rigidity implies graph connectivity. Therefore, any decrease in the degree of graph connectivity due to the weights $\mc{W}_{uv}$ vanishing will also result in a decrease of rigidity of the weighted framework $(K_{{n}},p,\mc{W})$ (in particular, rigidity is obviously lost for a disconnected graph).
By maintaining $\lambda_7>0$ (in the context of weighted frameworks) over time, it is then possible to preserve formation rigidity while, at the same time, explicitly considering and managing the above-mentioned sensing and communication constraints and requirements. 

\begin{rem}\label{rem:ell_0}
We note that the purpose of the weight $\beta_{uv}$ in~(\ref{eq:weights}) is to embed a basic level of \emph{formation control} into the rigidity maintenance action: indeed, every neighboring pair will try to keep the desired distance $\ell_0$ thanks to the shape of the weights $\beta_{uv}$. More complex formation control behaviors could be obtained by different choices of functions $\beta_{uv}$ (e.g., for maintaining given relative positions).  Furthermore, formation shapes can be uniquely specified owing to the infinitesimal rigidity property of the configuration.
\end{rem}

\begin{rem}\label{rem:weights}
We further highlight the following properties whose explicit proof can be found in~\citep{RoFrSeBu:13}: the chosen weights $\mc{W}_{uv}$ are functions of only \emph{relative distances} to other agents and obstacles, while their gradients with respect to the agent position $p_u$ (resp.~$p_v$) are functions of \emph{relative positions} expressed in a common reference frame. Furthermore, $\mc{W}_{uv}=\mc{W}_{vu}$ and $\frac{\partial \mc{W}_{uv}}{\partial p_u}=0$, $\forall v\notin \calN_u$. Finally, the evaluation of weights $\mc{W}_{uv}$ and of their gradients can be performed in a decentralized way by agent $u$ (reps.~$v$) by only resorting to local information and $1$-hop communication.
\end{rem}

As shown the next developments, these properties will be instrumental for expressing the gradient of the rigidity eigenvalue as a function of purely relative quantities with respect to only $1$-hop neighbors.

\subsection{The Gradient of the Rigidity Eigenvalue}\label{subsec:gradient_rigidity}

We now present an explicit characterization of the gradient of the rigidity eigenvalue with respect to the agent positions, as used in the control (\ref{eq:gradient_final}).
We first recall that the rigidity eigenvalue can be expressed as
$$ \lambda_{7} = {{\bf v}_7^T\mc{R}{\bf v}_7},$$ 
where ${\bf v}_7$ is the \emph{normalized} rigidity eigenvector associated with $\lambda_7$.  For notational convenience, we consider the permuted rigidity eigenvector $P{\bf v}_7 = \leftm{ccc} ({\bf v}^x)^T & ({\bf v}^y)^T & ({\bf v}^z)^T \rightm^T $, where $P$ is defined in Theorem \ref{thm:eigs_symrigid}. For the remainder of the work, we drop the subscript and reserve the bold font ${\bf v}$ for the rigidity eigenvector.   Note that in fact, the rigidity eigenvalue and eigenvector are state-dependent, and therefore also time-varying when the formation is induced by the spatial orientation of a mobile team of robots, or due to the action of state-dependent weights on the sensing and communication links.

We can now exploit the structure of the symmetric rigidity matrix for weighted frameworks.  Using the form of the symmetric rigidity matrix given in (\ref{rigidity_weightedlap}), we define $\tilde{Q}(p(\mc{V})) = (I_3 \otimes W)Q(p(\mc{V}))(I_3 \otimes W)$ as a generalized weight matrix, and observe that
\beas
P\mc{R}P^T &=&\left(I_3 \otimes E(\mc{G})\right)\tilde{Q}(p(\mc{V})) \left(I_3 \otimes E(\mc{G})^T\right). 
\eeas
The elements of $\tilde{Q}(p(\mc{V}))$ are entirely in terms of the relative positions of each agent and the weighting functions defined on the edges as in (\ref{eq:weights}).

The rigidity eigenvalue can now be expressed explicitly as
\begin{equation}
{\scriptsize\begin{split}
& \lambda_7 = \sum_{(i,\, j)\in \calE} \mc{W}_{ij} \left((p^x_i - p^x_j)^2({\bf v}^x_i-{\bf v}^x_j)^2  +  (p^y_i - p^y_j)^2({\bf v}^y_i-{\bf v}^y_j)^2 + \right. \\
&  (p^z_i - p^z_j)^2({\bf v}^z_i-{\bf v}^z_j)^2   + 2 (p^x_i - p^x_j)(p^y_i-p^y_j)({\bf v}^x_i-{\bf v}^x_j)({\bf v}^y_i-{\bf v}^y_j)   +  \\
& \left.   2(p^x_i - p^x_j)(p^z_i-p^z_j)({\bf v}^x_i-{\bf v}^x_j)({\bf v}^z_i-{\bf v}^z_j)  + \right. \\
& 2 \left. (p^y_i - p^y_j)(p^z_i-p^z_j)({\bf v}^y_i-{\bf v}^y_j)({\bf v}^z_i-{\bf v}^z_j) \right) =\sum_{(i,\, j)\in\calE} \mc{W}_{ij} S_{ij}.
\end{split}}\label{eq:exp_lambda7}
\end{equation}

From~(\ref{eq:exp_lambda7}), one can then derive a closed-form expression for $\frac{\partial \lambda_7}{\partial p^s_i}$, $s\in\{x,\,y,\,z\}$, i.e., the gradient of $\lambda_7$ with respect to each agent's position. In particular, by exploiting the structure of the terms $S_{ij}$ {and the properties of the employed weights $\mc{W}_{ij}$ (see, in particular, the previous Remark~\ref{rem:weights}),
it is possible to reduce $\frac{\partial \lambda_7}{\partial p^x_i}$ to the following \emph{sum over the neighbors}},
\begin{equation}
{\scriptsize\begin{split}
& \dfrac{\partial \lambda_7}{\partial p^x_i} = \sum_{j\in\calN_i} \mc{W}_{ij} \left(  2 (p^y_i-p^y_j)({\bf v}^x_i-{\bf v}^x_j)({\bf v}^y_i-{\bf v}^y_j)   + \right. \\
& \left.  2(p^x_i - p^x_j)({\bf v}^x_i-{\bf v}^x_j)^2   +  2(p^z_i-p^z_j)({\bf v}^x_i-{\bf v}^x_j)({\bf v}^z_i-{\bf v}^z_j)  \right) +\dfrac{\partial \mc{W}_{ij}}{\partial p^x_i} S_{ij},
\end{split}}\label{eq:exp_lambda7_bis}
\end{equation}
and similarly for the $y$ and $z$ components.

The gradient~(\ref{eq:exp_lambda7_bis}) possesses the following key feature: it is a function of relative quantities, in particular of $(i)$ relative components of the eigenvector ${\bf v}$, $(ii)$ relative distances, and $(iii)$ relative positions with respect to \emph{neighboring agents} {(see, again, Remark~\ref{rem:weights} for what concerns weights $\mc{W}_{ij}$),}
thus allowing for a distributed computation of its value once these quantities are locally available. The next sections~\ref{sec:relpos_est} and \ref{sec:rigidity_est} will  detail two estimation schemes able to recover all these relative quantities 
by resorting to only \emph{measured distances} with respect to $1$-hop neighbors owing to the infinitesimal rigidity of the group formation.

\subsection{The Control Architecture}

The explicit description of the gradient of the rigidity eigenvalue in (\ref{eq:exp_lambda7_bis}) motivates the general control architecture for the implementation of the rigidity maintenance action in (\ref{eq:gradient_final}).  We observe that each agent requires knowledge of the rigidity eigenvalue, appropriate components of the rigidity eigenvector, and relative positions with respect to neighboring agents in a common reference frame.  As already mentioned, all these quantities are inherently global quantities, and thus a fully distributed implementation of (\ref{eq:gradient_final}) must include appropriate estimators for recovering these parameters in a distributed manner.
 
\begin{figure}[t]
    \begin{center}
    	\scalebox{.35}{\includegraphics{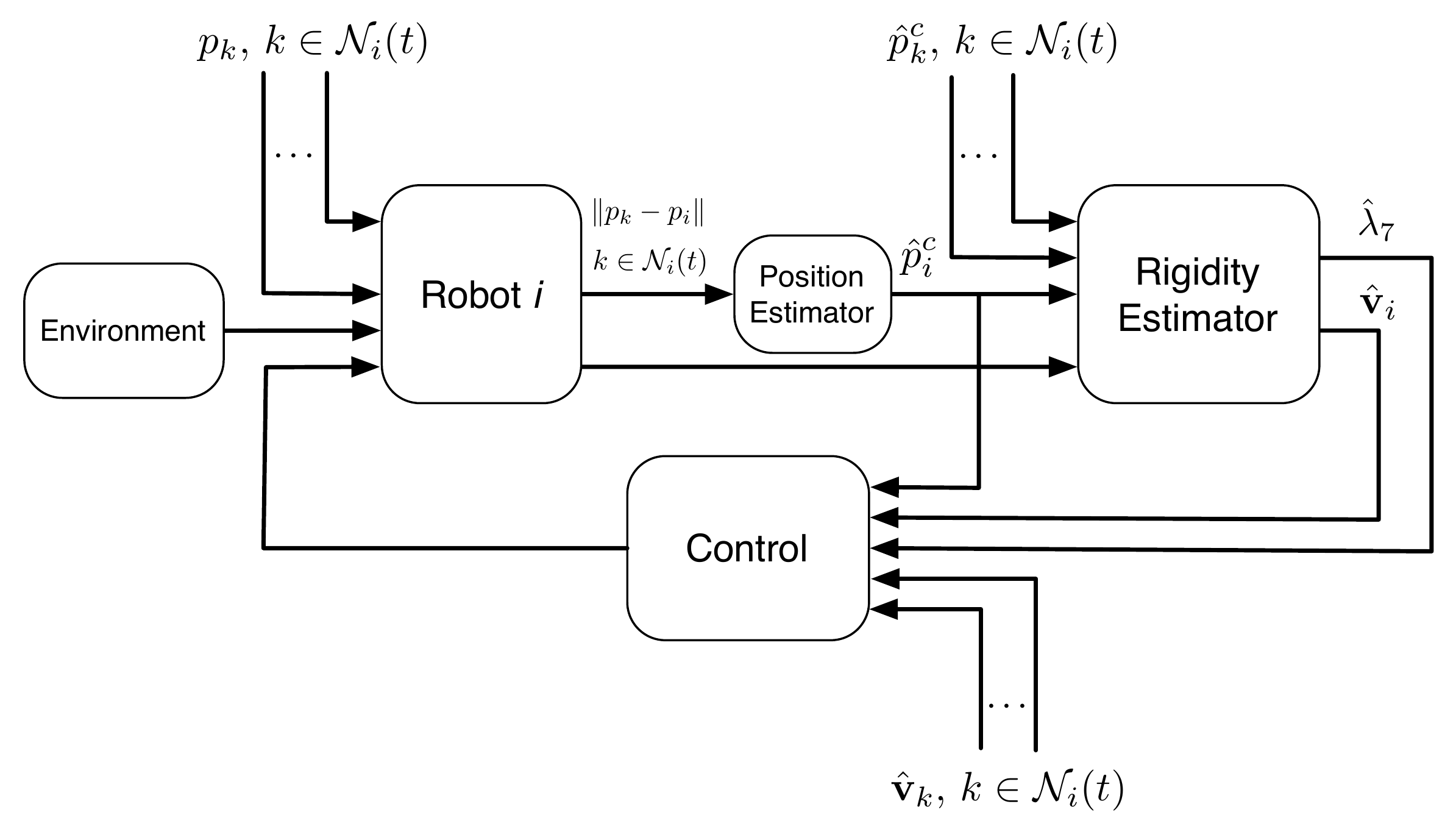}}
    \end{center}
    \caption{Control architecture for distributed rigidity maintenance. }\label{fig:rigidity_bd}
\end{figure}

As a preview of the next sections in this work, Figure~\ref{fig:rigidity_bd} depicts the general architecture needed by each agent 
to implement the rigidity maintenance control action~(\ref{eq:gradient_final}):
\begin{enumerate}
\item  exploiting measured distances with respect to its $1$-hop neighbors, and owing to the formation rigidity, each agent distributely estimates {relative positions in} a common reference frame, labeled as the \emph{position estimator} in the figure. This block is fully explained in section \ref{sec:relpos_est};

\item the output of the position estimator is then used by each agent to perform a distributed estimation of the rigidity eigenvalue ($\hat{\lambda}_7$) and of the relative components of the eigenvector ($\hat{{\bf v}}$), labeled as the \emph{rigidity estimator} in the figure.  This procedure is explained in section \ref{sec:rigidity_est}; 

\item thanks to these estimated quantities (relative positions, $\hat{\lambda}_7$ and $\hat{{\bf v}}$), each agent can finally implement the control action~(\ref{eq:gradient_final}) in a distributed way for maintaining infinitesimal rigidity of the formation during the group motion (while also coping with the various constraints and requirements embedded into weights $\mc{W}$). Maintaining infinitesimal rigidity guarantees in turn convergence of the position estimator from measured distances of step 1), 
and thus closes the `estimation-control loop.'
\end{enumerate}

{We finally note that the proposed control architecture also implicitly assumes the  
\emph{initial} spatial configuration of the agents (i.e., their positions $p(\mc{V})$ at time 0) to be infinitesimally rigid (with, in particular, a $\lambda_7>\lambda_7^{\mathrm{min}}$).  This assumption on the group initial condition is formally stated below.
\begin{assumption}\label{assumption:infringed}
The initial spatial configuration of the agents, $p(\mc{V})$ at time $t=0$, is infinitesimally rigid with $\lambda_7 > \lambda_7^{\mathrm{min}}$.
\end{assumption}
The purpose of requiring a minimum level of rigidity ($ \lambda_7^{\mathrm{min}}$) is discussed in greater detail in Section \ref{sec:experiment}.
}
\section{Decentralized Estimation of\\ Positions in a Common Frame}\label{sec:relpos_est}

As explained, evaluation of the gradient control~(\ref{eq:exp_lambda7_bis}) requires that each agent has access to the \emph{relative positions} of its neighboring agents. 
A main focus of this work, however, is to achieve rigidity maintenance using only \emph{relative distance} measurements.  In this section, we leverage the infinitesimal rigidity of the formation to estimate the relative position with respect to a \emph{common reference point}, $p_c$, shared by all agents.  In particular, each agent $i$, with $i=1\ldots {n}$, will be able to compute an estimate $\hat{p}_{i,c}$ of its relative position $p_{i,c} = p_i - p_c$ to this common point.  By exchanging their estimates over $1$-hop communication channels, two neighboring agents $i$ and $j$ can then build an estimate $\hat p_{j,c}-\hat p_{i,c}$ of their actual relative position $p_j-p_j$ in a common reference frame.
Notice that both the graph (i.e., neighbor sets, edges, etc.) and the robot positions are time-varying quantities. However, in this section we omit dependency on time for the sake of conciseness.

We also note that this common reference point does not need to be stationary, i.e., it can move over time.  In the following, we choose the point $p_c$ to be attached to a \emph{special agent} in the group, determined \emph{a priori}. This agent will be denoted with the index $i_c$ and, in the remainder of this section, we set $p_c = p_{i_c}$. We now proceed to describe a distributed scheme able to recover an estimation of the relative position $p_{i,c}=p_i-p_{i_c}$ for any agent in the group by exploiting the measured relative distances and the rigidity property of the formation.

To achieve this estimation, we first introduce additional assumptions on the capabilities of the special agent $i_c$.
While all agents other than $i_c$ are able to measure only the \emph{relative distance} to their neighbors, the special agent $i_c$ is required to be endowed with an additional sensor able to also measure, at any time $t$, the \emph{relative position} (i.e., distance \emph{and} bearing angles) of at least $2$ non-collinear neighbors;\footnote{Formation rigidity implies presence of at least $2$ non-collinear neighbors for each agent~\citep{1970-Lam}.} 
 these two sensed neighbors will be denoted with the indexes ($\iota(t),\,\kappa(t))\in\calN_{i_c}(t)$.  
\begin{rem}\label{rem:special_agents}
{We stress that the agent indexes $\iota(t)$ and $\kappa(t)$ are \emph{time-varying}; indeed, contrarily to the special agent $i_c$,  $\iota(t)$ and $\kappa(t)$ are not preassigned to any particular agent in the multi-robot team.  Therefore  the special agent $i_c$ only needs to measure its relative positions $p_{\iota(t)}-p_{i_c}$ and $p_{\kappa(t)}-p_{i_c}$ with respect to \emph{any} two agents within its neighborhood ($\iota$ and $\kappa$ are effectively arbitrary), 
with the points $p_{i_c}$, $p_{\iota(t)}$ and $p_{\kappa(t)}$ being non-collinear $\forall t\geq t_0$.
We believe this assumption is not too restrictive in practice, as it only require the presence of at least 
one robot equipped with a range plus bearing sensor while all the remaining ones can be equipped with simple range-only sensors.}
\end{rem}
In the following we omit for brevity the dependency upon the time $t$ of the quantities $\iota$ and $\kappa$.

In order to perform the distributed estimation of $p_{i,c}=p_i-p_c$, $\forall i\in\{1,\ldots, {n}\}$  we follow the approach presented in~\citep{2010-CalCarWei}, with some slight modifications dictated by the nature of our problem. 
Consistently with our notation, we define $\hat p=\leftm{ccc}\hat p_{1,c}^{T}& \ldots & \hat p_{{n},c}^{T} \rightm^T\in\mathbb{R}^{3{n}}$. For compactness, we also denote by $\ell_{ij}$ the measured distance  $\|p_j-p_i\|$, as introduced in Definition \ref{def:neigh}. 
We then consider the following least squares estimation error:
\begin{align}
\begin{aligned}
e(\hat p) &=  \frac{1}{4}\sum_{\{i,j\}\in\calE}\left(\|\hat p_{j,c}-\hat p_{i,c}\|^2 - \ell_{ij}^2\right)^2 + \frac{1}{2}\|\hat p_{i_c,c}\|^2 + \\
&+ \frac{1}{2}\|\hat p_{\iota,c}- (p_\iota - p_{i_c})\|^2 + \frac{1}{2}\|\hat p_{\kappa,c}- (p_\kappa - p_{i_c})\|^2.
\end{aligned}\label{eq:est_error}
\end{align}
Notice that the quantities $\ell_{ij}$, $p_\iota - p_{i_c}$, and $p_\kappa - p_{i_c}$ are measured while all the other quantities represent local estimates of the robots. 

The nonnegative error function $e(\hat p)$ is zero if and only if:
\begin{itemize}
\item $\|\hat p_{j,c}-\hat p_{i,c}\|$ is equal to the measured distance $\ell_{ij}$ for all the pairs $\{i,j\}\in\calE$;
\item $\|\hat p_{i_c,c}\|=0$;
\item $\hat p_{\iota,c}$ and $\hat p_{\kappa,c}$ are equal to the measured relative positions $p_\iota - p_{i_c}$ and $p_\kappa - p_{i_c}$, respectively.
\end{itemize}

Note that the estimates $\hat p_{i_c,c}$, $\hat p_{\iota,c}$ and $\hat p_{\kappa,c}$ could be directly set to $0$, $(p_\iota - p_{i_c})$, and $(p_\iota - p_{i_c})$, respectively, since the first quantity is known and the last two are measured. Nevertheless, we prefer to let the estimator obtaining these values via a `filtering action' for the following reasons: first, the estimator provides a relatively simple way to filter out noise that might affect the relative position measurements; secondly, implementation of the rigidity maintenance controller only requires that $(\hat p_{j,c}-\hat p_{i,c}) \to (p_j- p_i)$, 
which is achieved if $\hat p_{j,c}\to p_j-\hat p_{i_c,c}$ and $\hat p_{i,c}\to p_i-\hat p_{i_c,c}$ for \emph{any} common value of $\hat p_{i_c,c}$. Therefore any additional hard constraint on $\hat p_{i_c,c}$ (e.g., $\hat p_{i_c,c}\equiv 0$) might unnecessarily over-constrain the estimator.

Applying a first-order gradient descent method to $e(\hat p)$, we finally obtain the following \emph{decentralized} update rule for the $i$-th agent ($i \neq i_c$):
\begin{align}\small
\begin{aligned}
\dot{\hat {p}}_{i,c}&=-\parder{e}{\hat p_{i,c}} = \sum_{j\in\calN_i} (\|\hat p_{j,c}-\hat p_{i,c}\|^2 - \ell_{ij}^2) (\hat p_{j,c}-\hat p_{i,c}) - 
\\
&\delta_{ii_c}\hat p_{i,c} - \delta_{i\iota}\left(\hat p_{\iota,c} - (p_\iota - p_{i_c})\right) - \delta_{i\kappa}\left(\hat p_{\kappa,c} - (p_\kappa - p_{i_c})\right),
\end{aligned}\label{eq:err_gradient}
\end{align}
where $\delta_{ij}$ is the well known Kronecker's delta.\footnote{$\delta_{ij}=0$ if $i\neq j$ and $\delta_{ij}=1$ otherwise.}
The estimator~\eqref{eq:err_gradient} is clearly decentralized since:
\begin{itemize}
\item $\ell_{ij}$ is locally measured by agent $i$;
\item $\hat p_{i,c}$ is locally available to agent $i$; 
\item $\hat p_{j,c}$ can be transmitted using one-hop communication from agent $j$ to agent $i$, for every $j\in \calN_i$;
\item  $(p_\iota - p_{i_c})$ and $(p_\kappa - p_{i_c})$ are measured by agent $i_c$ and can be  transmitted using one-hop communication to agents $\iota$ and $\kappa$ respectively.
\end{itemize}
In order to show the relation between the proposed decentralized position estimator scheme and the infinitesimal rigidity property, one can  restate~\eqref{eq:err_gradient} in matrix form as
\begin{align}
\dot{\hat {p}} = -\mc{R}(\hat{p})\hat{p} +R(\hat{p})\ell + \Delta^c
\label{eq:err_gradient_matrix}
\end{align}
where $\mc{R}(\hat{p})$ and $R(\hat {p})$ are the symmetric  rigidity matrix and the rigidity matrix computed with the estimated positions, $\ell \in \mathbb{R}^{|\calE|}$ is a vector whose entries are $\ell_{ij}^2$, $\forall \{i,j\}\in\calE$, and $\Delta^c\in \mathbb{R}^{|\calE|}$ contains the remaining terms of the right-hand-side of~\eqref{eq:err_gradient}. 

\begin{prop}
If the framework is (infinitesimally) rigid then the vector of true values $p-(\ones_{{n}}\otimes p_c) =\leftm{ccc} (p_1-p_c)^T &\cdots & (p_{{n}}-p_c)^T \rightm^T$ is an isolated local minimizer of $e(\hat p)$. Therefore, there exists an $\epsilon > 0$ such that, for all initial conditions satisfying $\|\hat p(0)- p-(\ones_{{n}} \otimes p_c)\| < \epsilon$, the estimation $\hat p$ converges  to $p-(\ones_{{n}}\otimes p_c)$. 
\end{prop}
We point out that the estimator in the form (\ref{eq:err_gradient_matrix}) is identical to the formation controller proposed in \citep{2009-KriBroFra}.  Consequently, we refer the reader to this work for a discussion on the stability and convergence properties of this model.  A similar estimation scheme is also proposed in \citep{2010-CalCarWei}. We briefly emphasize that the property of having the true value of relative positions $p-(\ones_{{n}}\otimes p_c)$ as an isolated local minimizer of~(\ref{eq:est_error}) is a consequence of the definition of infinitesimal rigidity and of the non-collinearity assumption of the agents $i_c$, $\iota$, and $\kappa$. 

We finally note that, in general, the rate of convergence of a gradient descent method is known to be slower than other estimation methods. However, we opted for this method since is its directly amenable to a distributed implementation and requires only first-order derivative information.

\section{Distributed Estimation of the \\Rigidity Eigenvalue and Eigenvector }\label{sec:rigidity_est}

As seen in section \ref{sec:relpos_est}, when the multi-robot team possesses the infinitesimal rigidity property, it is possible to distributedly estimate the relative positions in a common reference frame for each agent.
However, the proposed distributed rigidity maintenance control action (\ref{eq:gradient_final}) requires knowledge of some additional global quantities that are explicitly expressed in the expressions (\ref{eq:exp_lambda7_bis}) and (\ref{eq:gradient_final}).  In particular, each agent must know also the current value of the rigidity eigenvalue and certain components of the rigidity eigenvector. In this section we propose a distributed estimation scheme inspired by the distributed connectivity maintenance solution proposed in \citep{2010-YanFreGorLynSriSuk} for obtaining the rigidity eigenvalue and eigenvector. 

For the reader's convenience, we first provide a brief summary of the \emph{power iteration method} for estimating the eigenvalues and eigenvectors of a matrix.  We then proceed to show how this estimation process can be distributed by employing \emph{PI consensus filters} and by suitably exploiting the structure of the symmetric rigidity matrix.

\subsection{Power Iteration Method}\label{subset:PI}

The power iteration method is one of a suite of iterative algorithms for estimating the dominant eigenvalue and eigenvector of a matrix.  Following the same procedure as in \citep{2010-YanFreGorLynSriSuk}, we employ a continuous-time variation of the algorithm that will compute the smallest non-zero eigenvalue and eigenvector of the symmetric rigidity matrix.  

The discrete-time power iteration algorithm is based on the following iteration,
$$ x^{(k+1)} = \frac{Ax^{(k)}}{\|Ax^{(k)}\|} =   \frac{A^kx^{(0)}}{\|A^kx^{(0)}\|}.$$
Under certain assumptions for the matrix $A$ (i.e., no repeated eigenvalues), the iteration converges to the eigenvector associated to the largest eigenvalue of the matrix.  

To adapt the power iteration to compute the rigidity eigenvector and eigenvalue, we leverage the results of Theorem \ref{thm:eigs_symrigid} and consider the iteration on a \emph{deflated} version of the symmetric rigidity matrix, i.e. $\tilde{\mc{R}} = I-TT^T-\alpha \mc{R}$ for some small enough $\alpha > 0$.  {The power iteration method estimates the largest eigenvalue of a matrix.  As all the eigenvalues of the symmetric rigidity matrix are non-negative, the largest eigenvalue of the deflated version $\tilde{\mc{R}}$ will correspond to $1-\alpha \lambda_7$, and thus can be used to estimate $\lambda_7$.  The constant $\alpha$ ensures the  matrix $\tilde{\mc{R}}$ is positive semi-definite.}The columns of the matrix $T \in \reals^{3{n} \times 6}$ contain the eigenvectors corresponding to the zero eigenvalues of $\mc{R}$, for example, as characterized in Theorem \ref{thm:eigs_symrigid}.
Note that the power iteration applied to the matrix $\tilde{\mc{R}}$ will compute the eigenvector associated with the rigidity eigenvalue.\footnote{Assuming the rigidity eigenvalue is unique and the framework is infinitesimally rigid (i.e., the rigidity eigenvalue is positive). We will discuss the implications of this assumption later.}

The continuous-time counterpart of the power iteration algorithm now takes the form \citep{2010-YanFreGorLynSriSuk} 
\bea \label{pi_ct}
\hspace{-7pt}\dot{\hat{{\bf v}}}(t) &\hspace{-8pt}=&\hspace{-8pt} -\left(k_1TT^T\hspace{-2pt}+\hspace{-1pt}k_2 \mc{R}\hspace{-1pt} +\hspace{-1pt} k_3\left(\hspace{-0pt}\tfrac{\hat{{\bf v}}(t)^T\hat{{\bf v}}(t)}{3{n}} \hspace{-2pt}-\hspace{-2pt} 1\right)I\right) \hat{{\bf v}}(t) ,
\eea
where $\hat{{\bf v}}$ is the \emph{estimate} of the rigidity eigenvector, and the constants $k_1,k_2,k_3>0$ are chosen to ensure the trajectories converge to the rigidity eigenvector.{\footnote{{Note that the constant $\alpha$ used to describe the deflated symmetric rigidity matrix is effectively replaced by $k_2$ in this formulation.}}}  We present here the main result and refer the reader to \cite{2010-YanFreGorLynSriSuk} for details of the proof, noting that the proof methodologies are the same for the system (\ref{pi_ct}) as that proposed in \cite{2010-YanFreGorLynSriSuk}.

\begin{thm}\label{thm:poweriter}
Assume that the {weighted framework $(\mc{G}, p, \mc{W})$ with }symmetric rigidity matrix $\mc{R}$ {is infinitesimally rigid and} has distinct non-zero eigenvalues, and let ${\bf v}$ denote the rigidity eigenvector.  Then for any initial condition $\hat{{\bf v}}(t_0) \in \reals^{3{n}}$ such that ${\bf v}^T\hat{{\bf v}}(t_0) \neq 0$, the trajectories of (\ref{pi_ct}) converge to the subspace spanned by the rigidity eigenvector, i.e., $\lim_{t \rightarrow \infty} \hat{{\bf v}}(t) = \gamma {\bf v}$ for $\gamma \in \reals$, if and only if the gains $k_1,k_2$ and $k_3$ satisfy the following conditions: 
\begin{enumerate}
	\item \label{thm_itm:1} $k_1,k_2,k_3 > 0$, 
	\item \label{thm_itm:2} $k_1 > k_2 \lambda_7$,
	\item \label{thm_itm:3} $k_3 > k_2 \lambda_7$.
\end{enumerate}
Furthermore, for any choice of constants $k_1,k_2,k_3>0$, the trajectories of (\ref{pi_ct}) remain bounded and satisfy
$$\|\hat{{\bf v}}(t)\| \leq \max \left\{\|\hat{{\bf v}}(t_0)\|, \sqrt{3{n}}\right\}, \; \forall \, t \geq t_0.$$
In particular, the trajectory converges to the rigidity eigenvector with 
$$ \lim_{t \rightarrow \infty } \|\hat{{\bf v}}(t)\| = \sqrt{3{n}\left(1-\frac{k_2}{k_3}\right) \lambda_7 }.$$
\end{thm}

\begin{rem}\label{rem:time-varying}
The power iteration proposed in (\ref{pi_ct}) assumes that the symmetric rigidity matrix is \emph{static}.  However, in a dynamic setting the parameters of the rigidity matrix are a function of the state of the robots in a multi-robot system, and both the symmetric rigidity matrix and the expression of its null space are inherently \emph{time-varying}.  While the proof provided in \citep{2010-YanFreGorLynSriSuk} does not explicitly address the time-varying case, our experience suggests that the dynamics of (\ref{pi_ct}) is able to track even a time-varying rigidity eigenvector, so long as the dynamics of the robots are slower than the estimator.  The speed of convergence of (\ref{pi_ct}), of course, is also tunable by the constants $k_i$.
\end{rem}
\begin{rem}\label{rem:lambda8}
Another important subtlety of the dynamics~(\ref{pi_ct}) is the requirement that the rigidity eigenvalue is \emph{unique}.  
When the rigidity eigenvalue is not unique, the associated eigenvector can belong to (at least) a two-dimensional subspace $\calL$, so that~(\ref{pi_ct}) can not be expected to converge to a unique eigenvector but rather to an equilibrium point in $\calL$ (see, e.g.,~\citep{2010-YanFreGorLynSriSuk}). This can pose difficulties in real-world conditions since non-idealities such as noise in measuring the agent states (used in evaluating the symmetric rigidity matrix $\mc{R}$), and discretization when numerically integrating~(\ref{pi_ct}), can make the equilibrium point for~(\ref{pi_ct}) in $\calL$ to abruptly vary over time, thus preventing a successful convergence of the estimation of ${\bf v}$. {}

\end{rem}

\subsection{A Distributed Implementation}\label{sec:distr_est}

The results of section \ref{subset:PI} provide a continuous-time estimator for estimating the rigidity eigenvalue and eigenvector of the symmetric rigidity matrix.  The estimator given in (\ref{pi_ct}), however, is a \emph{centralized} implementation.  Moreover, certain parameters used in (\ref{pi_ct}) are expressed using a common reference frame (i.e., the quantity $TT^T$, see Theorem \ref{thm:eigs_symrigid} and Remark \ref{rem:rigmtx_eigenvector}) or require each robot to know the entire estimator state (i.e., the quantity $\hat{{\bf v}}(t)^T\hat{{\bf v}}(t)$ in (\ref{pi_ct})). 
We propose in this sub-section a distributed implementation for the rigidity estimator that overcomes these difficulties, in particular by leveraging the results of Section \ref{sec:relpos_est}.  In the same spirit as the solution proposed in \citep{2010-YanFreGorLynSriSuk}, we make use of the \emph{PI average consensus filter} \citep{Freeman2006} to distributedly compute the necessary quantities of interest, and strongly exploit the particular structure of the symmetric rigidity matrix.

Our approach to the distribution of (\ref{pi_ct}) is to exploit both the built-in distributed structure (i.e., the symmetric rigidity matrix $\mc{R}$) and the reduction of the other parameters to values that all agents can obtain via a distributed algorithm.  In this direction, we now proceed to analyze each term in (\ref{pi_ct}) and discuss the appropriate strategies for implementing the estimator in a distributed fashion.

Concerning the first term $TT^T\hat{{\bf v}}$, Theorem \ref{thm:eigs_symrigid} provides an analytic characterization of the eigenvectors associated with the zero eigenvalues of the symmetric rigidity matrix (assuming the graph is infinitesimally rigid).  To begin the analysis, we explicitly write out the matrix $T$ and examine the elements of the matrix $TT^T$.   Following the comments of Remark \ref{rem:rigmtx_eigenvector}, we express the null-space vectors in terms of \emph{relative positions} to an arbitrary point $p_c {=\leftm{ccc} p_c^x & p_c^y & p_c^z\rightm} \in \reals^3$; in particular, the point $p_c$ will be the special agent $i_c$ described in Section~\ref{sec:relpos_est}.  

\beas
T &\hspace{-9pt}=&\hspace{-9pt} \leftm{cccccc}  \ones_{{n}} & {\bf 0} & {\bf 0} &p^y-p_c^y \ones_{{n}} & p^z-p_c^z \ones_{{n}}  & {\bf 0} \\ {\bf 0} & \ones_{{n}} & {\bf 0} & p_c^x \ones_{{n}} - p^x & {\bf 0} &p^z-p_c^z \ones_{{n}} \\ {\bf 0} & {\bf 0} & \ones_{{n}}& {\bf 0} & p_c^x \ones_{{n}} - p^x & p_c^y \ones_{{n}} - p^y\rightm
\eeas
For the remainder of this discussion, we assume that all agents have access to their state in an estimated coordinate frame relative to the point $p_{i_c}$, the details of which were described in Section \ref{sec:relpos_est}.

\begin{figure*}[!t]
\normalsize
\bea\label{TT}
TT^T = \leftm{ccc} \ones_{{n}} \ones_{{n}}^T+p^{y,c} (p^{y,c})^T +  p^{z,c} (p^{z,c})^T & - p^{y,c} (p^{x,c})^T & -p^{z,c} (p^{x,c})^T \\ -p^{x,c} ( p^{y,c})^T & \ones_{{n}} \ones_{{n}}^T+ p^{x,c} (p^{x,c})^T + p^{z,c} (p^{z,c})^T  & -  p^{z,c} (p^{y,c})^T \\
- p^{x,c} (p^{z,c})^T & -p^{y,c} (p^{z,c})^T &\ones_{{n}} \ones_{{n}}^T+ p^{x,c} (p^{x,c})^T + p^{y,c} (p^{y,c})^T \rightm
\eea
\hrulefill
\vspace*{4pt}
\end{figure*}

To simplify notations, we write as in Section \ref{sec:relpos_est}, for example, $p^{y,c}= p^y - p_c^y\ones_{{n}}$, and $p_{i,c} = p_i - p_c$.
Following our earlier notation, we also partition the vector $\hat{{\bf v}}$ into each coordinate, $\hat{{\bf v}}^x$, $\hat{{\bf v}}^y$, and $\hat{{\bf v}}^z$.  Let ${{\bf avg}}(r)$ denote the average value of the elements in the vector $r\in \reals^n$, i.e. ${{\bf avg}}(r) = \tfrac{1}{n} \ones_{{n}}^Tr$.  Then it is straightforward to verify that
\begin{align}
\ones_{{n}} \ones_{{n}}^T \hat{{\bf v}}^k(t) &= {n} \, {{\bf avg}}(\hat{{\bf v}}^k(t)) \ones_{{n}}, \; k  \in \{x,y,z\} \label{avg_rigvec}\\
p_{i,c}(p_{j,c})^T \hat{{\bf v}}^k(t) &= {n}\, {{\bf avg}}(p_{j,c} \circ \hat{{\bf v}}^k)p_{i,c}, \;  i,j,k \in \{x,y,z\}, \label{avg_relpos}
\end{align}
where `$\circ$' denotes the element-wise multiplication of two vectors.

This characterization highlights that, in order to evaluate the term $TT^T\hat{{\bf v}}$, each agent must compute the average amongst all agents of a certain value that is a function of the current state of the estimator and of the positions in some common reference frame whose origin is the point $p_c$.  It is well known that the \emph{consensus protocol} can be used to distributedly compute the average of a set of numbers \citep{2010-MesEge}.  The speed at which the consensus protocol can compute this value is a function of the connectivity of the underlying graph and the weights used in the protocol.  In this framework, however, a direct application of the consensus protocol will not be sufficient.  Indeed, it is expected that each agent will be physically moving, leading to a time-varying description of the matrix $TT^T$ (see Remark \ref{rem:time-varying}).  Additionally, the underlying network is also dynamic as sensing links between agents are inherently state dependent.

The use of a \emph{dynamic} consensus protocol introduces additional tuning parameters that can be used to ensure that the distributed average calculation converges faster than the underlying dynamics of each agent in the system, as well as the ability to track the average of a time-varying signal.  We employ the following \emph{PI average consensus filter} proposed in \citep{Freeman2006},

\begin{align} \label{pi_consensus}
\leftm{c} \dot{z}(t) \\ \dot{w}(t) \rightm &= \leftm{cc} -\gamma I_{{n}} - K_PL(\mc{G}(t)) & K_IL(\mc{G}(t)) \\ -K_IL(\mc{G}(t))& 0 \rightm \leftm{c} z(t) \\ w(t) \rightm \nonumber \\
&\quad\; + \leftm{c}\gamma I_{{n}} \\ 0 \rightm u(t) \\
y(t) &= \leftm{cc} I_{{n}} & 0 \rightm \leftm{c} z(t) \\ w(t) \rightm.
\end{align} 
The parameters $K_P, K_I \in \reals $ and $\gamma \in \reals$ are used to ensure stability and tune the speed of the filter.  An analysis of the stability and performance of this scheme with time-varying graphs is given in \citep{Freeman2006}.  Figure \ref{fig:TT_bd} provides a block diagram representation of how the PI consensus filters are embedded into the calculation of $TT^T\hat{{\bf v}}(t)$ (in only the $x$-coordinate).

\begin{figure}[!t]
    \begin{center}
    	\scalebox{.3}{\includegraphics{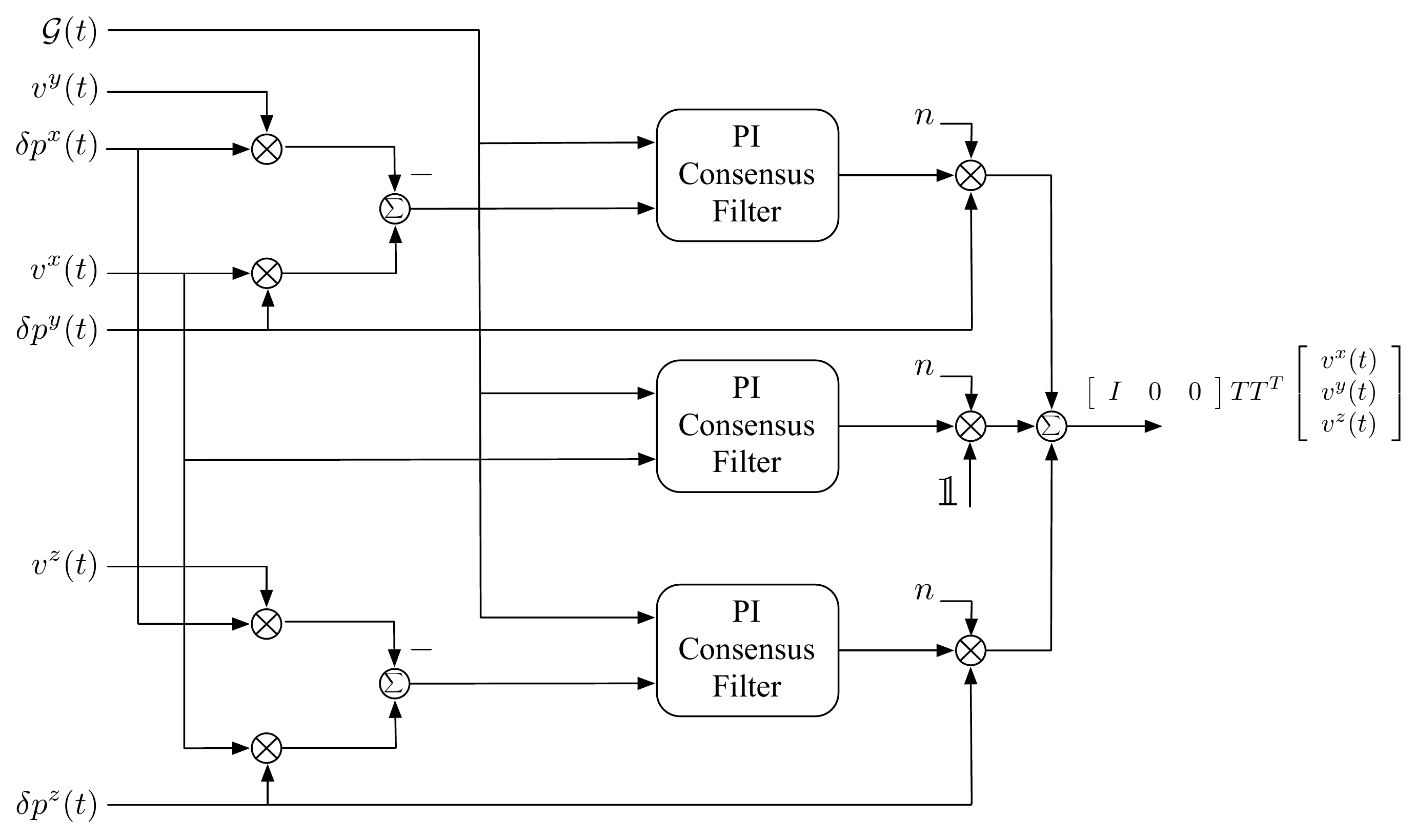}}
    \end{center}
    \caption{Block diagram showing PI consensus filters in calculation of $TT^T\hat{{\bf v}}(t)$.}\label{fig:TT_bd}
\end{figure}

As for the second term in~(\ref{pi_ct}), as shown in $\S$\ref{subsec:rigidity_eig} the symmetric rigidity matrix is by construction a distributed operator.  The term $\mc{R}\hat{{\bf v}}(t)$ can be computed using only information exchanged between neighboring agents, as determined by the sensing graph.  

The final term in~(\ref{pi_ct}) is a normalization used to drive the eigenvector estimate to the surface of a sphere of radius $\sqrt{3{n}}$.  Using the same analysis as above, it can be verified that 
\bea
 \left( \frac{\hat{{\bf v}}(t)^T\hat{{\bf v}}(t)}{3{n}} -1 \right) \hat{{\bf v}}(t) = \left({{\bf avg}}(\hat{{\bf v}}(t)\circ \hat{{\bf v}}(t)) - 1 \right) \hat{{\bf v}}(t). \label{avg_rigvec_sq}
 \eea
This quantity can therefore be distributedly computed using an additional PI consensus filter.

Using the result of Theorem \ref{thm:poweriter} and the PI consensus filters, each agent is also able to estimate the rigidity eigenvalue.
\begin{cor}\label{cor:rig_eig_est}
Let $\overline{{\bf v}}_i^{2}(t)$ denote the output of the PI consensus filter for estimating the quantity ${{\bf avg}}(\hat{{\bf v}}(t) \circ \hat{{\bf v}}(t))$ for agent $i$.  Then agent $i$'s estimate of the rigidity eigenvalue, $\hat{\lambda}_7^i$, can be obtained as
$$\hat{\lambda}_7^i = \frac{k_3}{k_2}\left(1- \overline{{\bf v}}_i^{2}(t)\right).$$
\end{cor}

In summary, each agent implements the following filters:
\begin{itemize}
	\item Estimation of a common reference frame using (\ref{eq:err_gradient}).
	\item Estimation of the rigidity eigenvector using (\ref{pi_ct}).
	\item A PI-Consensus filter for tracking the average of the estimate of the rigidity eigenvector, (\ref{avg_rigvec}).
	\item A PI-Consensus filter for tracking the quantity described in (\ref{avg_relpos}).
	\item A PI-Consensus filter for tracking the average of the square of the rigidity eigenvector estimate, (\ref{avg_rigvec_sq}).
\end{itemize}
For completeness, we now present the full set of filters that each robot executes in (\ref{fullest1})-(\ref{fullest6b}).  These equations are written only for the $x$-coordinate associated with all the quantities.  Observe, however, that the filters needed for the $y-$ and $z-$coordinates do not require additional integrators, as similar filters can be vectorized (for example, the PI filters can be combined as in (\ref{pi_consensus})).   For the readers convenience, a summary of the notations and variable definitions used in (\ref{fullest1})-(\ref{fullest6b}) is provided in Table \ref{tab:notation}.

\begin{figure*}
\normalsize
\begin{align}
\dot{\hat{{\bf v}}}_i^x &= -k_1{n}\left(\overline{{\bf v}}_i^x +z_i^{xy}(t)\hat{p}_{i,c}^y+ z_i^{xz}\hat{p}_{i,c}^z(t)\right)  -k_2 \sum_{j \in \mc{N}_i(t)} W_{ij}\left(\hat{{\bf v}}_i^x(t)-\hat{{\bf v}}_j^x  \right)  - k_3\left(\overline{{\bf v}}_i^x-1  \right)\hat{{\bf v}}_i^x \label{fullest1}\\
\dot{\hat {p}}_{i,c}&= \sum_{j\in\calN_i(t)} (\|\hat p_{j,c}-\hat p_{i,c}\|^2 - \ell_{ij}^2) (\hat p_{j,c}-\hat p_{i,c}) - \delta_{ii_c}\hat p_{i,c} - \delta_{i\iota}\left(\hat p_{\iota,c} - (p_\iota - p_{i_c})\right) - \delta_{i\kappa}\left(\hat p_{\kappa,c} - (p_\kappa - p_{i_c})\right) \label{fullest2}\\
\tikzmark{top1}\,\,
\dot{\overline{{\bf v}}}_i^{x} &=\gamma \left(\hat{{\bf v}}_i^x-\overline{{\bf v}}_i^{x}\right) - K_P \sum_{j \in \mc{N}_i} \left(\overline{{\bf v}}_i^{x}-\overline{{\bf v}}_j^{x}(t)\right) +K_I \sum_{j \in \mc{N}_i(t)} \left(\overline{w}_i^{x}-\overline{w}_j^{x}\right) \label{fullest3a}\\
\tikzmark{bot1}\,\,
\dot{\overline{w}}_i^{x} &= -K_I \sum_{j \in \mc{N}_i(t)} \left(\overline{{\bf v}}_i^{x}-\overline{{\bf v}}_j^{x} \right) \label{fullest3b}\\
\tikzmark{top2}\,\,\,
\dot{\overline{{\bf v}}}_i^{2x} &=\gamma \left((\hat{{\bf v}}_i^x)^2-\overline{{\bf v}}_i^{2x}\right) -K_P \sum_{j \in \mc{N}_i(t)} \left(\overline{{\bf v}}_i^{2x}-\overline{{\bf v}}_j^{2x}\right) +K_I \sum_{j \in \mc{N}_i(t)} \left(\overline{w}_i^{2x}-\overline{w}_j^{2x}\right)\label{fullest4a}\\
\tikzmark{bot2}\,\,
\dot{\overline{w}}_i^{2x} &= -K_I \sum_{j \in \mc{N}_i(t)} \left(\overline{{\bf v}}_i^{2x}-\overline{{\bf v}}_j^{2x}\right)  \label{fullest4b}\\
\tikzmark{top3}\,\,\,
\dot{z}_i^{xy} &=\gamma \left(\left(\hat{p}^y\circ \hat{{\bf v}}^x-\hat{ p}^x\circ \hat{{\bf v}}^y  \right)-z_i^{xy}\right) -K_P \sum_{j \in \mc{N}_i(t)} \left(z_i^{xy}-z_j^{xy}\right) + K_I \sum_{j \in \mc{N}_i(t)} \left(w_i^{xy}(t)-w_j^{xy}\right) \label{fullest5a}\\
\tikzmark{bot3}\,\,
\dot{w}_i^{xy}&= -K_I \sum_{j \in \mc{N}_i(t)} \left(z_i^{xy}-z_j^{xy} \right)  \label{fullest5b}\\
\tikzmark{top4}\,\,\,
\dot{z}_i^{xz}&=\gamma \left(\left(\hat{ p}^{z}\circ \hat{{\bf v}}^x-\hat{p}^x\circ \hat{{\bf v}}^z \right)-z_i^{xz}\right) -K_P \sum_{j \in \mc{N}_i(t)} \left(z_i^{xy}-z_j^{xy}\right) + K_I \sum_{j \in \mc{N}_i(t)} \left(w_i^{xy}-w_j^{xy}\right) \label{fullest6a}\\
\tikzmark{bot4}\,\,
\dot{w}_i^{xz} &= -K_I \sum_{j \in \mc{N}_i(t)} \left(z_i^{xz}-z_j^{xz} \right) \label{fullest6b}
\end{align}
\hrulefill
\DrawVerticalBrace[black, thick]{top1}{bot1}%
\DrawVerticalBrace[black, thick]{top2}{bot2}%
\DrawVerticalBrace[black, thick]{top3}{bot3}%
\DrawVerticalBrace[black, thick]{top4}{bot4}%
\end{figure*}

\begin{rem}
Equations (\ref{fullest1})-(\ref{fullest6b}) show that each agent requires a 10-th order dynamic estimator for estimating the rigidity eigenvector and eigenvalue.  This filter is comprised of three PI-Consensus filters, an relative position estimation filter, and the power iteration filter.  An important point to emphasize is the order of the overall filter is \emph{independent} of the number of agents in the ensemble, and thus is a scalable solution. 
\end{rem}

\section{The Rigidity Maintenance Controller} \label{sec:final_control}

The primary focus of this work until now was a detailed description of how the rigidity of a multi-robot formation can be maintained in a distributed fashion.  The basic idea was to follow the gradient of an appropriately defined potential function of the rigidity eigenvalue; this control strategy was presented in (\ref{eq:exp_lambda7_bis}).  The fundamental challenge for the implementation of this control strategy was twofold: on the one hand, rigidity of a formation is an inherently \emph{global} property of the network, and on the other hand, the control law depended on relative position measurements in a \emph{common} reference fame.

A truly distributed solution based on this control strategy requires each agent to estimate a common inertial reference frame and also estimate the rigidity eigenvalue and eigenvector of the formation.  The solution to these estimation problems was presented in Sections~\ref{sec:relpos_est} and \ref{sec:rigidity_est}, with the complete set of filter equations summarized in (\ref{fullest1})--(\ref{fullest6b}).  Note that both estimation strategies implicitly require that the underlying formation is infinitesimally rigid {(see also Assumption~\ref{assumption:infringed})}.  The final step for implementation of the rigidity maintenance controller is then to replace all the state-variables given in~(\ref{eq:exp_lambda7_bis}) with the appropriate estimated states computed by the relative position estimators and rigidity eigenvalue estimators.  The local controller for each agent is thus given as,\footnote{The control is shown in the $x$-coordinate; a similar expression can be obtained for the $y$- and $z$- coordinates.}

\begin{equation}
{\scriptsize\begin{split}
&\xi_i^x = - \dfrac{\partial V(\hat{\lambda}_7^i)}{\partial {\lambda}_7}  \sum_{j\in\calN_i} W_{ij} \left(2(\hat{p}^x_{i,c} - \hat{ p}^x_{j,c})(\hat{{\bf v}}^x_i-\hat{{\bf v}}^x_j)^2   +   \right. \\
& \left.  2 (\hat{ p}^y_{i,c}-\hat{ p}^y_{j,c})(\hat{{\bf v}}^x_i-\hat{{\bf v}}^x_j)(\hat{{\bf v}}^y_i-\hat{{\bf v}}^y_j)   + 2(\hat{ p}^z_{i,c}-\hat{ p}^z_{j,c})(\hat{{\bf v}}^x_i-\hat{{\bf v}}^x_j)(\hat{{\bf v}}^z_i-\hat{{\bf v}}^z_j)  \right) + \\
& \dfrac{\partial W_{ij}}{\partial p^x_i} \hat S_{ij},
\end{split}}\label{eq:exp_lambda7_bis2}
\end{equation}
in conjunction with all the estimation filters of (\ref{fullest1})-(\ref{fullest6b}).

{
\begin{rem}
The interconnection of the relative position estimator, rigidity eigenvalue estimator, and gradient controller leads to a highly non-linear dynamics for which a formal proof analysis is not straightfoward. While we are currently working towards a deeper analysis in this sense, the approach taken in this work is to exploit the typical (although informal) time-scale separation argument commonly found in many robotics applications relying on feedback control from an estimated state (as, e.g., when using an extended Kalman filter). Basically, the estimator dynamics is assumed ``fast enough'' such that its transient behavior can be considered as a second-order perturbation with respect to the robot motion (see also~\citep{2010-YanFreGorLynSriSuk}) for an equivalent assumption in the context of decentralized connectivity maintenance control.
\end{rem}
}

\section{Experimental Results}\label{sec:experiment}

In this section we report some experimental results aimed at illustrating the machinery proposed so far for distributed rigidity maintenance. The experiments involved a total of $N=6$ quadorotor UAVs ($5$ real and $1$ simulated) flying the environment shown in Fig.~\ref{fig:exp_setup}. {A video illustrating the various phases of the experiment {(Multimedia Extension~1)} is attached to the paper.} 

All the quadrotor UAVs were implementing the rigidity maintenance action~(\ref{eq:exp_lambda7_bis2}) in addition to the estimation filters presented in (\ref{fullest1})-(\ref{fullest6b}). Additionally, for two of the quadrotor UAVs (namely, quadrotors $1$ and $2$) an exogenous bounded velocity term $\xi^*_i\in\calbR^3$ was also added to~(\ref{eq:exp_lambda7_bis2}); this allows for two human operators to independently control the motion of quadrotors $1$ and $2$ during the experiment, so as to steer the whole formation and trigger the various behaviors embedded in the weights $\calW_{uv}$ (formation control, obstacle avoidance, sensing limitations).\footnote{We note that, being $\xi^*_i$ bounded, its effect does not threaten rigidity maintenance since the control action $\xi_i$ in~(\ref{eq:gradient_final}) always results dominant as $V_\lambda(\lambda_7)\to\infty$ if $\lambda_7(t)\to \lambda^{\mathrm{min}}_7$.}

Our experimental quadrotor platform is a customized version of the MK-Quadro\footnote{\url{mikrokopter.de}} implementing the TeleKyb ROS framework\footnote{ros.org/wiki/telekyb} for flight control, experimental workflow management and human inputing. Attitude is stabilized with a fast inner loop that takes advantage of high-rate/onboard accelerometer and gyroscope measurements while the velocity stabilization is achieved by a slower control loop that measures the current velocity thanks to an external motion capture system.
 The motion capture system is also used to obtain relative distance measurements among the robots and the two bearing measurements needed by the special robot $i_c$.
The reader is referred to~\citep{2012f-FraSecRylBueRob} for a detailed description of the quadrotor-based experimental setup.

\begin{figure}
\centering
\includegraphics[width=\columnwidth]{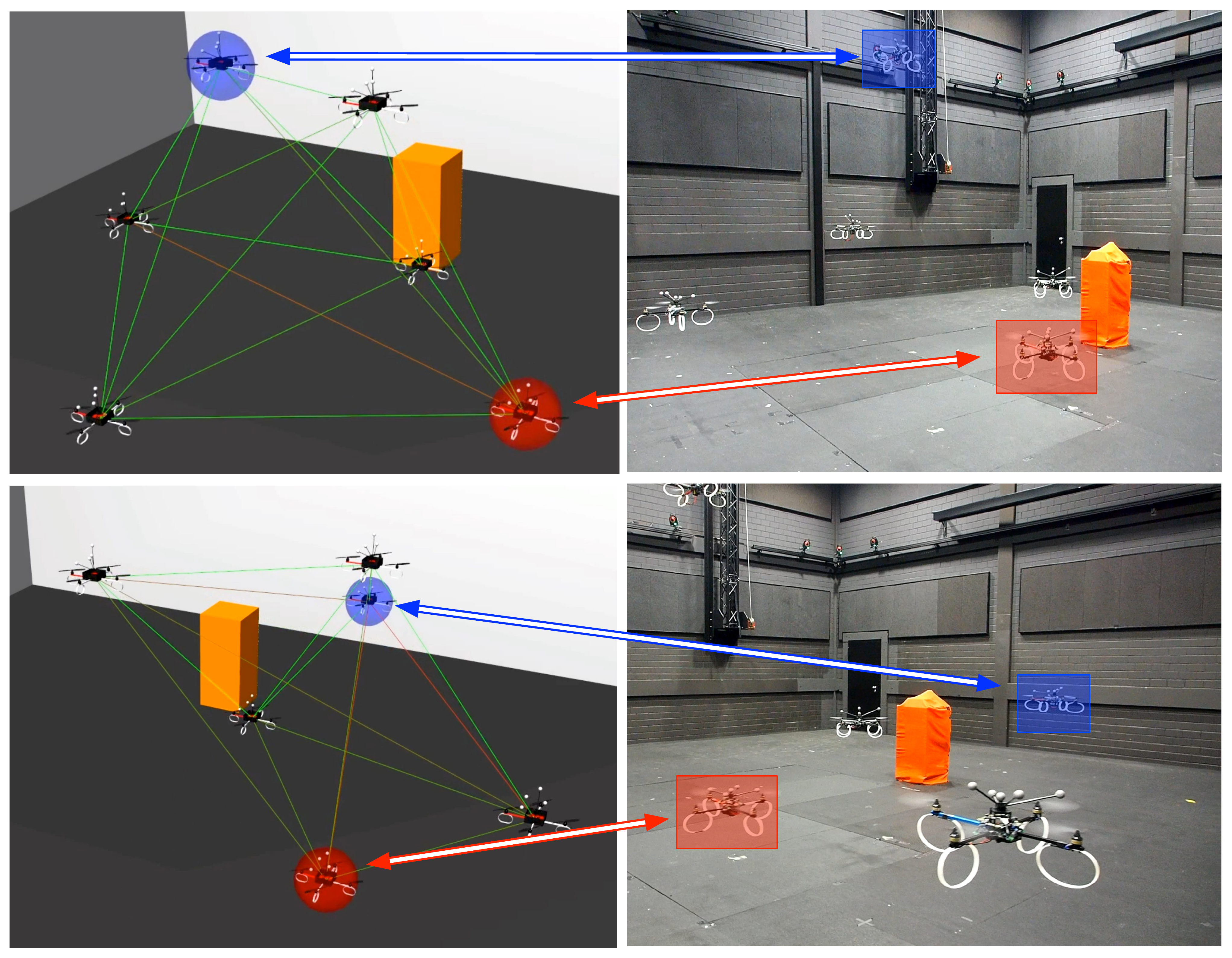}
\caption{Two snapshots of the reported experiment. Left: simulated 3D views showing, in particular, the inter-agent links (red -- almost disconnected link, green -- optimally connected link). Right: corresponding pictures of the experimental setup. The two highlighted quadrotor UAVs are partially controlled by two human operators}
\label{fig:exp_setup}
\end{figure}

We start illustrating the behavior of the relative position estimator described in Sect.~\ref{sec:relpos_est} and upon which all the subsequent steps are based (estimation of $\lambda_7$ and ${\bf v}$ and evaluation of the control action~(\ref{eq:gradient_final})). As explained in Sect.~\ref{sec:relpos_est}, owing to the formation infinitesimal rigidity,
the scheme~(\ref{eq:err_gradient}) allows each agent $i$ to build an estimation $\hat p_{i,c}$ of its relative position $p_i- p_c$ with respect to the agent $i_c$,  with $i_c=1$ in this experiment. 
Figures~\ref{fig:e_pos_exp2}(a--e) report the behavior of the norm of the estimation errors $\|p_i- p_c-\hat p_{i,c}\|$ for $i=2\ldots 6$ together with their mean values (dashed horizontal black line). It is then possible to verify how the relative position estimation errors keep low values over time, 
thus effectively allowing every agent to recover its correct relative position with respect to $p_c$ from the measured relative distances.
\begin{figure}[!t]
\begin{center}
\subfigure[]{\includegraphics[width=0.45\columnwidth]{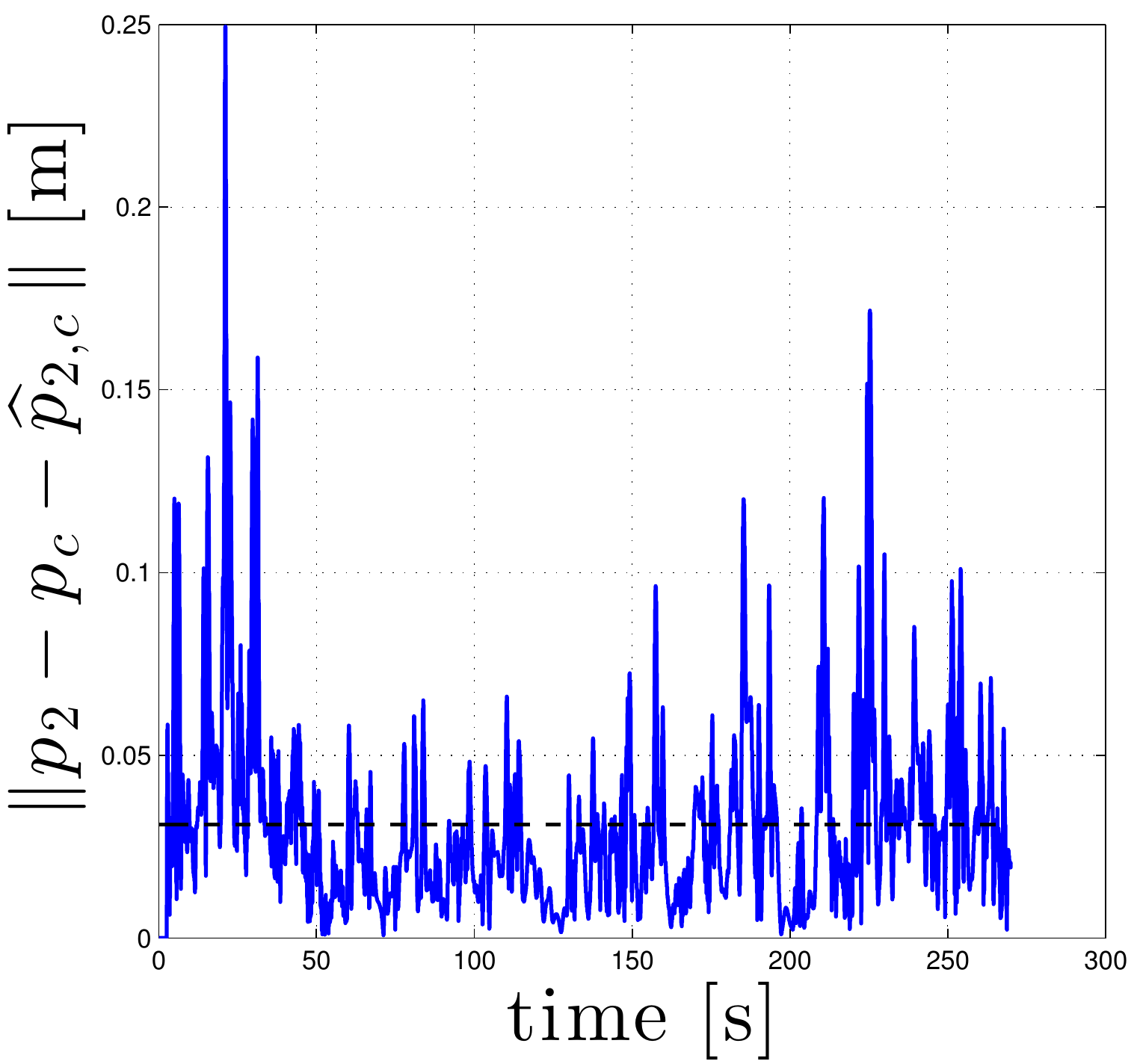}}\quad\subfigure[]{\includegraphics[width=0.45\columnwidth]{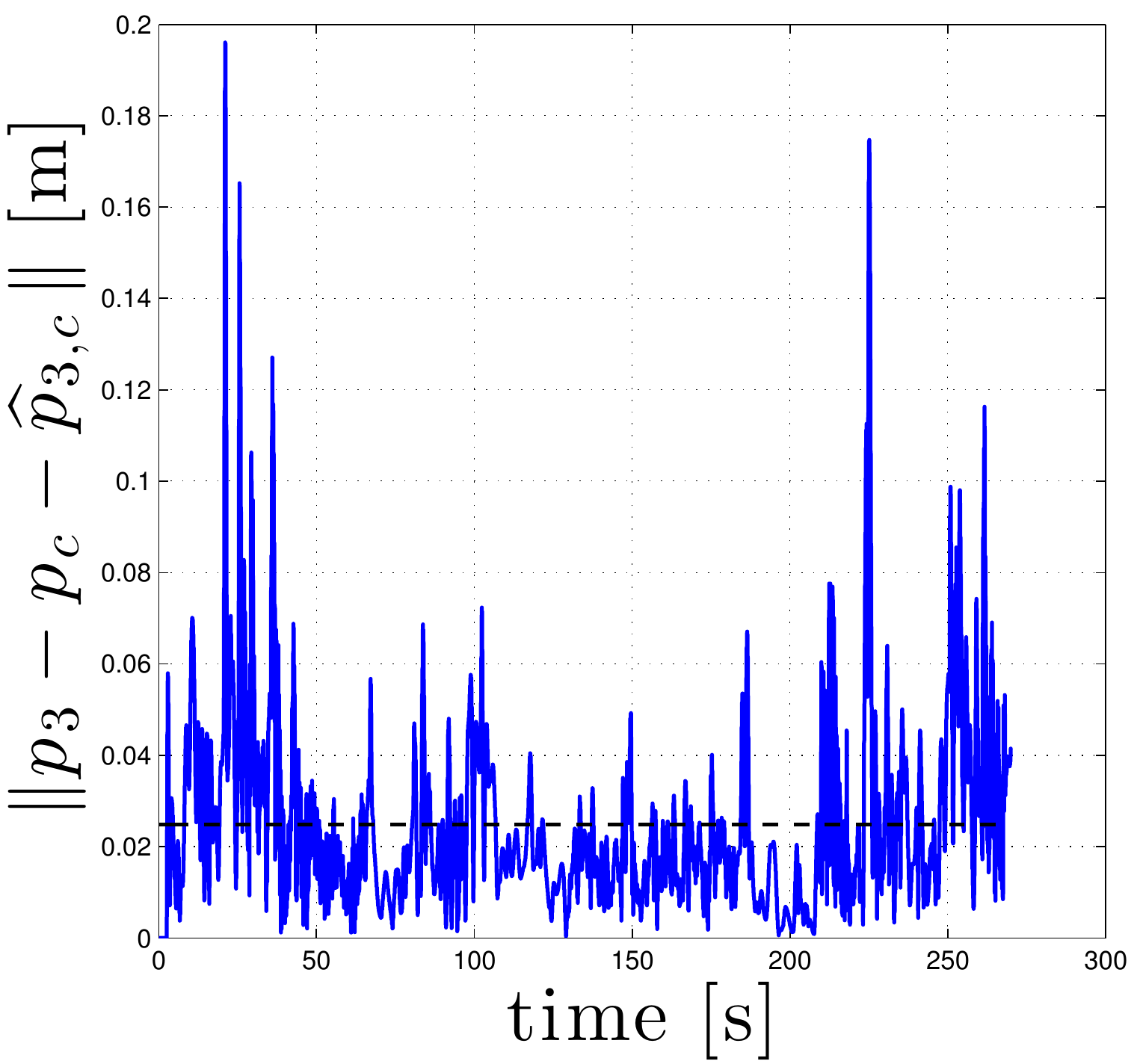}}\\
\subfigure[]{\includegraphics[width=0.45\columnwidth]{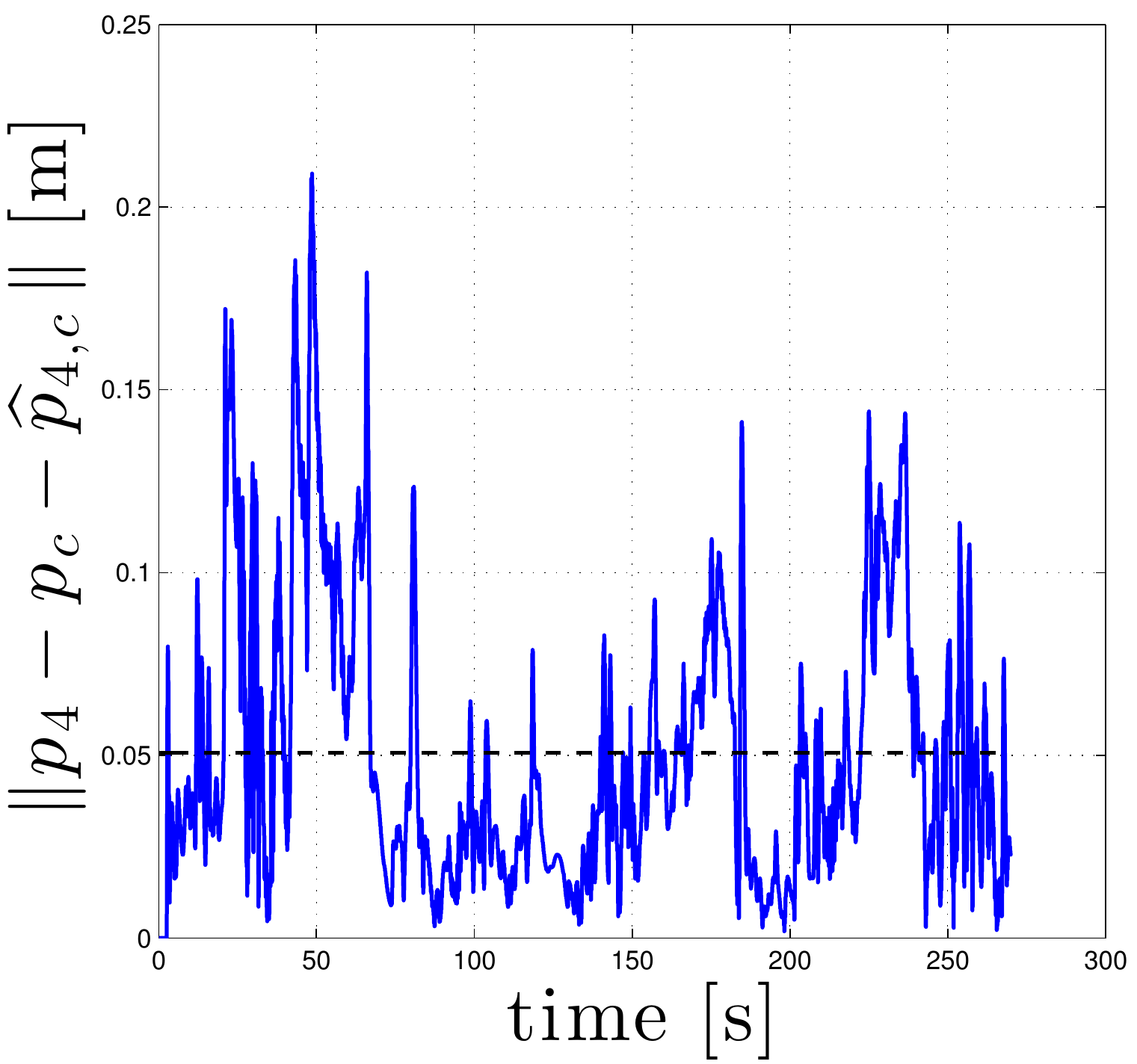}}\quad\subfigure[]{\includegraphics[width=0.45\columnwidth]{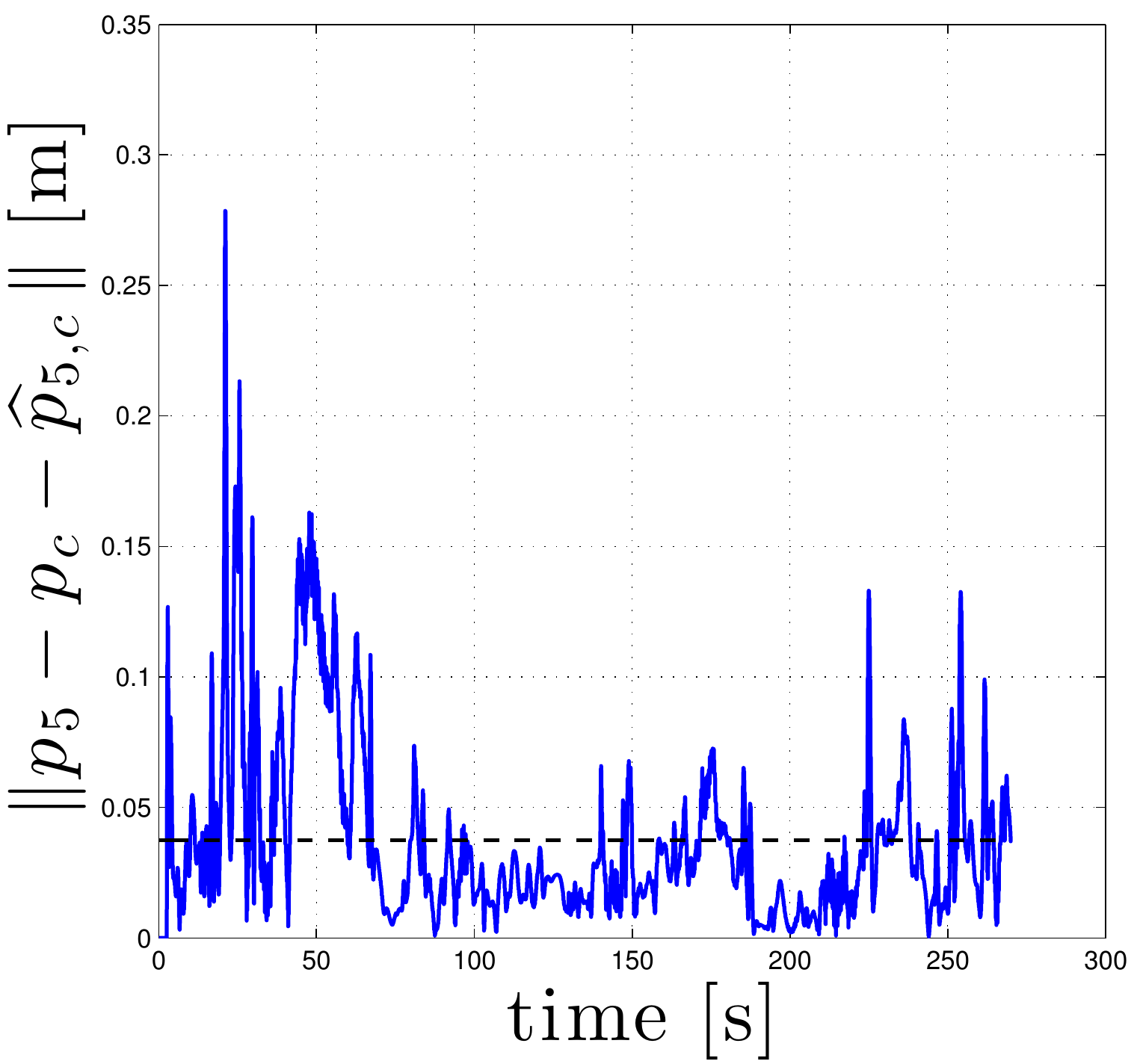}}\\
\subfigure[]{\includegraphics[width=0.45\columnwidth]{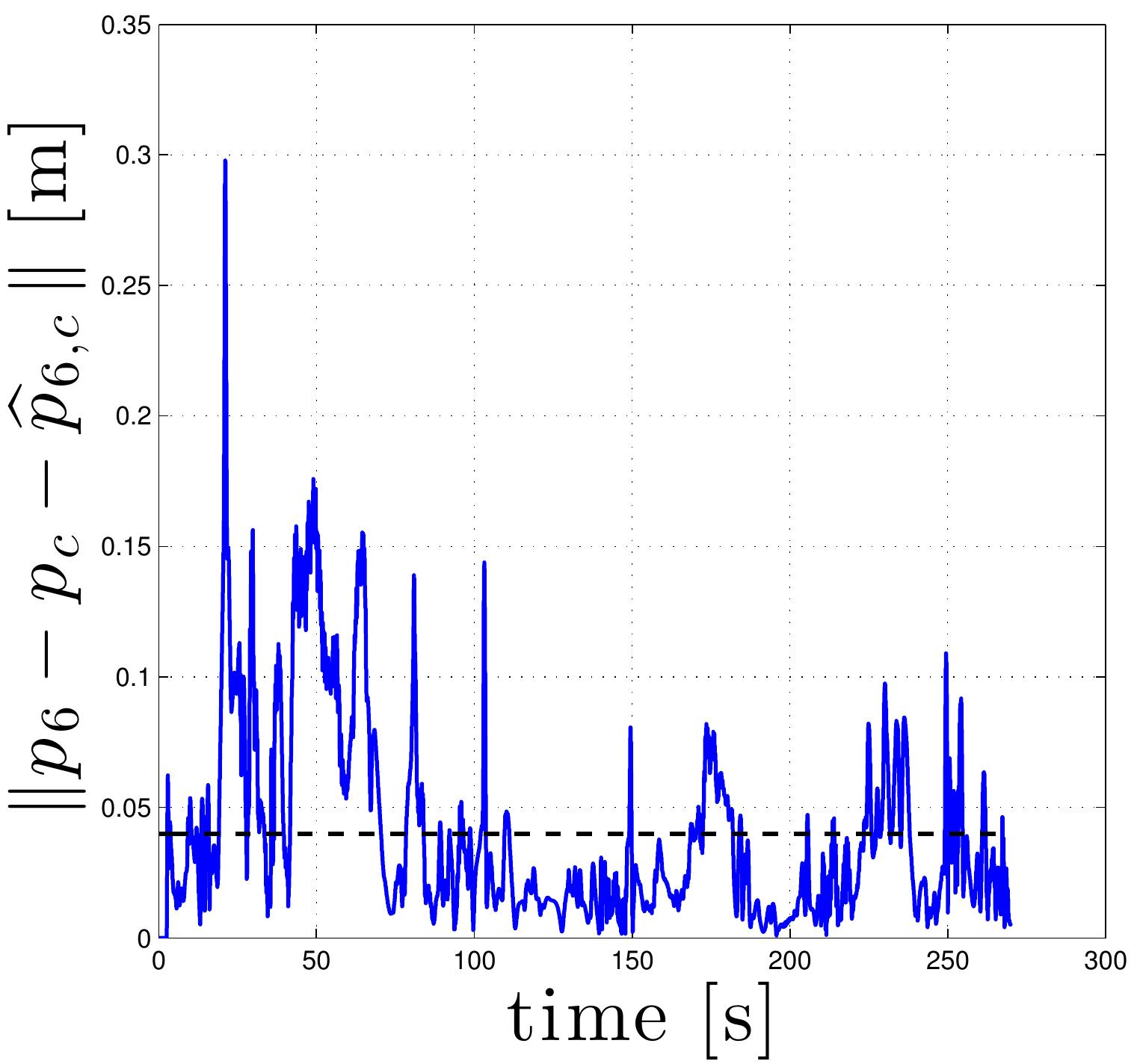}}
  \caption{Behavior of $\|p_i- p_c-\hat p_{i,c}\|$, $i=2\ldots 6$, the norm of the estimation error for the relative positions of agents $2\ldots 6$ w.r.t. agent $i_c=1$. The horizontal dashed black line represents the mean value of each error norm over time. Note how the estimation errors keep a low value during the group motion and thus indicate the ability of each robot to recover its relative position with respect to the robot $i_c=1$ by only exploiting measured distances with respect to its neighbors and the infinitesimal rigidity of the formation}\label{fig:e_pos_exp2}
\end{center}
\end{figure}

As for the rigidity eigenvalue estimation of Sect.~\ref{sec:rigidity_est}, Fig.~\ref{fig:lambda_est_exp2}(a) reports the behavior of $\lambda_7(t)$ (solid blue line), of the $6$ estimations $\hat \lambda^i_7(t)$ (solid colored lines almost superimposed to $\lambda_7(t)$), and of the minimum threshold $\lambda^{\mathrm{min}}_7=7.5$ (horizontal dashed line). From the plot one can verify: $(i)$ the accuracy in recovering the value of $\lambda_7(t)$ (note how the $6$ estimations are almost superimposed on the real value) and $(ii)$ that $\lambda_7(t)>\lambda^{\mathrm{min}}_7$ {at all times apart from few isolated spikes}, implying that \emph{formation rigidity} was maintained during the task execution. As an additional indication of the eigenvalue estimation performance, Fig.~\ref{fig:lambda_est_exp2}(b) shows the total estimation error for the rigidity eigenvalue
\begin{equation}\label{eq:e_lambda}
e_\lambda(t)=\dfrac{\sum^N_{i=1}|\lambda_7(t) - \hat\lambda^i_7(t)|}{N}
\end{equation}
which again confirms the accuracy of the estimation strategy.
\begin{figure}[!t]
\begin{center}
\subfigure[]{\includegraphics[width=0.51\columnwidth]{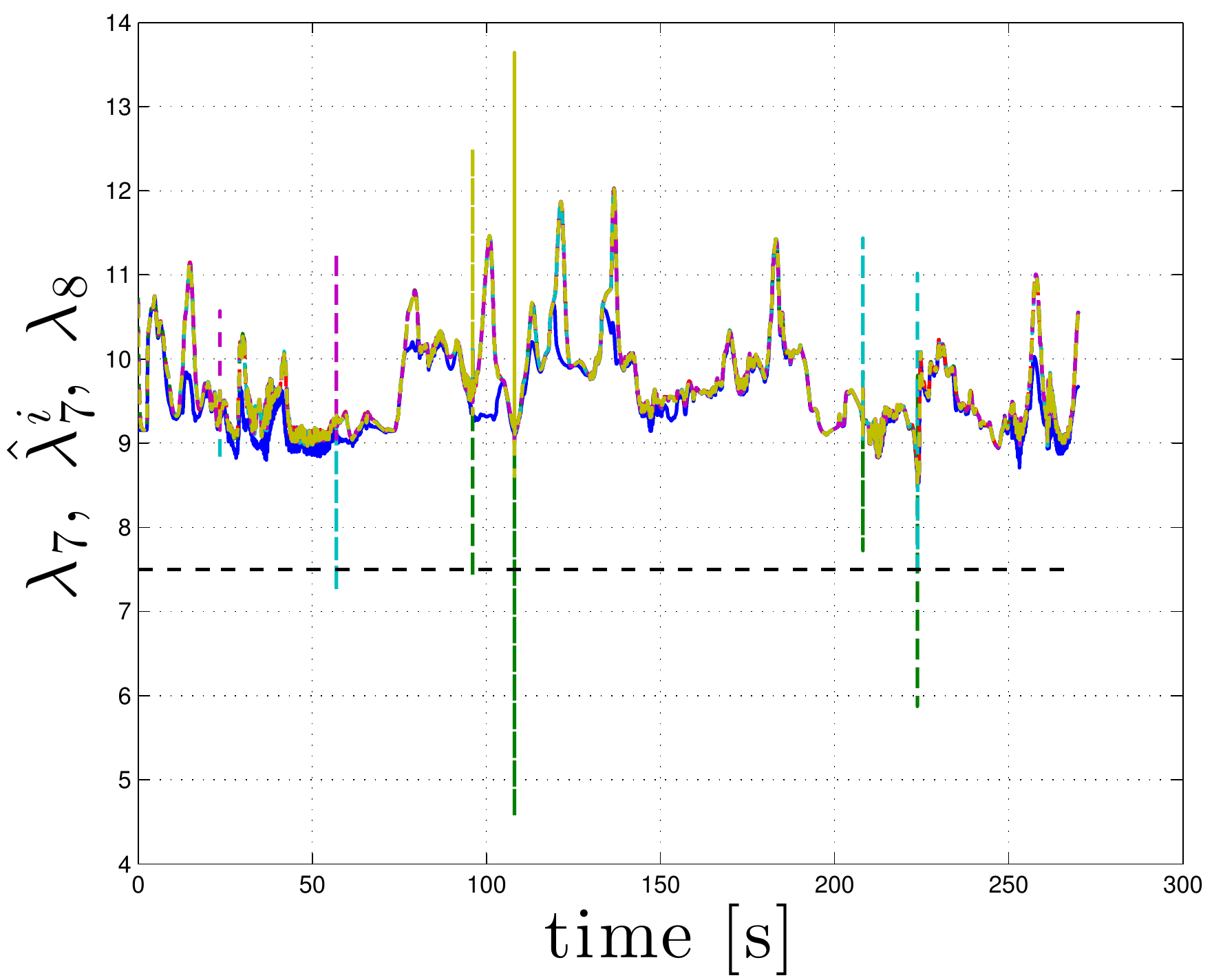}}
\subfigure[]{\includegraphics[width=0.45\columnwidth]{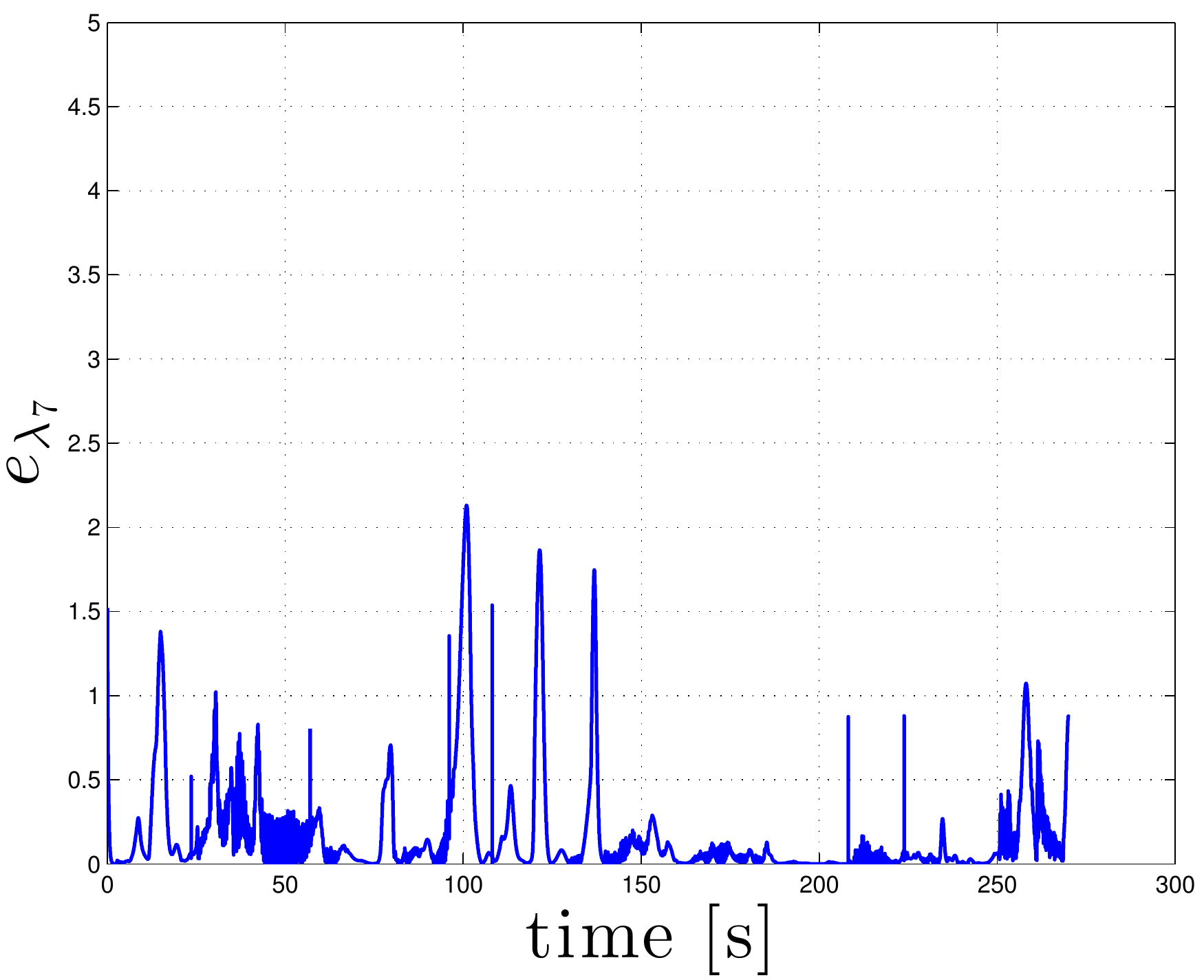}}
  \caption{Left: behavior of $\lambda_7(t)$ (blue line) and the $6$ estimations $\hat\lambda^i_7(t)$ (dashed colored lines) which result almost coincident. Right: behavior of the overall rigidity eigenvalue estimation error $e_\lambda(t)$ as defined in~(\ref{eq:e_lambda})}\label{fig:lambda_est_exp2}
\end{center}
\end{figure}

\begin{figure}[!t]
\begin{center}
\subfigure[]{\includegraphics[width=0.32\columnwidth]{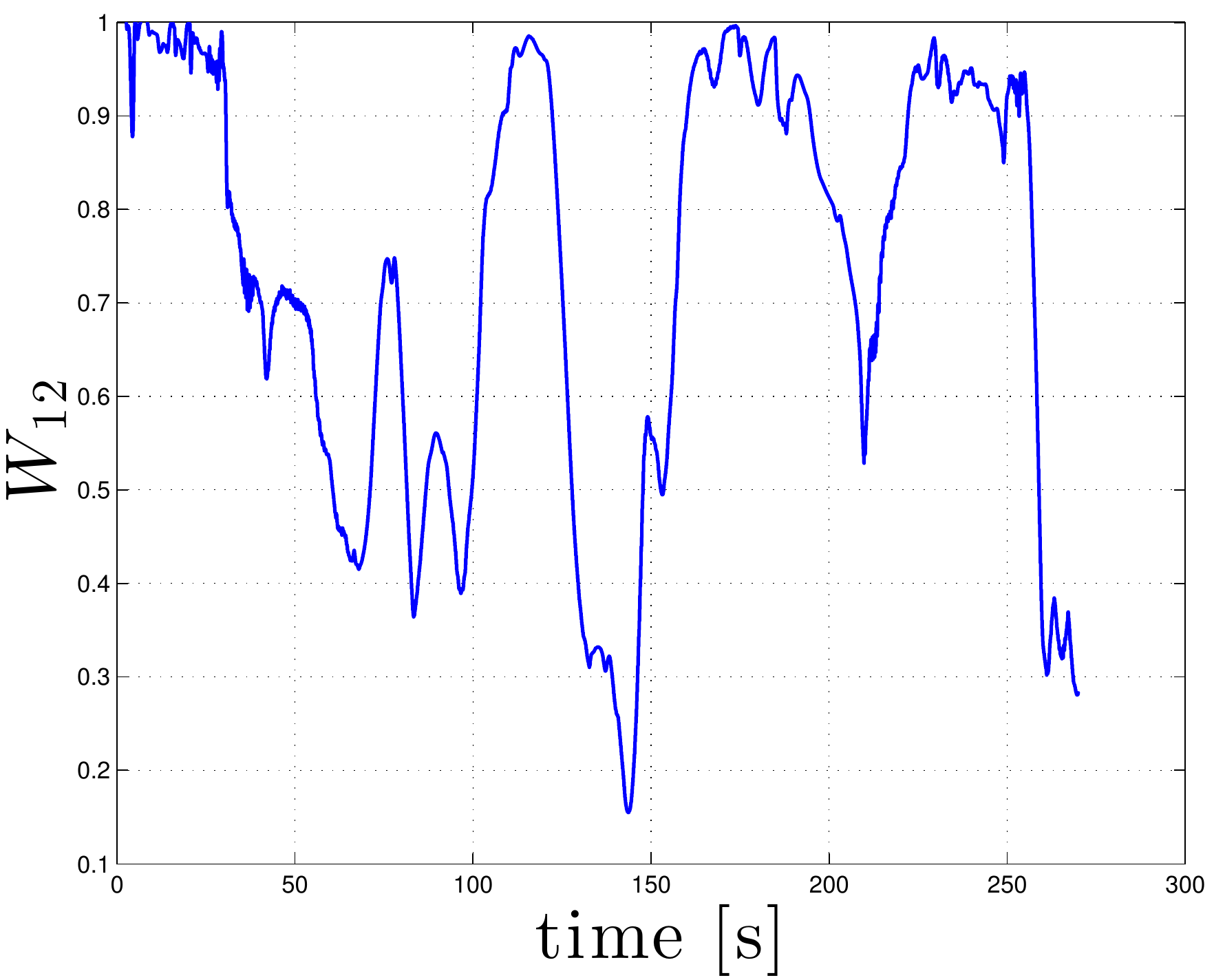}}
\subfigure[]{\includegraphics[width=0.32\columnwidth]{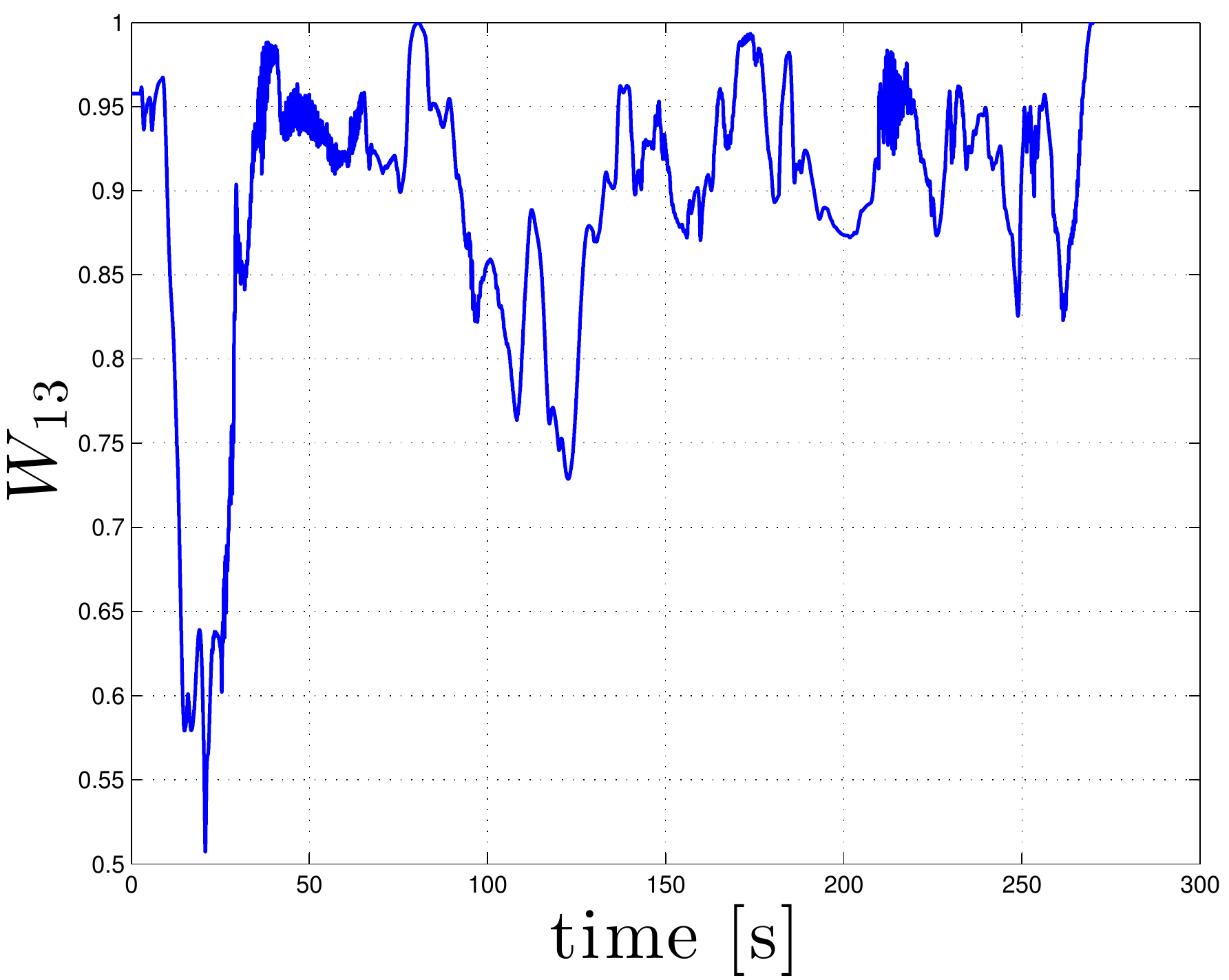}}
\subfigure[]{\includegraphics[width=0.32\columnwidth]{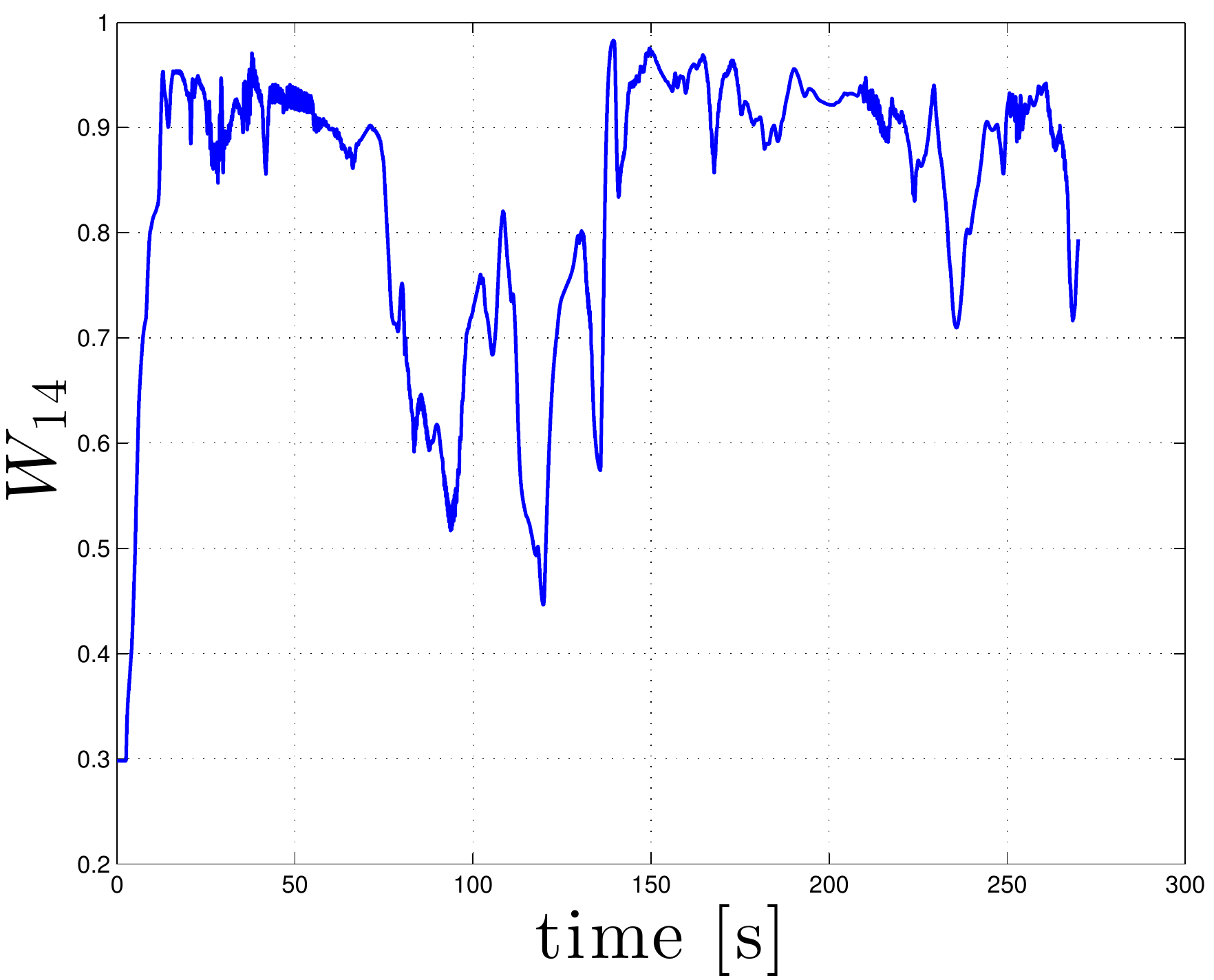}}\\
\subfigure[]{\includegraphics[width=0.32\columnwidth]{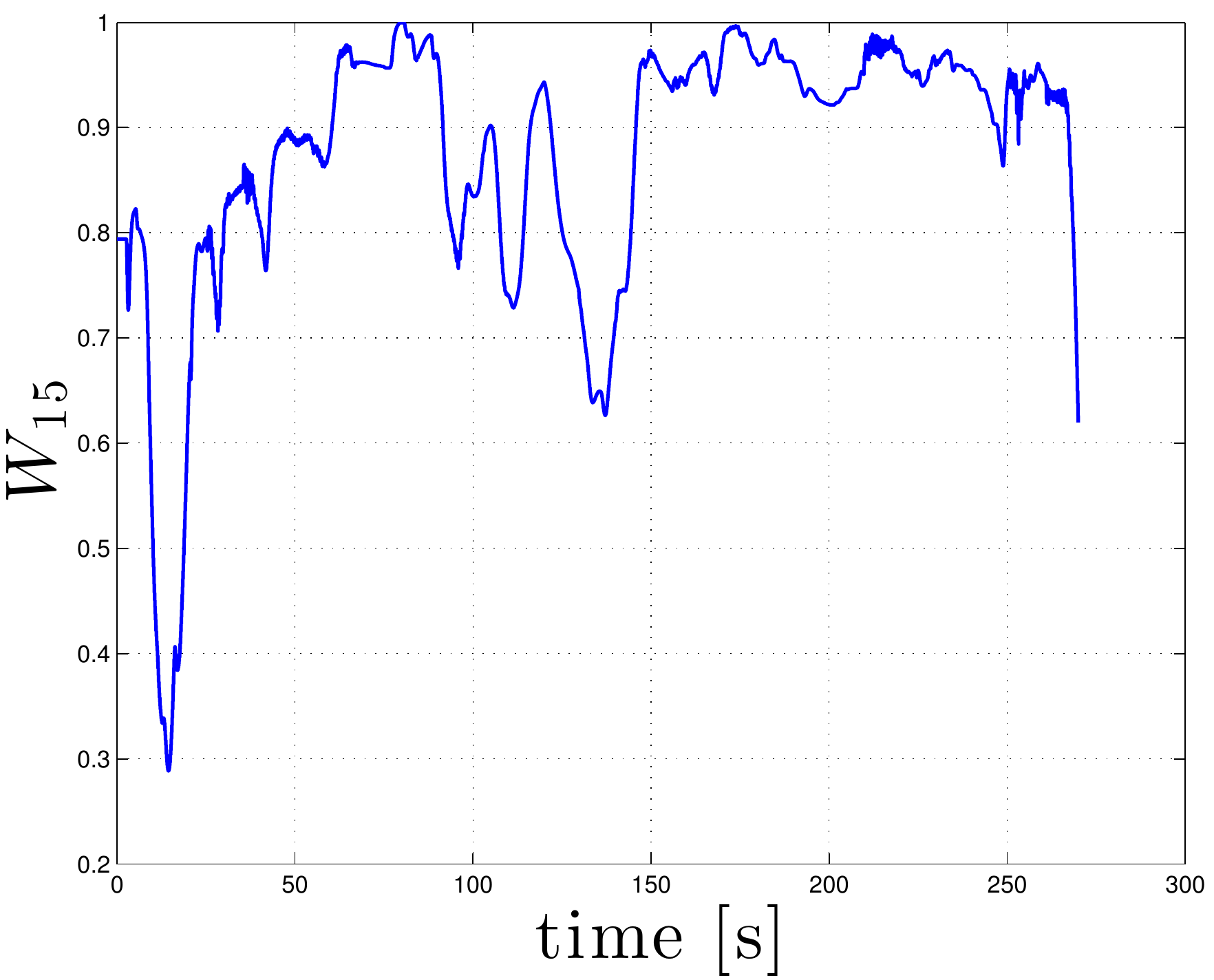}}
\subfigure[]{\includegraphics[width=0.32\columnwidth]{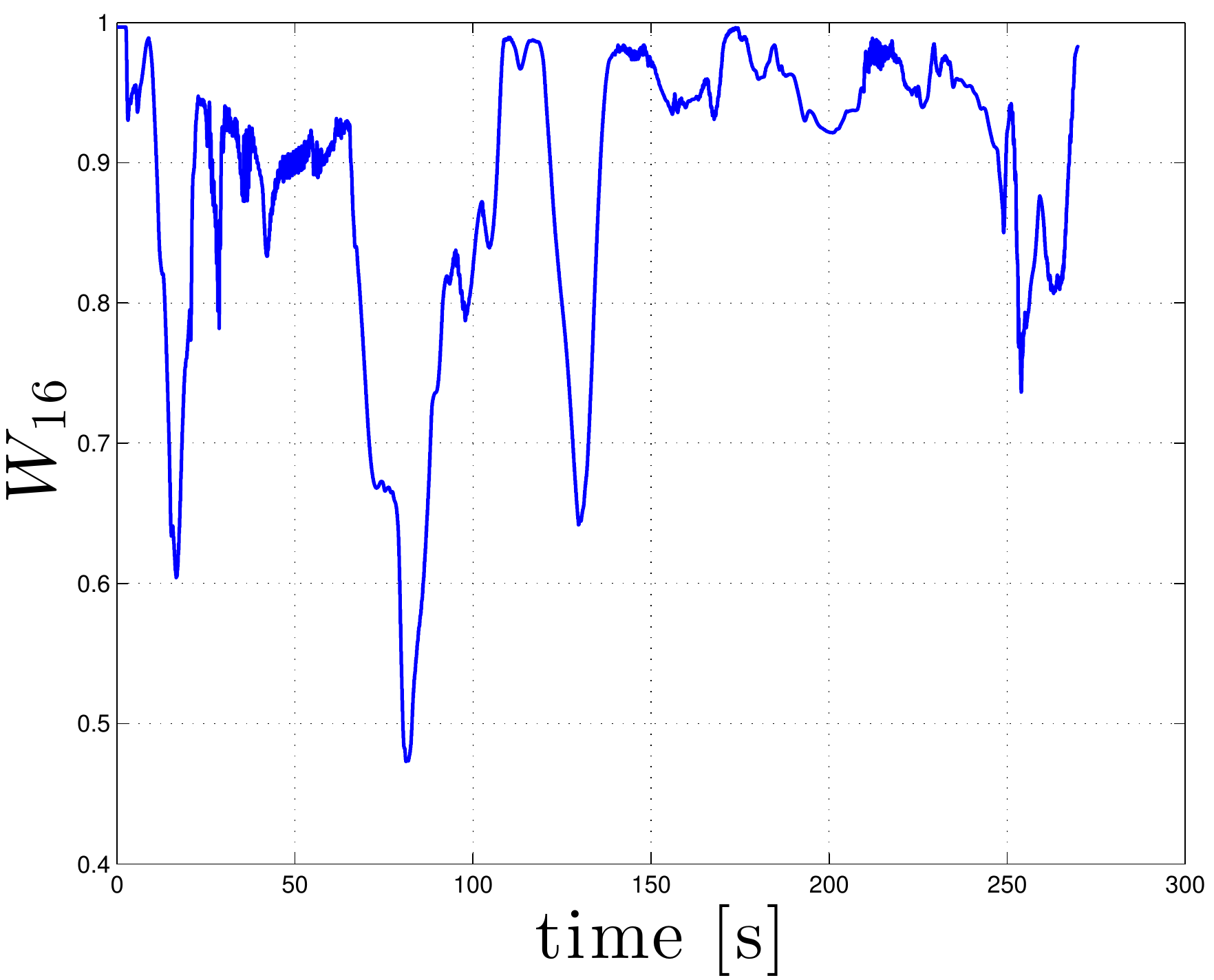}}
\subfigure[]{\includegraphics[width=0.32\columnwidth]{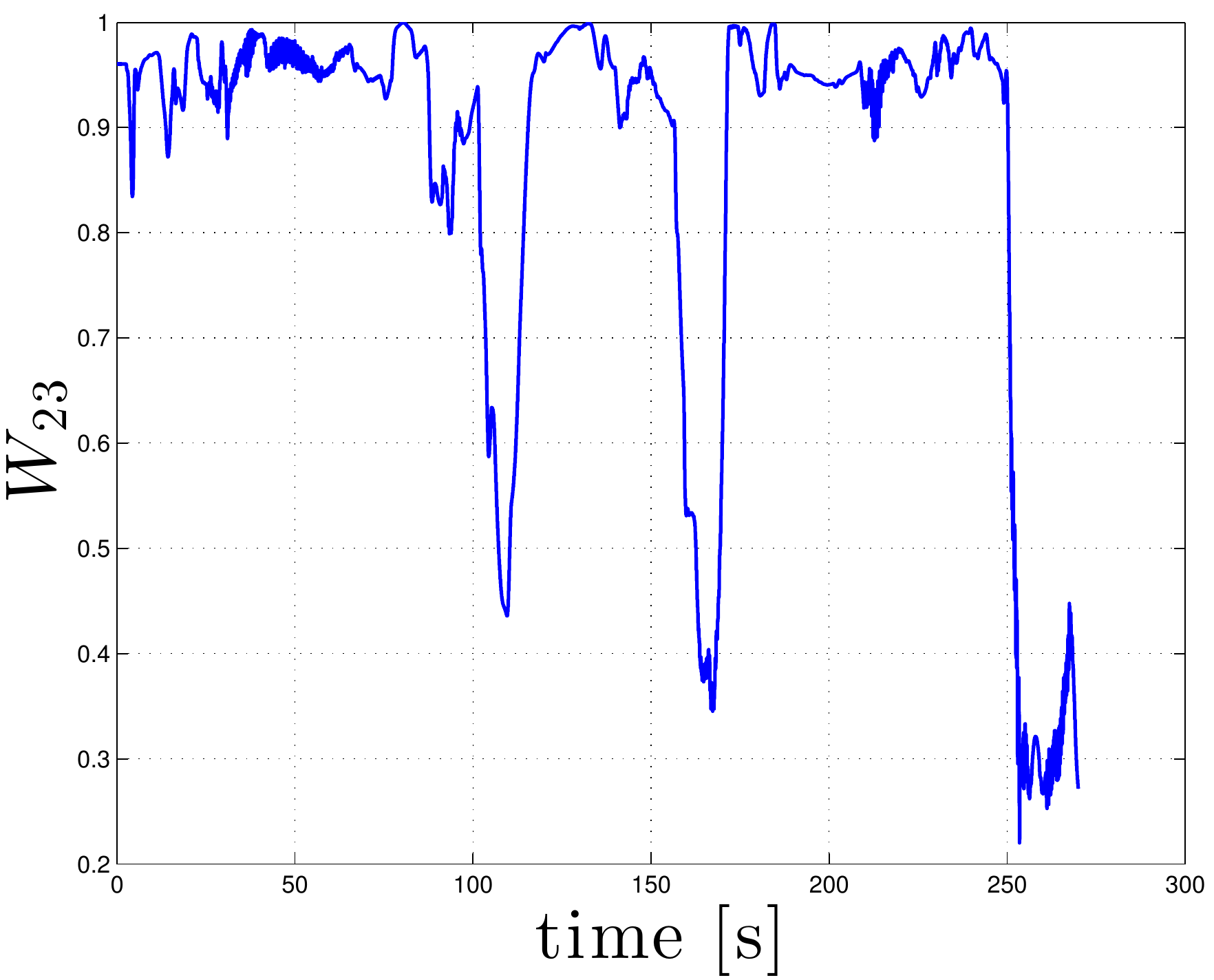}}\\
\subfigure[]{\includegraphics[width=0.32\columnwidth]{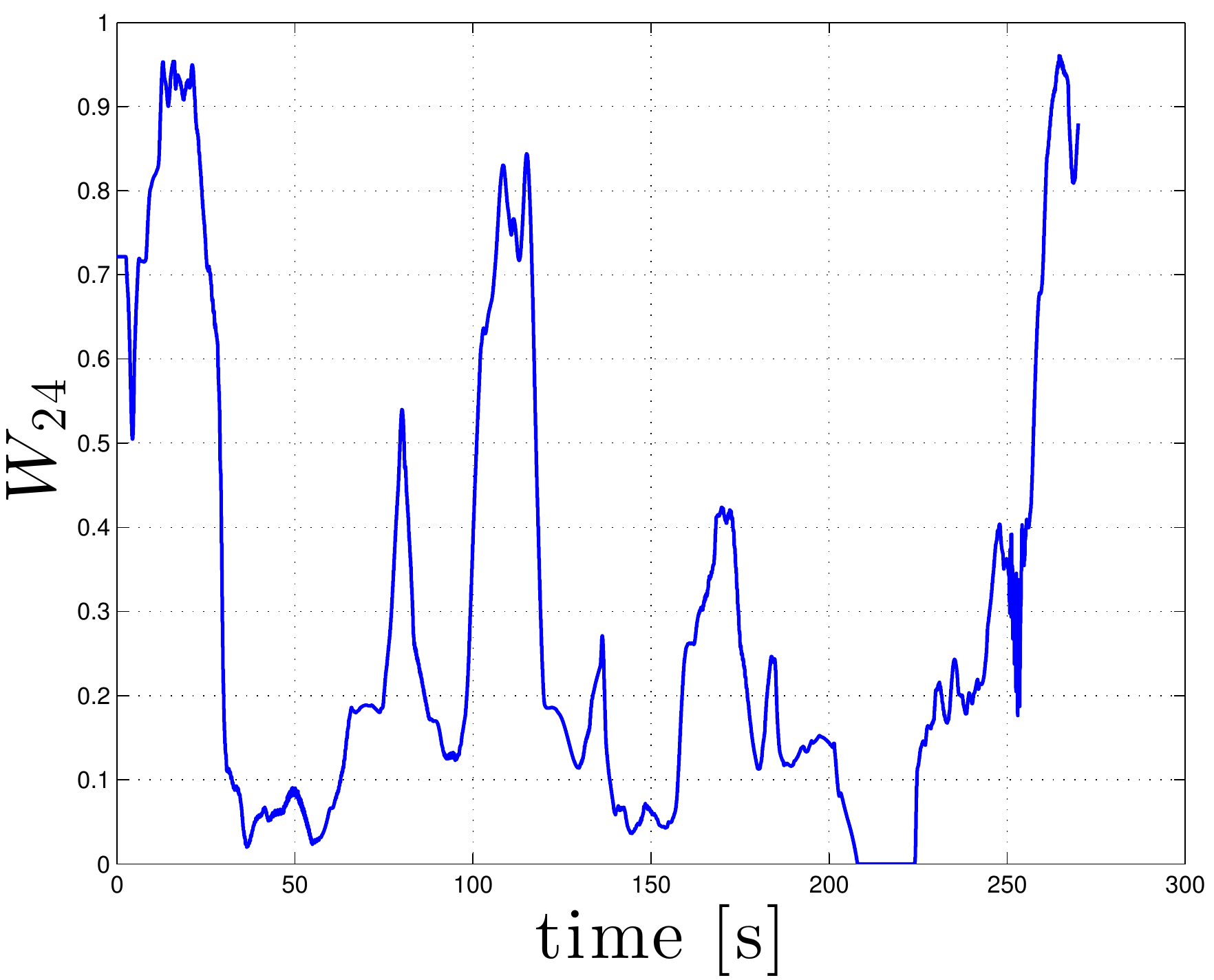}}
\subfigure[]{\includegraphics[width=0.32\columnwidth]{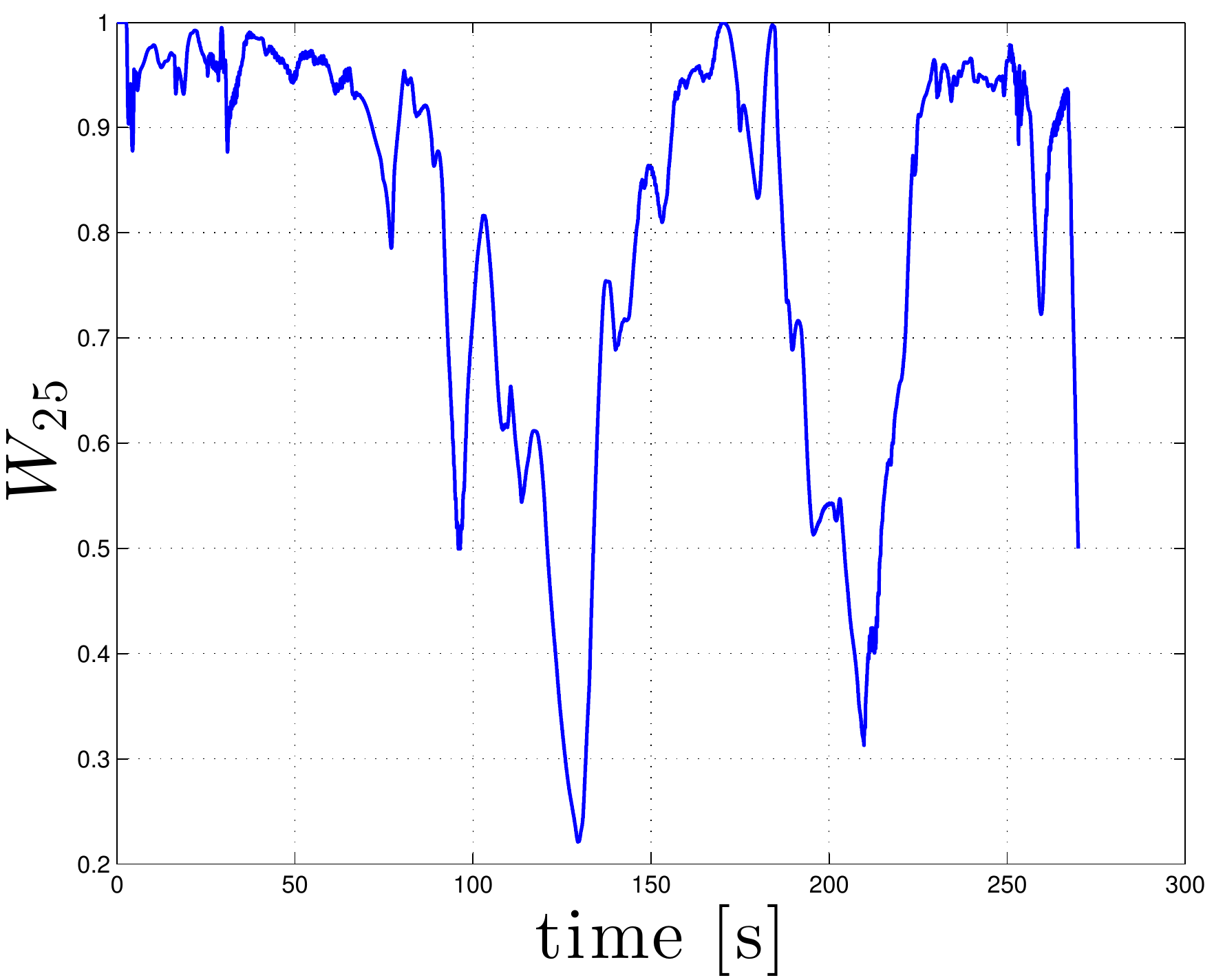}}
\subfigure[]{\includegraphics[width=0.32\columnwidth]{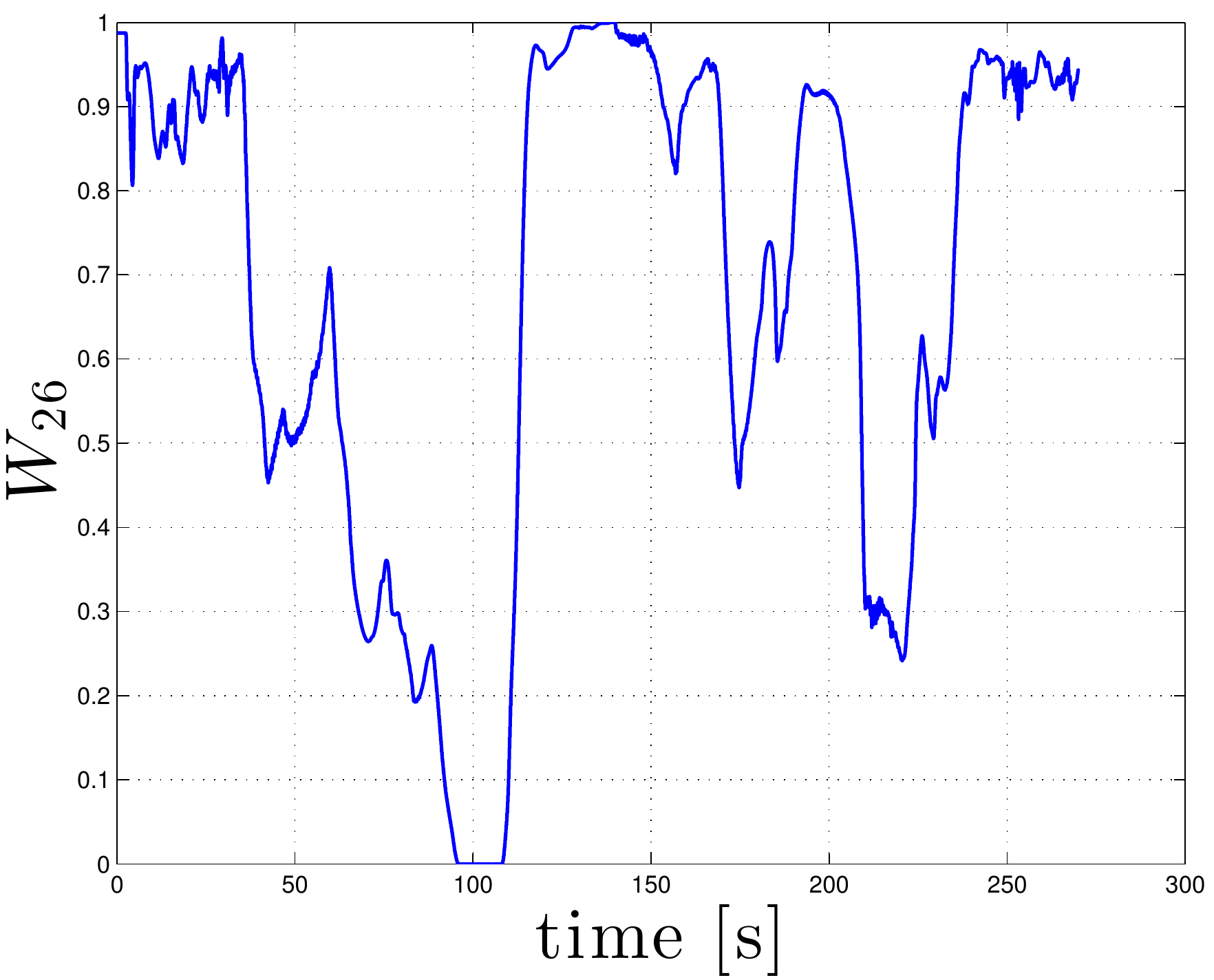}}\\
\subfigure[]{\includegraphics[width=0.32\columnwidth]{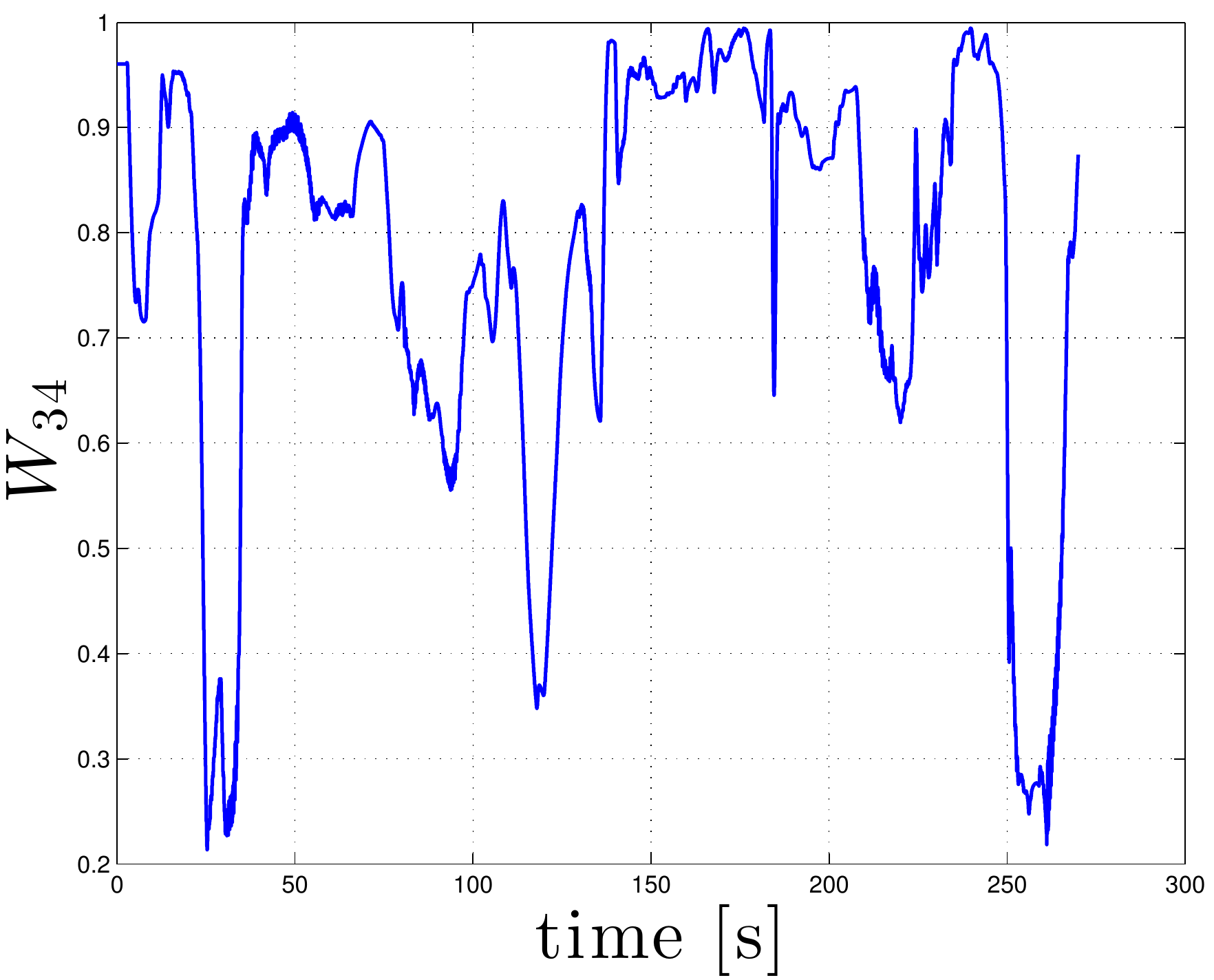}}
\subfigure[]{\includegraphics[width=0.32\columnwidth]{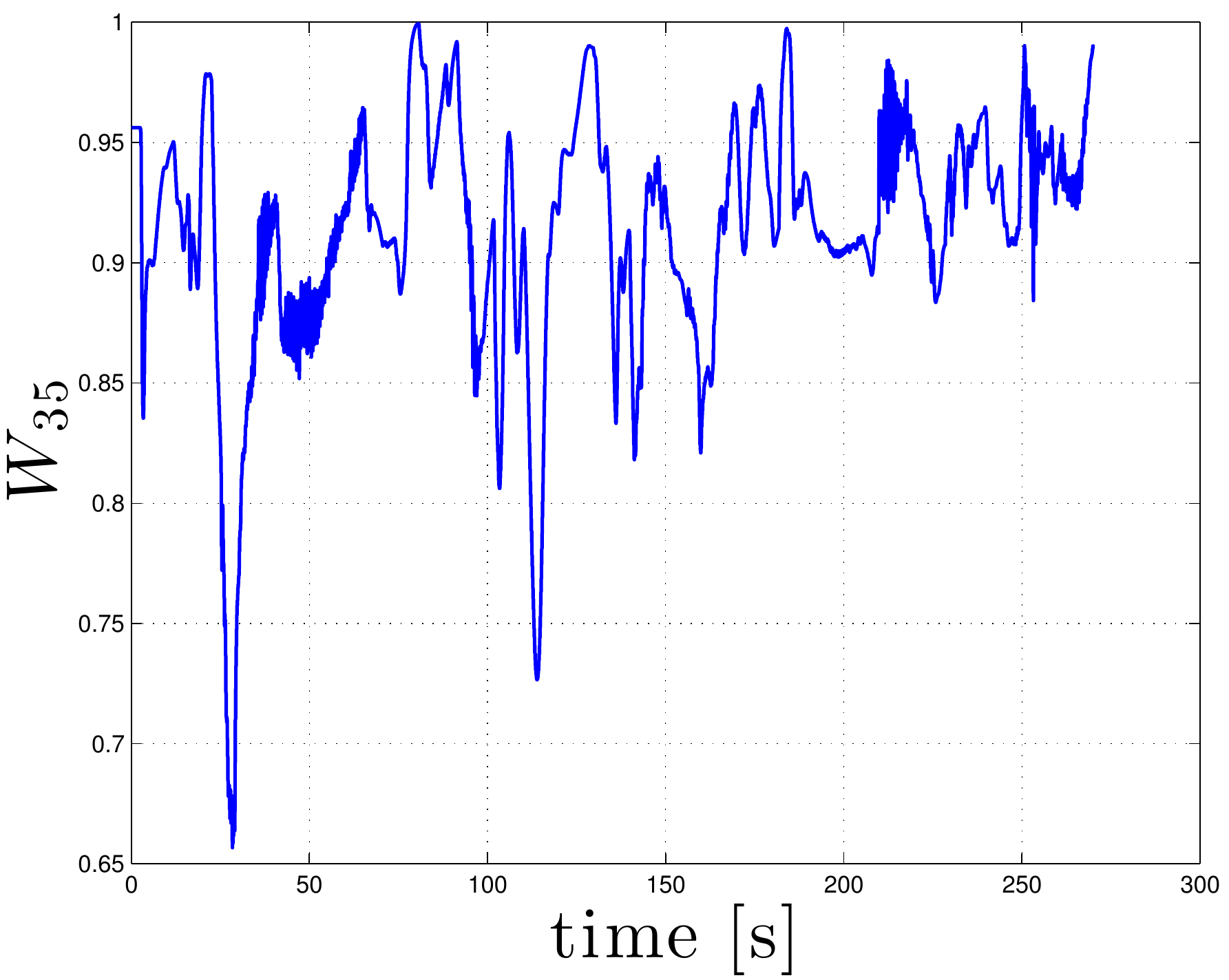}}
\subfigure[]{\includegraphics[width=0.32\columnwidth]{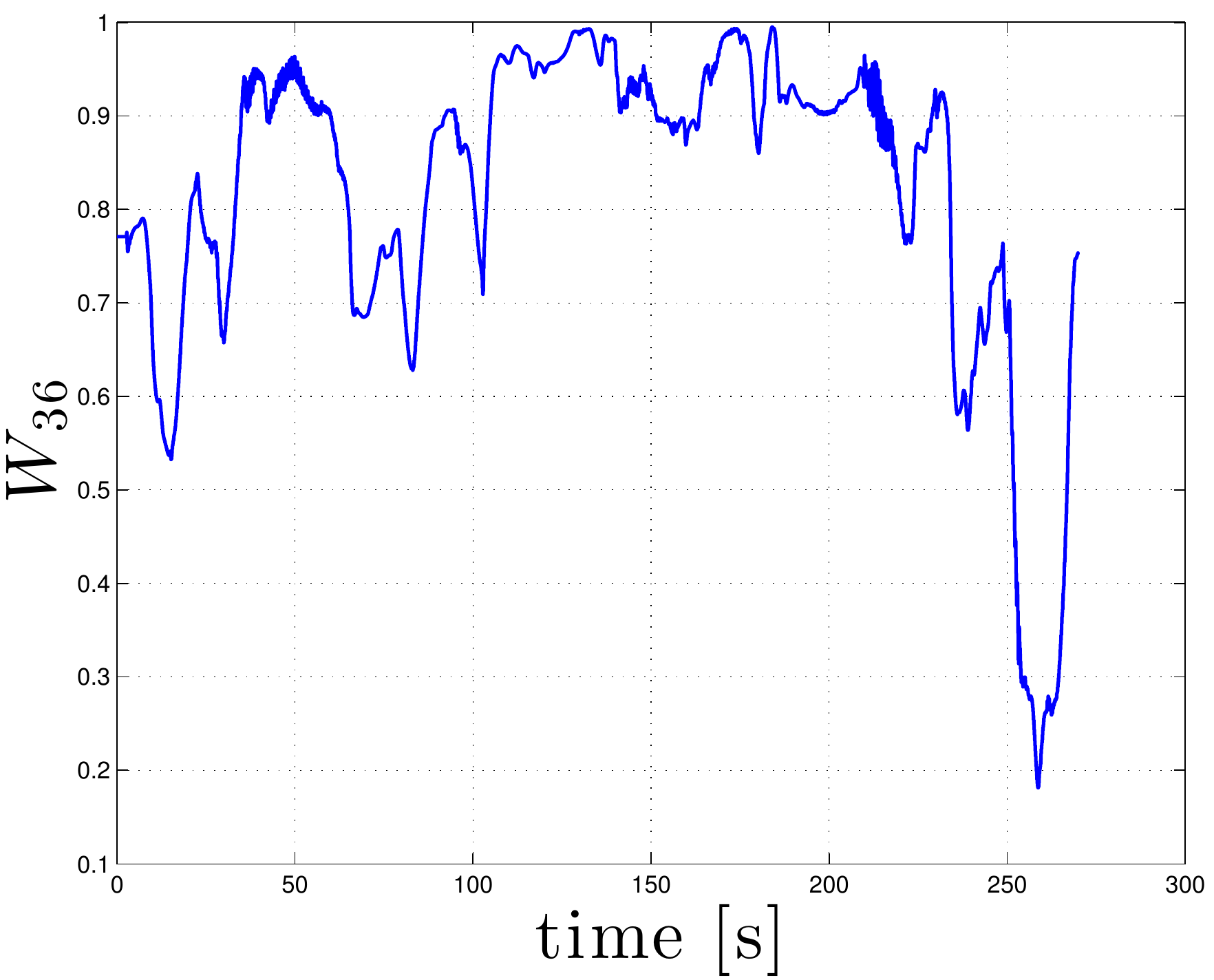}}\\
\subfigure[]{\includegraphics[width=0.32\columnwidth]{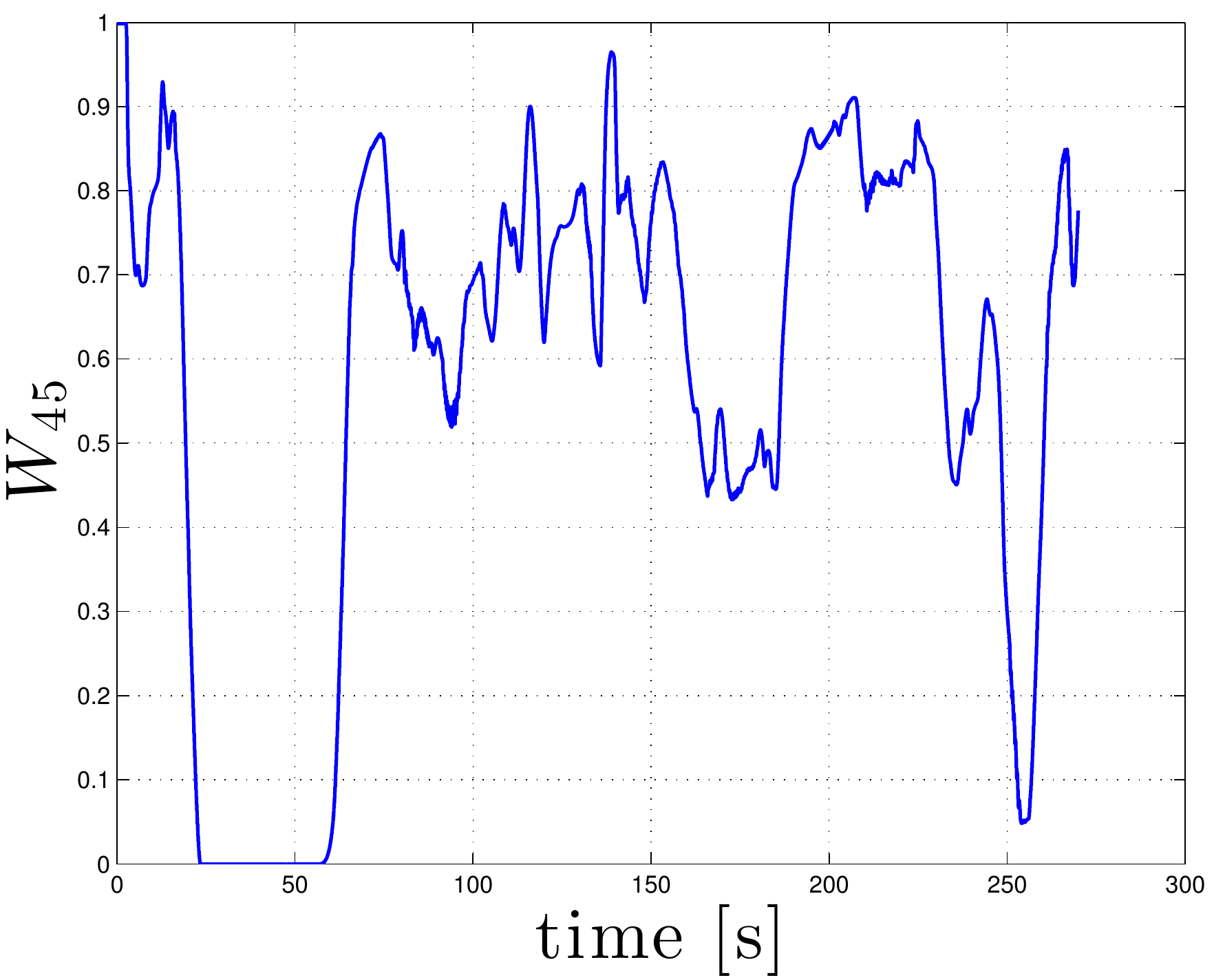}}
\subfigure[]{\includegraphics[width=0.32\columnwidth]{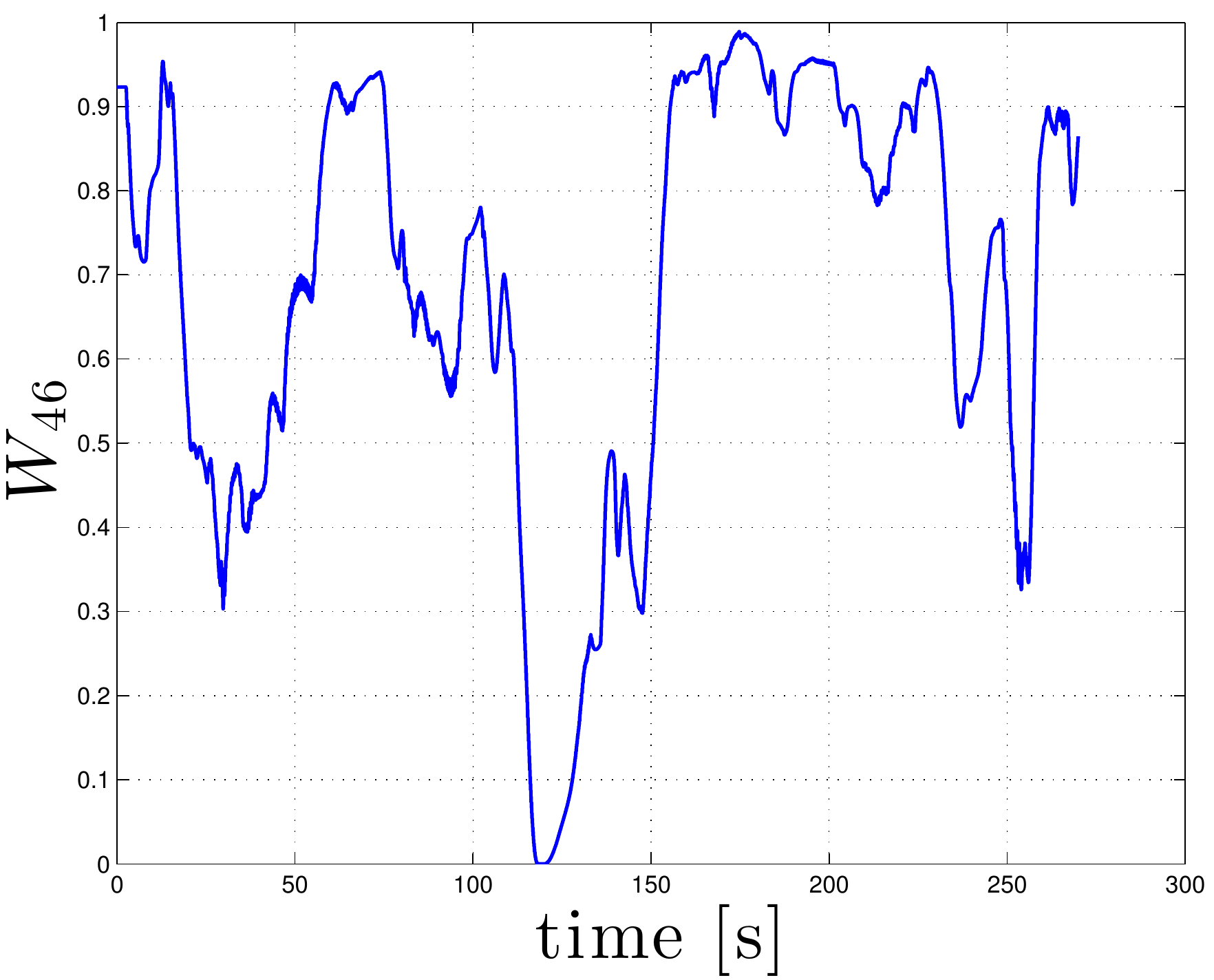}}
\subfigure[]{\includegraphics[width=0.32\columnwidth]{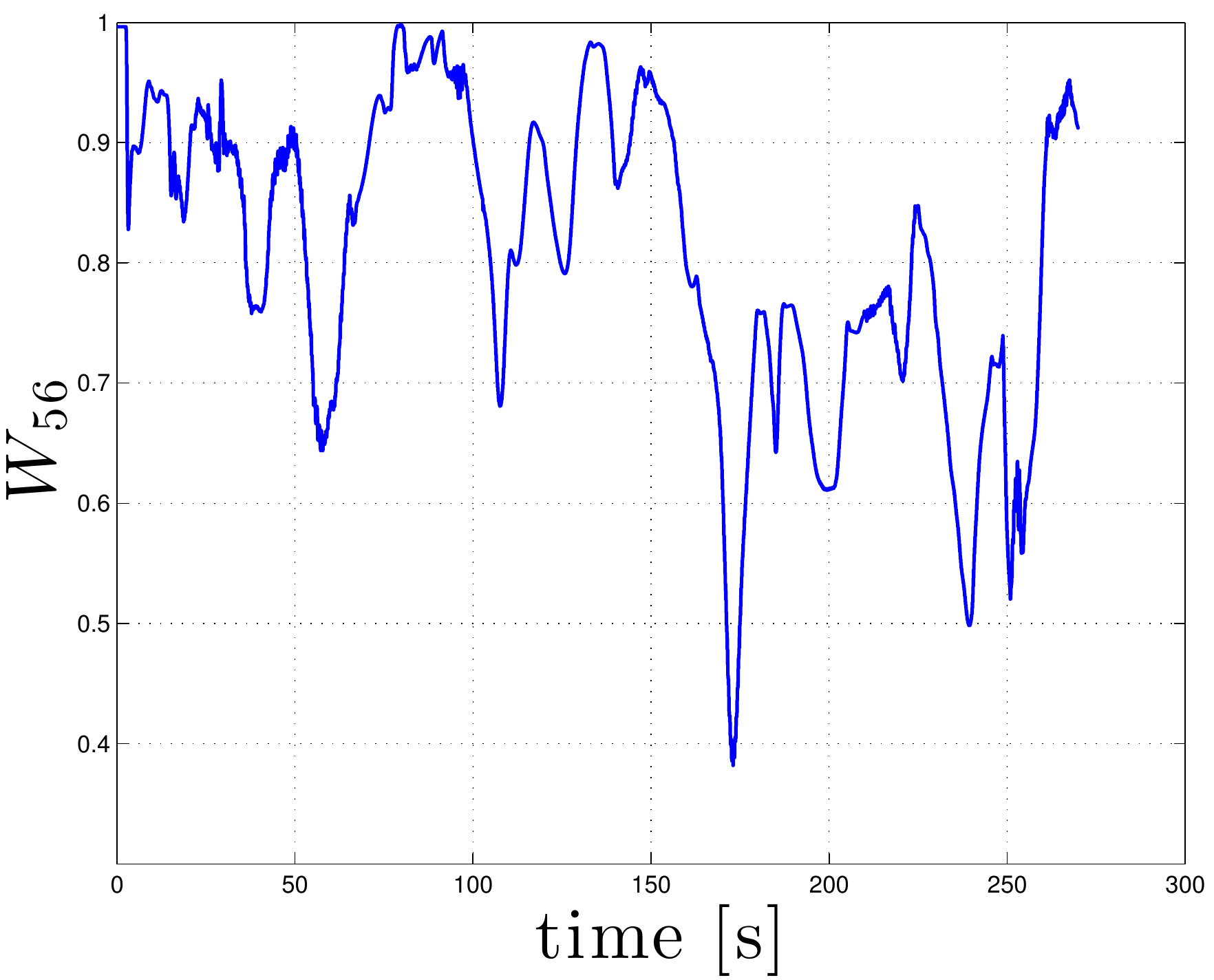}}
 \caption{Behavior of the $15$ weights $W_{uv}(t)$ for all the possible edges of graph $\calG$. Note how the values of weights $W_{uv}(t)$ vary over time because of the sensing/communication constraints and requirements embedded within their definition (see sec.~\ref{subsec:weights}). Some weights (e.g., $W_{24}$ and $W_{45}$) also temporarily vanish indicating loss of the corresponding edge (and, thus, the time-varying nature of graph $\calG$)}\label{fig:W_exp2}
\end{center}
\end{figure}
Figures~\ref{fig:W_exp2}(a--o) report the behavior of the $15$ weights $W_{uv}$ defined in~(\ref{eq:weights}) and associated to the all the possible edges of graph $\calG$ in order to show their \emph{time-varying} nature because of the constraints and requirements listed in Sect.~\ref{subsec:weights}. Note how the value of some weight drops to zero over time (e.g., $W_{45}(t)$ at about $t=25$ [s] or $W_{24}(t)$ at about $t=210$ [s]), thus indicating loss of the corresponding edge. In the same spirit, Fig.~\ref{fig:num_edges} shows the total number of edges {$|\hat\calE|$ of the unweighted graph $\hat\calG$ (i.e., of non-zero weights $W_{uv}$, see Definition~\ref{def:unweighted}) during the group motion. These results highlight the time-varying nature of graph $\calG$ which, as explained in the previous sections, is not constrained to keep a given fixed topology but is free to lose or gain edges as long as infinitesimal rigidity of the formation is preserved.}
\begin{figure}[!t]
\begin{center}
{\includegraphics[width=0.45\columnwidth]{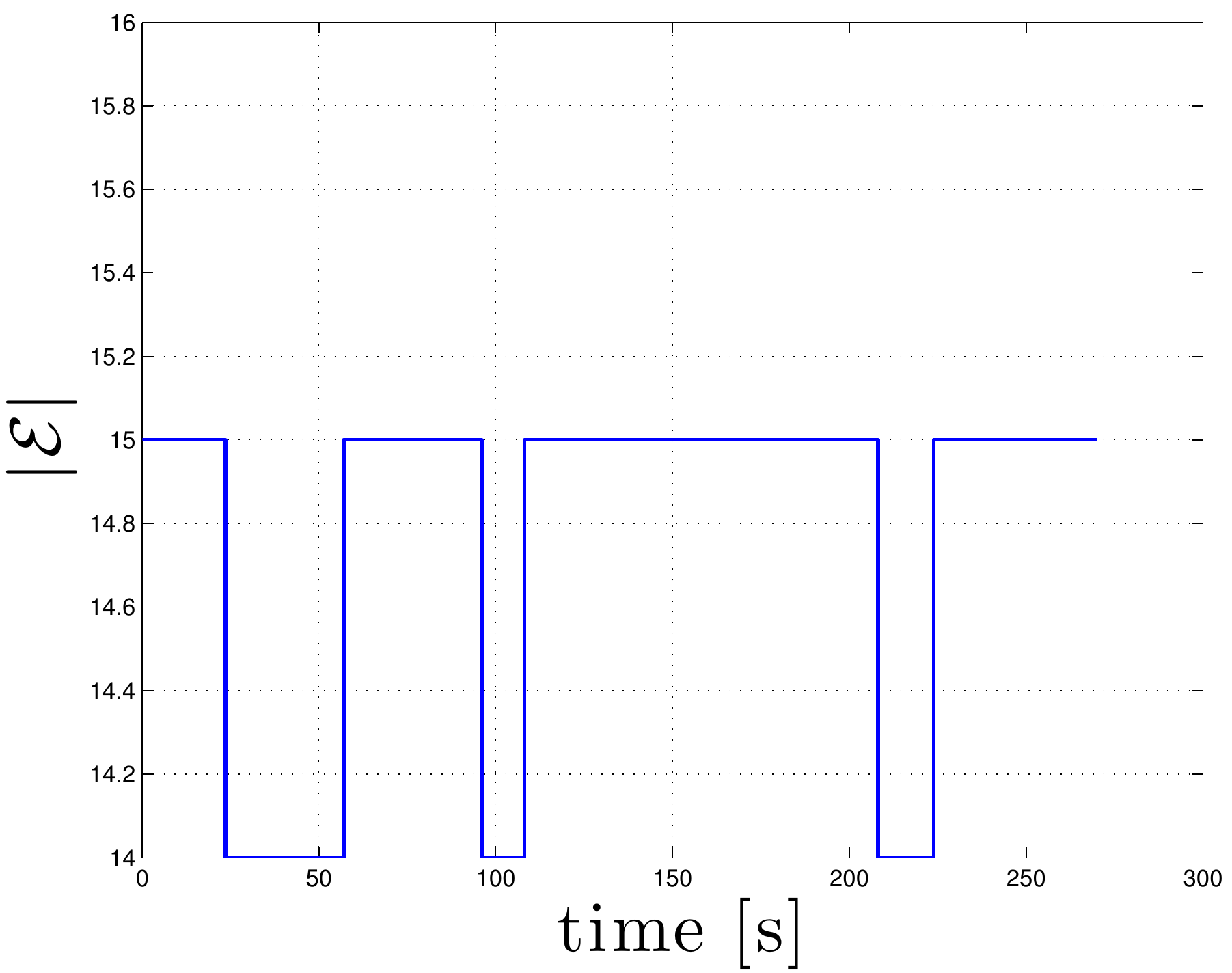}}
  \caption{Total number of edges in the graph $\calG$ during the group motion}\label{fig:num_edges}
\end{center}
\end{figure}

Finally, Figs.~\ref{fig:QC_pos_exp2}(a--f) report the behavior over time of $p_i(t)$ (the $i$-th agent position, solid lines) and of $p_{i,\mathrm{real}}(t)$ (the $i$-th quadrotor position, dashed lines) while tracking the motion of $p_i(t)$. The two position vectors result almost perfectly coincident, thus indicating a successful tracking performance of the quadrotors (and the soundness of our modeling assumptions). As a further confirmation of this fact, the norm of the overall tracking error defined as
\begin{equation}\label{eq:track_error}
e_{\mathrm{track}}(t)=\dfrac{\sum^N_{i=1}\|p_i(t) - p_{i,real}(t)\|}{N}
\end{equation}
is also reported in Fig.~\ref{fig:e_track_err_exp2}.
\begin{figure}[!t]
\begin{center}
\subfigure[]{\includegraphics[width=0.45\columnwidth]{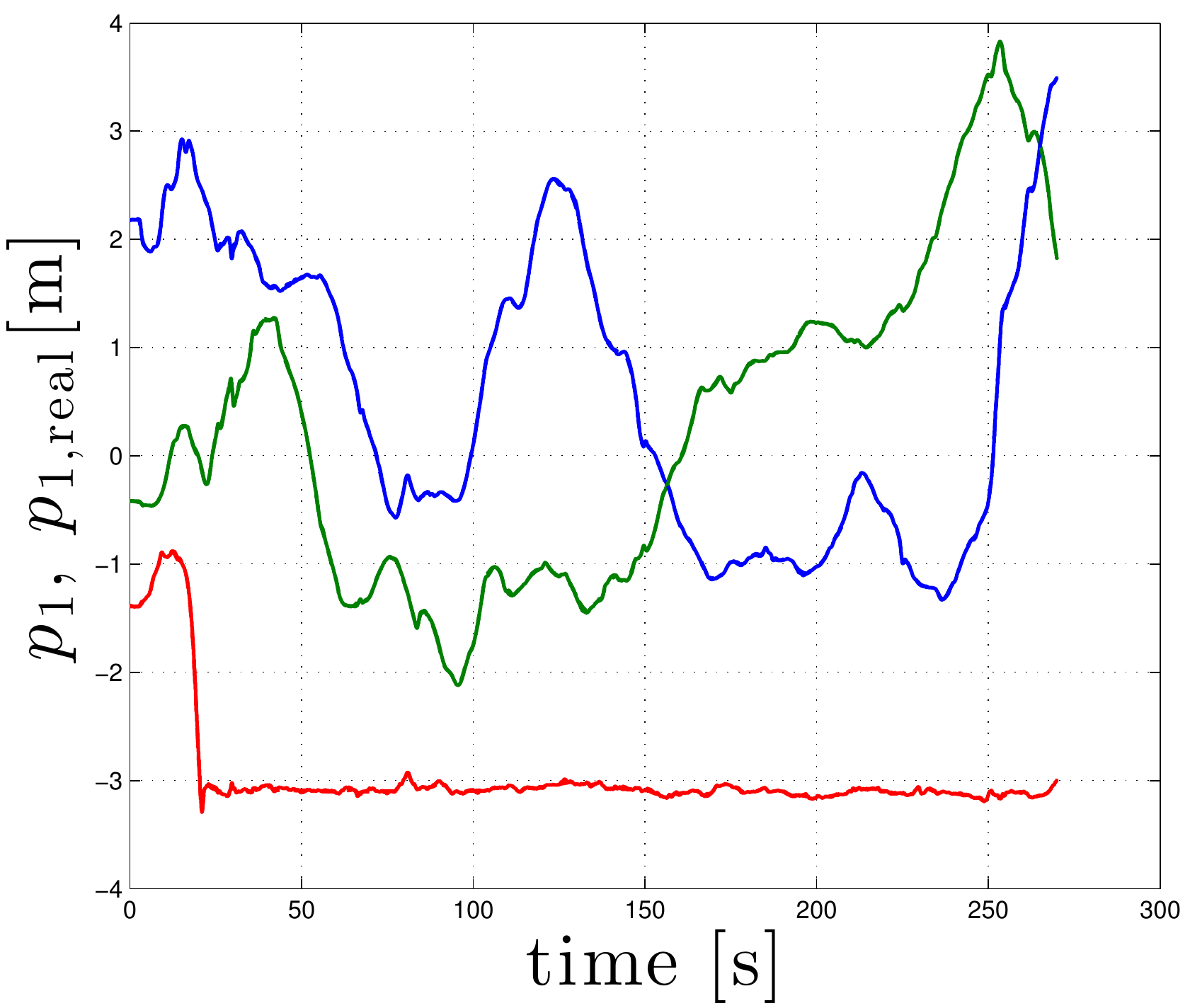}}
\subfigure[]{\includegraphics[width=0.45\columnwidth]{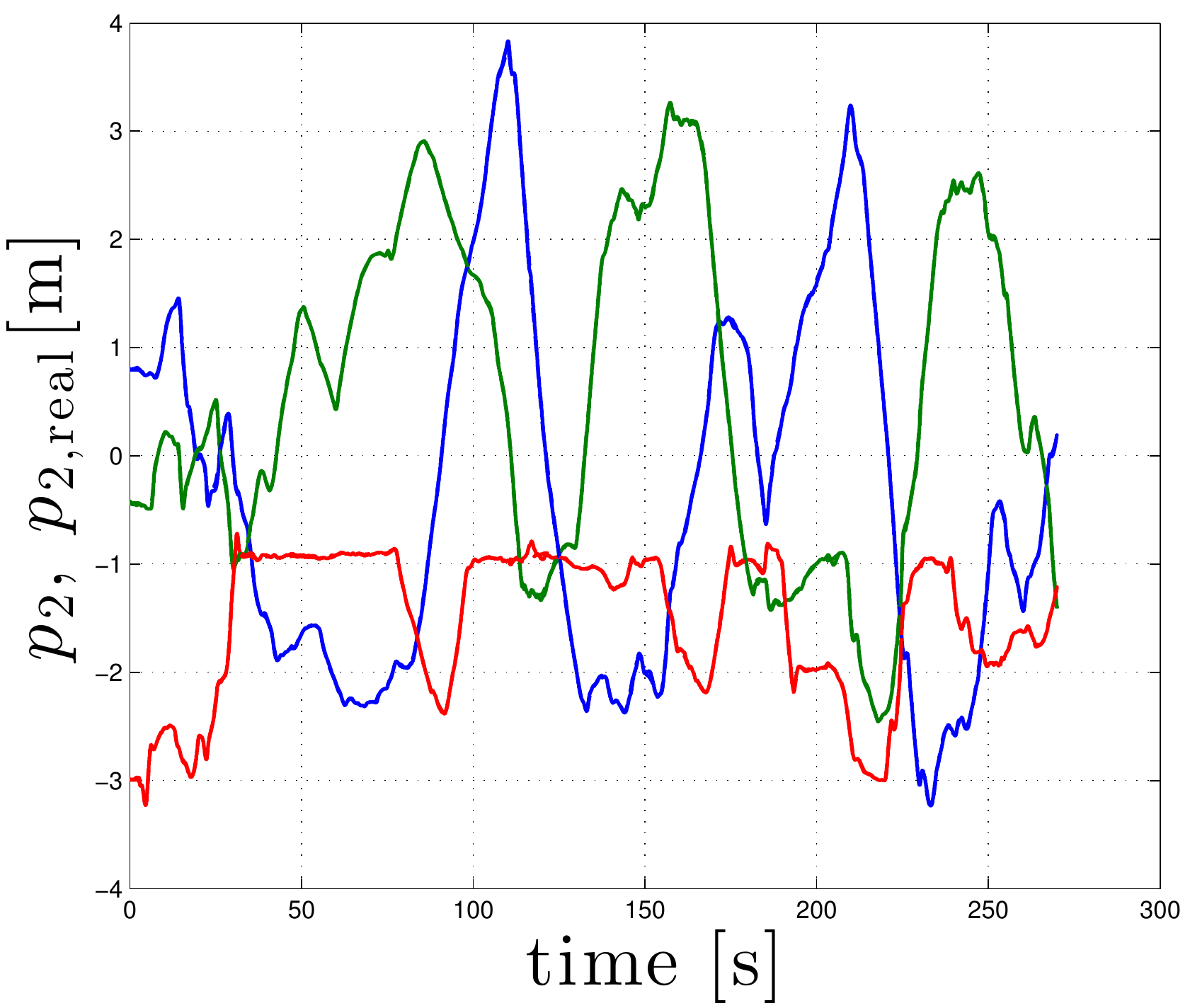}}\\
\subfigure[]{\includegraphics[width=0.45\columnwidth]{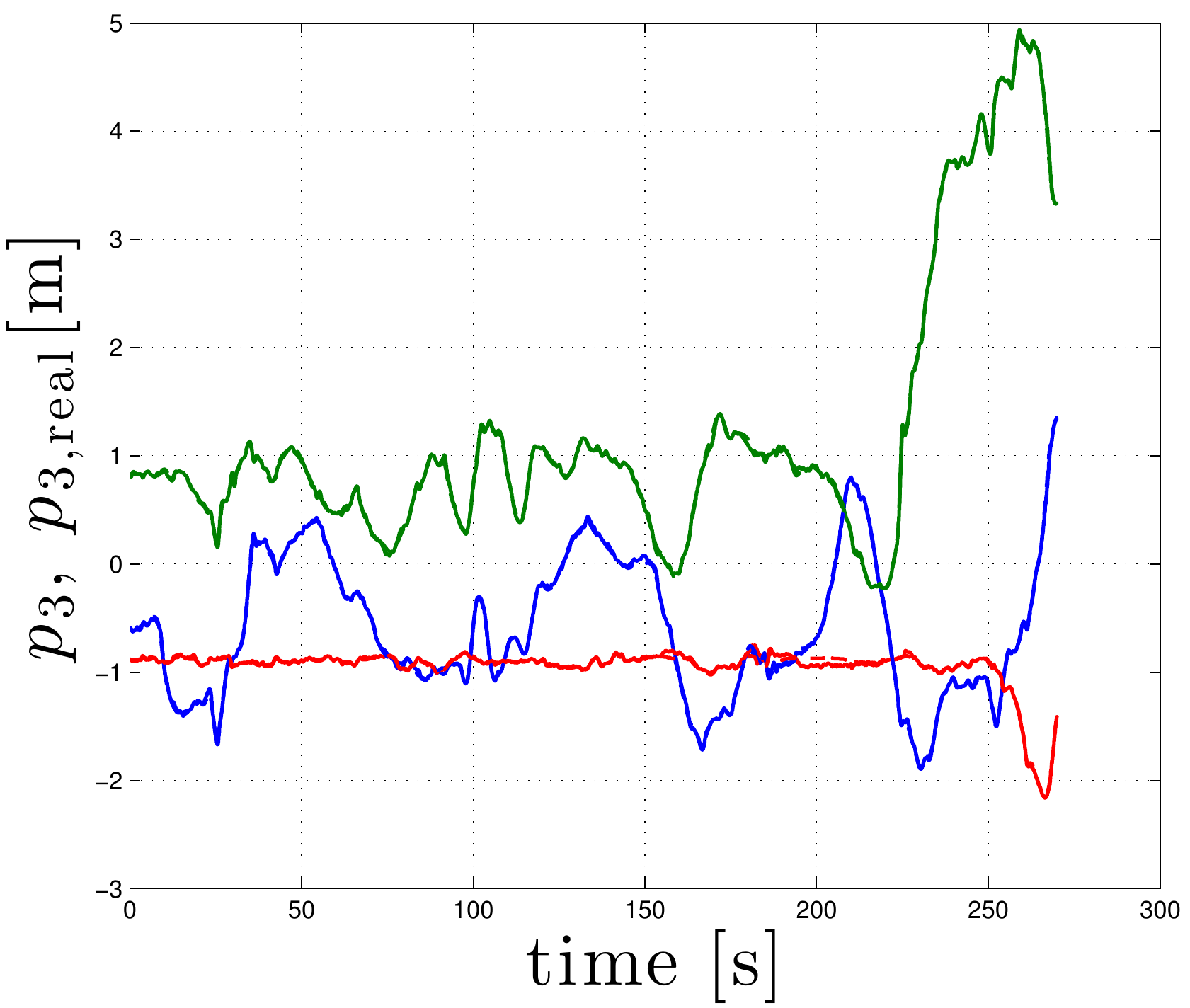}}
\subfigure[]{\includegraphics[width=0.45\columnwidth]{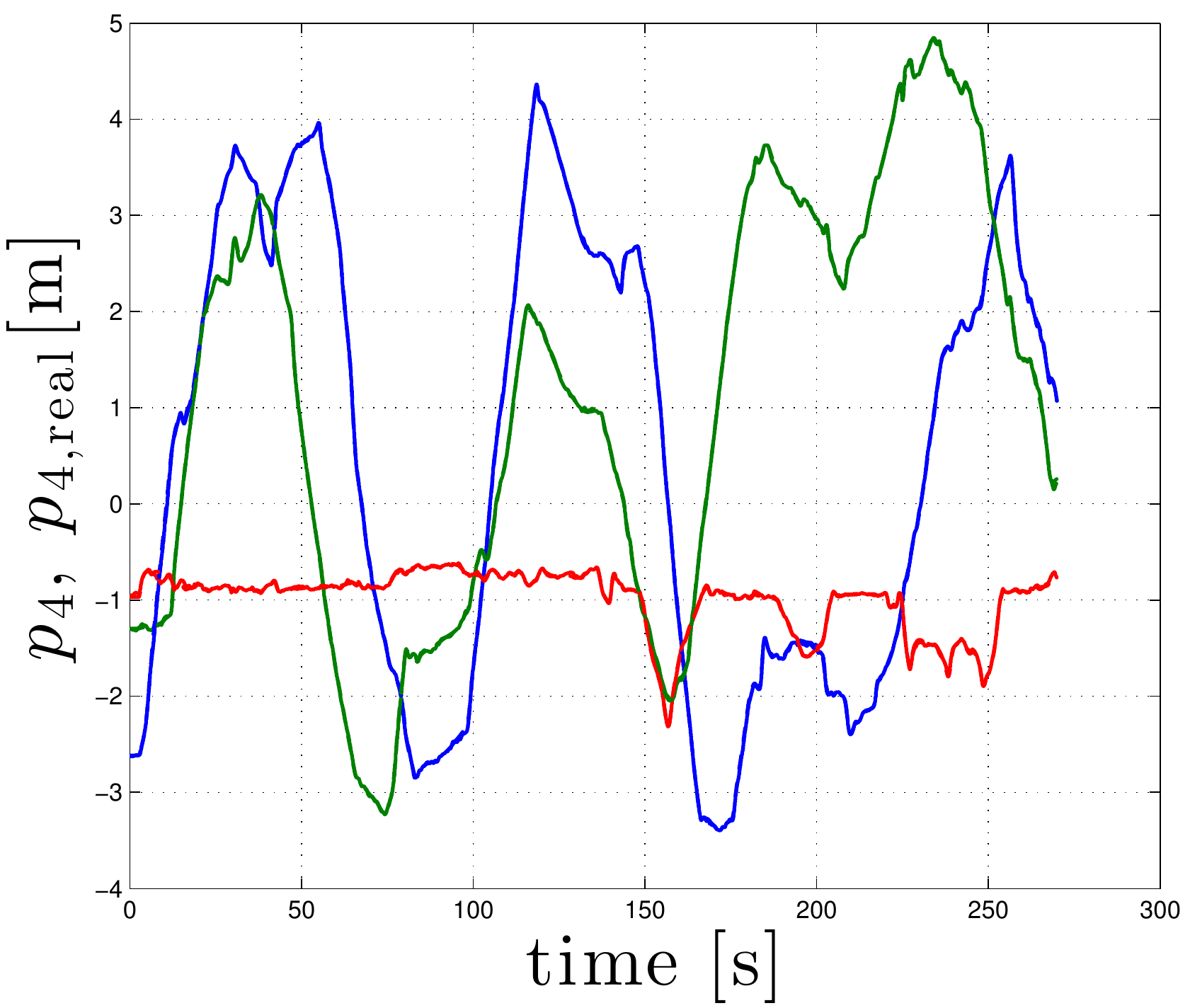}}\\
\subfigure[]{\includegraphics[width=0.45\columnwidth]{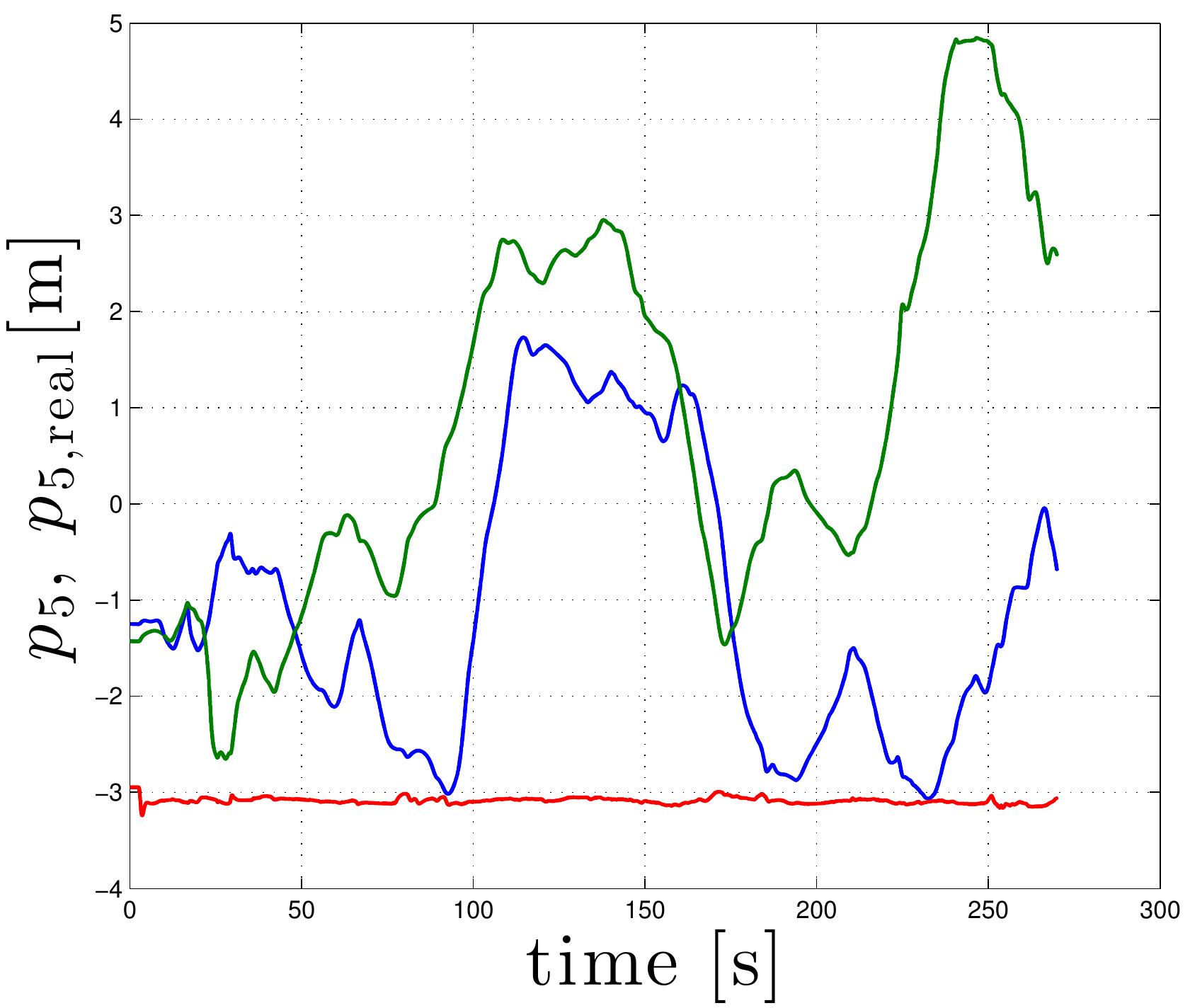}}
\subfigure[]{\includegraphics[width=0.45\columnwidth]{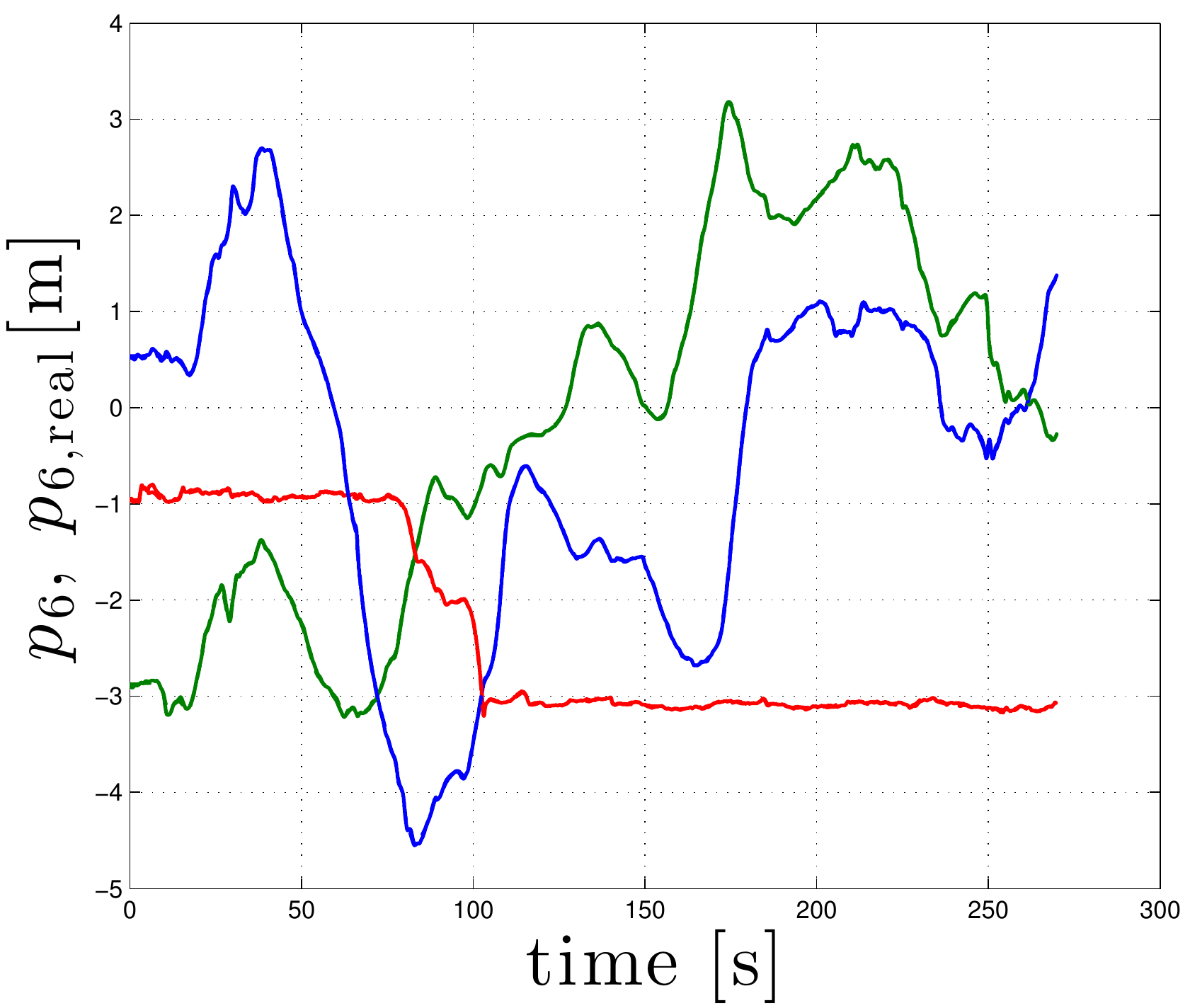}}
  \caption{Figs.~(a--f): behavior of $p_i(t)$ (solid) and $p_{i,real}$ (dashed): these are basically superimposed, showing the accuracy of the quadrotors in tracking the reference trajectory $p_i(t)$. {In the plots the following color code is used: blue/red/green solid/dashed lines correspond to the $x$/$y$/$z$ components of $p_i(t)$ and $p_{i,real}$}}\label{fig:QC_pos_exp2}
\end{center}
\end{figure}
\begin{figure}[!t]
\begin{center}
\includegraphics[width=0.45\columnwidth]{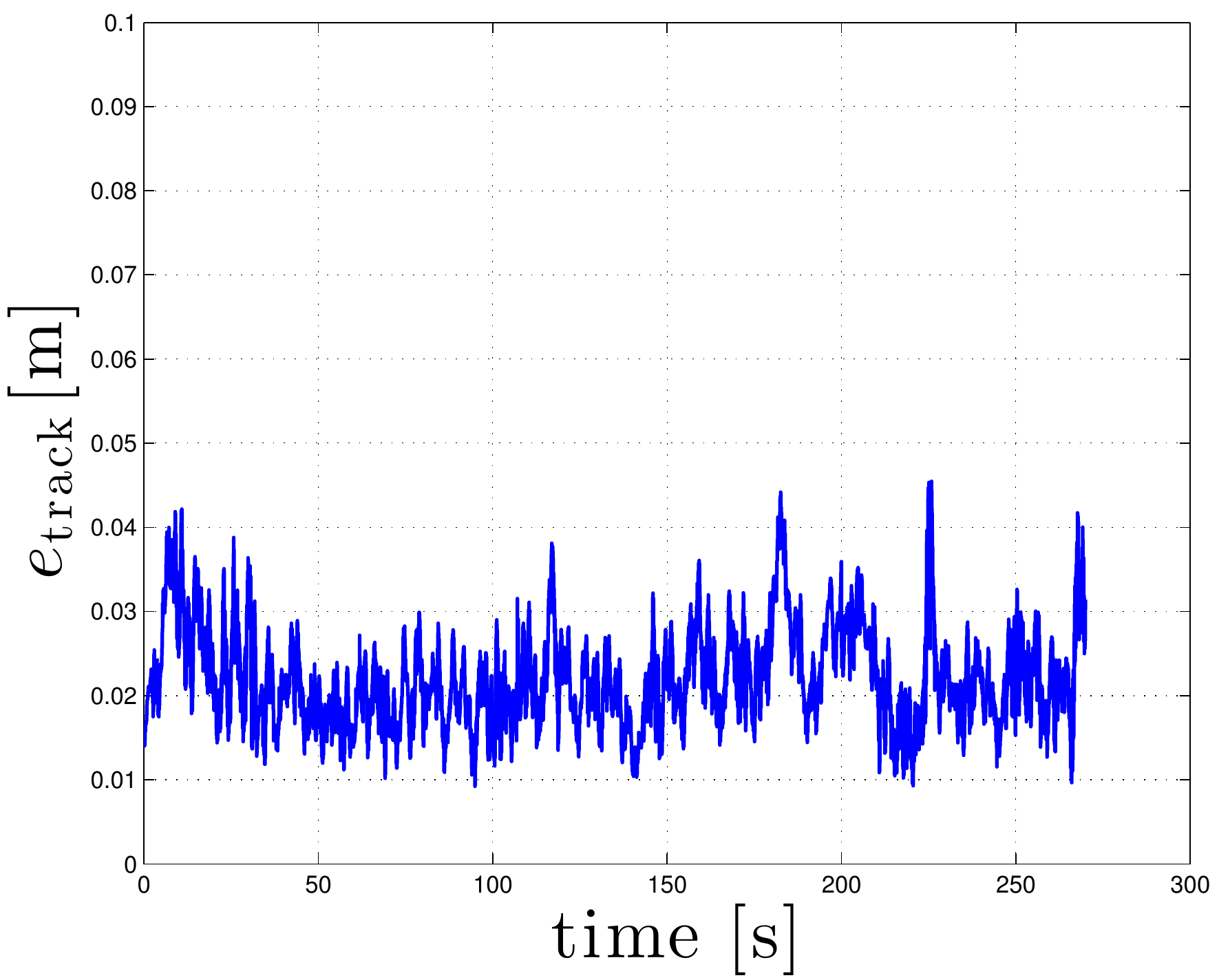}
  \caption{Behavior of the tracking error $e_{\mathrm{track}}(t)$ defined in~(\ref{eq:track_error}) showing again the good tracking performance of the $6$ quadrotors}\label{fig:e_track_err_exp2}
\end{center}
\end{figure}

\section{Concluding Remarks }\label{sec:conclusion}

This work presented a fully distributed solution for the rigidity maintenance control of a multi-robot system.  As discussed in the introduction, rigidity is an important architectural feature for multi-robot systems that enables, for example, formation keeping and localization using only {range-based measurements}.  The main theme of this work, therefore, was the distributed implementation of a number of algorithms for estimation and control in a multi-robot system related to rigidity maintenance.  In particular, we demonstrated how the \emph{rigidity eigenvalue} and eigenvector, used to decide if a formation is infinitesimally rigid, can be distributedly estimated using a suite of estimators based on dynamic consensus filters and the power iteration method for eigenvalue estimation.  The rigidity property also allowed for estimation of a common inertial reference frame using only \emph{range based measurements}, along with one single endowed agent that is able to sense both range and bearing.  The estimation of these quantities were then embedded in a gradient-based distributed control action ensuring each agent moves in a way that guarantees rigidity of the formation is maintained.  This control scheme also explicitly handles a variety of practical multi-robot constraints, including sensing and communication ranges, collision and obstacle avoidance, and line-of-sight requirements.   The validity of the proposed algorithms was demonstrated by a team of $6$ quadrotor UAVs flying in a cluttered environment.

This work also highlighted a number of directions for future research.  In particular, the estimation of the rigidity eigenvalue assumed that there is a separation between the rigidity eigenvalue and the next largest eigenvalue, i.e. $|\lambda_7  - \lambda_8| > 0$.  
While the reported experimental results showed a large degree of robustness w.r.t.~this effect, there remain both theoretical and practical questions related to this problem. For instance, it would be interesting to complement the rigidity maintenance controller with an additional term meant to maintain a minimum separation among $\lambda_8$ and $\lambda_7$. 
Another extension is to relax the requirement for having a special agent endowed with additional sensing capabilities (i.e. range and bearing).  This would lead to a distributed solution involving only range measurements for all robots in the ensemble. 

Despite these remaining challenges, this work has successfully demonstrated the power of distributed strategies for multi-robot systems.  Indeed, it is remarkable to observe the behavior of the multi-robot team running many distributed filters to achieve a common global objective.  The refinement of these strategies will no doubt become an important requirement as autonomous multi-robot systems are integrated more into a variety of application domains.

{
\section*{Appendix A: Index to Multimedia Extensions}\label{sec:multimedia}

The multimedia extensions to this article are at:
\url{http://www.ijrr.org.}

\begin{center} 
\begin{tabular}{l|l|p{5cm}}
\hline
 Extension  & Type & Description \\\hline
  1 & Video & { Experiments of rigidity maintenance with a group of UAVs}\\
\hline
\end{tabular}
\end{center}
}

{
\section*{Appendix B:  Rigidity Matrix Example}\label{apdx:rigidity}
The development of the alternative representation of the Rigidity Matrix given in Proposition \ref{prop:alt_rigidity} of the document is aided by a simple example.  To begin, we make some qualitative observations of the rigidity matrix.  For this example we consider a framework in $\reals^2$ with the complete graph on 3 nodes (denoted $K_3$).  The rigidity matrix can be written by inspection as
\beas R(p) &=& \\
&&\hspace{-50pt} \leftm{cccccc}\hspace{-5pt} p^x_1-p^x_2 &\hspace{-5pt} p^y_1-p^y_2 &\hspace{-5pt} p^x_2-p^x_1 &\hspace{-5pt} p^y_2-p^y_1 &\hspace{-5pt} 0 & 0 \hspace{-5pt}\\
				      \hspace{-5pt}p^x_1-p^x_3 &\hspace{-5pt} p^y_1-p^y_3 &\hspace{-5pt} 0 &\hspace{-5pt} 0 &\hspace{-5pt} p^x_3-p^x_1 &\hspace{-5pt} p^y_3-p^y_1 \hspace{-5pt}\\
				   \hspace{-5pt}   0 &\hspace{-5pt} 0 &\hspace{-5pt}\hspace{-5pt} p^x_2-p^x_3 & p^y_2 - p^y_3 &\hspace{-5pt} p^x_3-p^x_2 &\hspace{-5pt} p^y_3-p^y_2\hspace{-5pt}
				    \rightm .
\eeas

For the complete graph and an arbitrary orientation assigned to each edge, the incidence matrix $E(\mc{G})$ can be written as
$$E(\mc{G}) = \leftm{ccc} 1 & 1 & 0 \\ -1 & 0 & 1 \\ 0 & -1 & -1 \rightm.$$

The transpose of the incidence matrix functions as a ``difference" operator.  If the position of each agent is formed into a vector, we have
$$ E(\mc{G})^T \leftm{cc} p^x_1 & p^y_1 \\ p^x_2 & p^y_2 \\ p^x_3 & p^y_3 \rightm = \leftm{cc} p^x_1-p^x_2 & p^y_1 - p^y_2 \\ p^x_1-p^x_3 & p^y_1-p^y_3 \\ p^x_2-p^x_3 & p^y_2-p^y_3 \rightm.$$
The point to illustrate here is that this difference operation between positions is \emph{redundantly embedded inside the rigidity matrix}.  This fact can be made more precise by defining a \emph{directed local graph at node $v_i$} from the graph $\mc{G}$ as in Definition \ref{def:localgraph} in the main text.  Intuitively, the idea is that each node only has some local information about the connectivity of the entire graph; indeed, it only knows of the existence of other nodes that it can sense.  In this way, we can define a sub-graph induced by each node in the graph as follows.

Let $\mc{G}_j = (\mc{V},\mc{E}_j)$ be a sub-graph induced by node $v_j$ such that 
$$\mc{E}_j = \{(v_j,v_i) \, | \, \{v_i,v_j\} \in \mc{E}\}.$$
Here we emphasize that the original graph $\mc{G}$ is undirected, while in the new induced graph $\mc{G}_i$ we assign a direction to the edge such that node $v_j$ is always the tail.  Furthermore, observe that $\cup_j \mc{G}_j = \mc{G}$.\footnote{Here, we assume that the directed edges $(v_i,v_j)$ and $(v_j,v_i)$ are equivalent to the undirected edge $\{v_i,v_j\}$.} This is illustrated in Figure \ref{fig:graphex}.  
\begin{figure}[!t]
\begin{center}
	\subfigure[A graph.]{\scalebox{.49}{\includegraphics{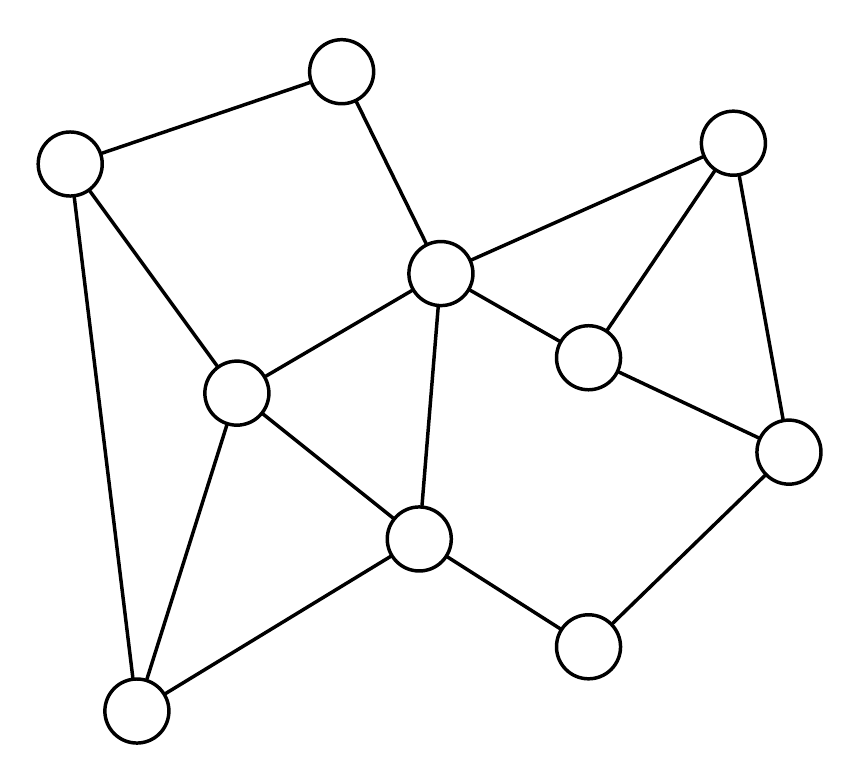}}\label{fig:fullgraph}}
	\subfigure[Local directed graph at a node.]{\scalebox{.49}{\includegraphics{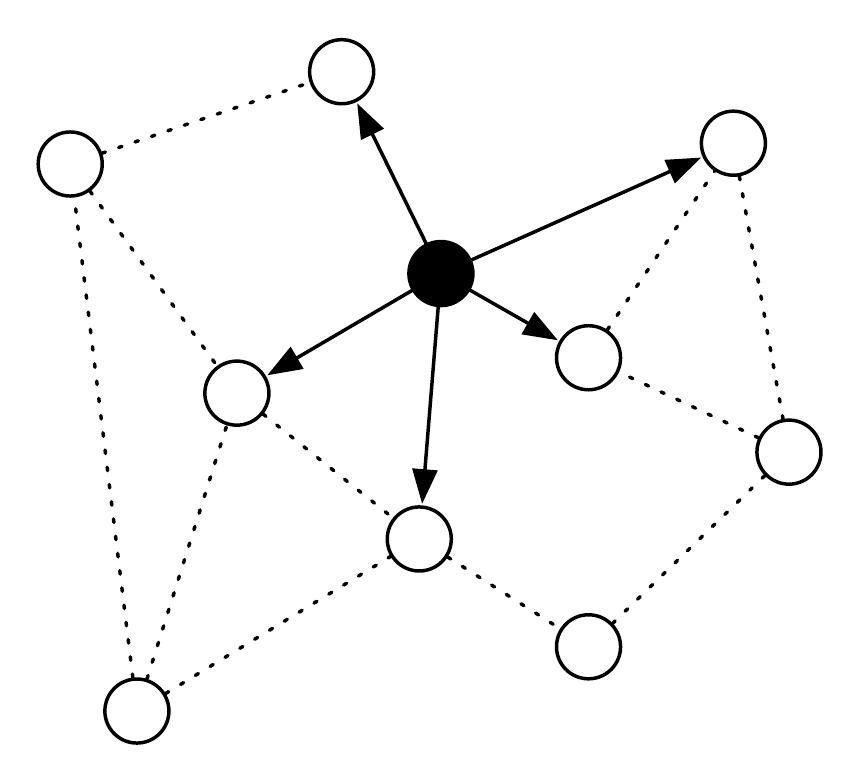}}\label{fig:localgraph}}
  \caption{Example of a directed local graph.} \label{fig:graphex}
\end{center}
\vspace{-20pt}
\end{figure}

To continue with the $K_3$ example, we can write the local incidence matrix for node $v_1$ as
$$ E_l(\mc{G}_1) = \leftm{ccc} 1 & 1 & 0 \\ -1 & 0 & 0 \\ 0 & -1 & 0 \rightm.$$
Note that this matrix is not truly an incidence matrix for the graph $\mc{G}_1$; ``placeholders" for the other edges in the graph $\mc{G}$ are kept.  As a result, the local incidence matrix is defined as $E_l(\mc{G}_j) \in \reals^{|\mc{V}| \times |\mc{E}|}$ to have zero-columns corresponding to the edges not in $\mc{E}_j$.\footnote{This representation also assumes that all the edges have been assigned a label, and this labeling is maintained even for the local graphs (local graphs do not relabel their edges; for example if edge 2 is not in local graph $\mc{G}_j$, then the second column of $E(\mc{G}_j)$ will be zero ).}

Now, consider the local incidence matrix as the difference operator,
$$E_l(\mc{G}_1)^T \leftm{cc} p^x_1 & p^y_1 \\ p^x_2 & p^y_2 \\ p^x_3 & p^y_3 \rightm = \leftm{cc} p^x_1-p^x_2 & p^y_1 - p^y_2 \\ p^x_1-p^x_3 & p^y_1-p^y_3 \\ 0 & 0 \rightm.$$
Note that this is identical to the the first 2 columns of the rigidity matrix $R(p)$.  In fact, this shows that the rigidity matrix can be written entirely in terms of local incidence matrices, as formally stated in Proposition \ref{prop:alt_rigidity} of the main document.

}

\section*{ACKNOWLEDGMENTS}

Part of Heinrich B\"ulthoff's research was supported  by the Brain Korea 21 PLUS Program through the National Research Foundation of Korea funded by the Ministry of Education. Correspondence should be directed to Heinrich~H.~B\"ulthoff.

{\small
\bibliographystyle{plainnat}
\bibliography{./bibCustomAlias,./bibCustom}

\end{document}